\documentclass{article}
\usepackage{latexsym}

\newtheorem{definition}{Definition}
\newtheorem{theorem}{Theorem}

\newtheorem{lemma}{Lemma}
\newtheorem{fact}{Fact}
\newtheorem{assumption}{Assumption}
\newtheorem{proof}{Proof}

\begin{document}
\title{The Characterization of \\ Abstract Truth and its Factorization}
\author{Robert E. Kent}
\maketitle
\tableofcontents

\begin{abstract}
Human knowledge is made up of the conceptual structures of many communities of interest.
In order to establish coherence in human knowledge representation,
it is important to enable communcation between the conceptual structures of different communities
The conceptual structures of any particular community is representable in an ontology.
Such a ontology provides a formal linguistic standard for that community.
However,
a standard community ontology is established for various purposes,
and makes choices that force a given interpretation, 
while excluding others that may be equally valid for other purposes.
Hence, a given representation is relative to the purpose for that representation.
Due to this relativity of represntation,
in the larger scope of all human knowledge
it is more important to standardize methods and frameworks for relating ontologies
than to standardize any particular choice of ontology.
The standardization of methods and frameworks is called the semantic integration of ontologies.
\end{abstract}

This paper,
whose orientation is given in the previous paper ``Conceptual: Institutions in a Topos'',
makes two advances:
it develops the notion of the lattice of theories (LOT) facterization,
and
it offers an order-enriched axiomatization that characterizes conceptual structures and the lattice of theories.

\section{Preliminaries\label{sec:preliminaries}}

\subsection{Elements of Order\label{subsec:elements:of:order}}

A \emph{preorder} 
$\mathbf{A} = \langle A, \leq_{\mathbf{A}} \rangle$
consists of a set $A$ and a binary relation $\leq_{\mathbf{A}} \subseteq A \times A$
that is reflexive and transitive:
$a \leq_{\mathbf{A}} a$ for all elements $a \in A$,
and
$a_1 \leq_{\mathbf{A}} a_2$ and $a_2 \leq_{\mathbf{A}} a_3$ implies $a_1 \leq_{\mathbf{A}} a_3$
for all triples of elements $a_1, a_2, a_3 \in A$.
Every preorder $\mathbf{A}$ has an associated equivalence relation 
$\equiv_{\mathbf{A}}$ on $A$ defined by
$a_1 \equiv_{\mathbf{A}} a_2$ 
when $a_1 \leq_{\mathbf{A}} a_2$ and $a_2 \leq_{\mathbf{A}} a_1$
for all pairs of elements $a_1, a_2 \in A$.
A partially ordered set (\emph{poset}) is a preorder that is antisymmetric:
$a_1 \equiv_{\mathbf{A}} a_2$ implies $a_1 = a_2$
for all pairs of elements $a_1, a_2 \in A$.
Any preorder 
$\mathbf{A}$
has an associated \emph{quotient} poset
$\mathsf{quo}(\mathbf{A}) = [\mathbf{A}] 
= \langle A/{\equiv_{\mathbf{A}}}, \leq_{[\mathbf{A}]} \rangle$,
where
$[a_1] \leq_{[\mathbf{A}]} [a_1]$ when
$a_1 \equiv_{\mathbf{A}} a_2$.

A \emph{monotonic (isotonic) function} 
$\mathbf{f} : \mathbf{A} \rightarrow \mathbf{B}$
from preorder $\mathbf{A}$ to preorder $\mathbf{B}$
is a function $f : A \rightarrow  B$ that preserves (preserves and respects) order:
$a_1 \leq_{\mathbf{A}} a_2$ implies (iff) $f(a_1) \leq_{\mathbf{A}} f(a_2)$
for all pairs of source elements $a_1, a_2 \in A$.
The composition and identities of monotonic functions
can be defined in terms of the underlying sets and functions.
Let $\mathsf{Ord}$ 
denote the category of preorders and monotonic functions\footnote{Also denoted $\mathsf{Pre}$}.
There is an underlying functor
$|\mbox{-}| : \mathsf{Ord} \rightarrow \mathsf{Set}$
that gives the underlying set of a preorder and the underlying function of a monotonic function.
For every preorder $\mathbf{A}$,
there is a surjective \emph{canon}(ical) isotonic function
$[\mbox{-}]_{\mathbf{A}} : \mathbf{A} \rightarrow \mathsf{quo}(\mathbf{A})$
where $[a]_{\mathbf{A}}$ is the equivalence class of $a \in A$.
For every monotonic function $\mathbf{f} : \mathbf{A} \rightarrow \mathbf{B}$,
there is a monotonic function 
$\mathsf{quo}(\mathbf{f}) = [\mathbf{f}] 
: \mathsf{quo}(\mathbf{A}) \rightarrow \mathsf{quo}(\mathbf{B})$
that satisfies the naturality condition
$\mathbf{f} \cdot [\mbox{-}]_{\mathbf{B}} = [\mbox{-}]_{\mathbf{A}} \cdot [\mathbf{f}]$.

Let $\mathsf{Ord}_{=} \subset \mathsf{Ord}$ 
denote the full subcategory of posets and monotonic functions\footnote{Also denoted $\mathsf{Pos}$}.
There is an inclusion functor
$\mathsf{incl} : \mathsf{Ord}_{=} \rightarrow \mathsf{Ord}$
and a quotient functor
$\mathsf{quo} : \mathsf{Ord} \rightarrow \mathsf{Ord}_{=}$.
There is a canon(ical) natural transformation
$\eta 
: \mathsf{id}_{\mathsf{Ord}} \Rightarrow \mathsf{quo} \circ \mathsf{incl}
: \mathsf{Ord} \rightarrow \mathsf{Ord}$
whose $\mathbf{A}^{\mathrm{th}}$-component is the surjective canonical isometry
$[\mbox{-}]_{\mathbf{A}} : \mathsf{A} \rightarrow \mathsf{quo}(\mathbf{A})$.
The quotient functor is left adjoint to the inclusion functor 
$\mathsf{quo} \dashv \mathsf{incl}$
with counit being an isomorphism and unit being the canon.
This adjunction is a  reflection: 
$\mathsf{Ord}_{=}$ is a reflective subcategory of $\mathsf{Ord}$ with the quotient functor being the reflector (Figure~\ref{order-fibration}).

\begin{figure}
\begin{center}
\setlength{\unitlength}{0.55pt}
\begin{picture}(120,120)(-25,10)
\put(-30,75){\makebox(60,30){$\mathsf{Ord}$}}
\put(90,75){\makebox(60,30){$\mathsf{Ord}_{=}$}}
\put(30,-15){\makebox(60,30){$\mathsf{Set}$}}
\put(45,100){\makebox(30,20){\footnotesize{$\mathsf{quo}$}}}
\put(28,80){\makebox(30,20){\scriptsize{$[\mbox{-}]$}}}
\put(45,80){\makebox(30,20){\footnotesize{$\dashv$}}}
\put(67,80){\makebox(30,20){\scriptsize{$\cong$}}}
\put(45,58){\makebox(30,20){\footnotesize{$\mathsf{incl}$}}}
\put(-3,35){\makebox(30,20){\footnotesize{$|\mbox{-}|$}}}
\put(93,35){\makebox(30,20){\footnotesize{$|\mbox{-}|$}}}
\put(30,102){\vector(1,0){60}}
\put(90,78){\vector(-1,0){60}}
\put(10,75){\vector(2,-3){40}}
\put(110,75){\vector(-2,-3){40}}
\end{picture}
\end{center}
\caption{Order Fibration}
\label{order-fibration}
\end{figure}
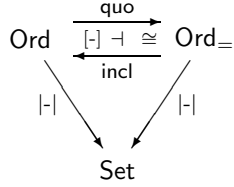

Given two preorders $\mathbf{A}_1$ and $\mathbf{A}_2$,
the \emph{binary product} is the preorder 
$\mathbf{A}_1 {\times} \mathbf{A}_2
= \langle A_1 {\times} A_2, \leq \rangle$
where $(a_1,a_2) \leq (a_1^\prime,a_2^\prime)$
when $a_1 \leq_{\mathbf{A}_1} a_1^\prime$ and $a_2 \leq_{\mathbf{A}_2} a_2^\prime$.
There are two component projection monotonic functions
$\pi_1 : \mathbf{A}_1 {\times} \mathbf{A}_2 \rightarrow \mathbf{A}_1$
and $\pi_2 : \mathbf{A}_1 {\times} \mathbf{A}_2 \rightarrow \mathbf{A}_2$.
This is a categorical product,
since given any pair of monotonic functions
$\mathbf{f}_1 : \mathbf{C} \rightarrow \mathbf{A}_1$,
and $\mathbf{f}_2 : \mathbf{C} \rightarrow \mathbf{A}_2$ 
with common source,
there is a unique monotonic function
$\mathbf{f} = (\mathbf{f}_1,\mathbf{f}_2) : \mathbf{C} \rightarrow \mathbf{A}_1 {\times} \mathbf{A}_2$
such that
$\mathbf{f} \cdot \pi_1 = \mathbf{f}_1$
and
$\mathbf{f} \cdot \pi_2 = \mathbf{f}_2$.
This definition can be extended to any finite number of preorders.
The finite product of posets is a poset.
Any one element set forms a poset $\mathbf{1}$ that is the nullary product,
since for any preorder $\mathbf{A}$ there is a unique (constant) monotonic function
$\mathbf{!} : \mathbf{A} \rightarrow \mathbf{1}$.
Given any parallel pair of monotonic functions
$\mathbf{f},\mathbf{g} : \mathbf{A} \rightarrow \mathbf{B}$,
the \emph{equalizer} is the subpreorder
$\mathbf{E} = \langle E, \leq \rangle$
with the $\mathsf{Set}$-equalizer $E = \{ a \in A \mid f(a) = g(a) \} \subseteq A$ 
and the induced subset order.
The inclusion monotonic function 
$\mathrm{incl} : \mathbf{E} \rightarrow \mathbf{A}$
and the composite
$\mathrm{incl} \cdot \mathbf{f} = \mathrm{incl} \cdot \mathbf{g} 
: \mathbf{E} \rightarrow \mathbf{B}$
form the limiting cone.
Hence,
the categories $\mathsf{Ord}$ and $\mathsf{Ord}_{=}$ are finite complete,
and the underlying functors preserve these limits. 

\subsection{Fibrations\label{subsec:fibrations}}

In this section consider a functor
\[\mathsf{P} : \mathsf{E} \rightarrow \mathsf{B}.\]
A \emph{fiber pair} $(b, E_1)$
consists of a $\mathsf{B}$-morphism $b : B_2 \rightarrow B_1$ 
and an $\mathsf{E}$-object $E_1 \in \mathsf{P}^{-1}B_1$ over $B_1$.
An $\mathsf{E}$-morphism $e : E_2 \rightarrow E_1$ is \emph{cartesian} for $(b, E_1)$
when
(1) $\mathsf{P}(e) = b$ and 
(2) for any $\mathsf{E}$-morphism $v : E \rightarrow E_1$
and any $\mathsf{B}$-morphism $x : \mathsf{P}(E) \rightarrow B_2$
for which $x \cdot_{\mathsf{B}} b = \mathsf{P}(v)$,
there is a unique $\mathsf{E}$-morphism $u : E \rightarrow E_2$
such that $u \cdot_{\mathsf{E}} e = v$ and $\mathsf{P}(u) = x$. 
Any two cartesian $\mathsf{E}$-morphisms for $(b, E_1)$
are isomorphic;
that is,
if $e : E_2 \rightarrow E_1$ and $e^\prime : E_2^\prime \rightarrow E_1$
are cartesian for $(b, E_1)$,
then $E_2 \cong E_2^\prime$
and there are inverse $\mathsf{E}$-morphisms in the $B_2$-fiber,
$u^\prime : E_2 \rightarrow E_2^\prime$ and $u : E_2^\prime \rightarrow E_2$,
with
$u,u^\prime \in \mathsf{P}^{-1}B_2$,
$u^\prime \cdot_{\mathsf{E}} u = 1_{E_2}$,
$u \cdot_{\mathsf{E}} u^\prime = 1_{E_2^\prime}$,
$u^\prime \cdot_{\mathsf{E}} e^\prime = e$
and $u \cdot_{\mathsf{E}} e = e^\prime$.
Identities are cartesian;
that is,
any $\mathsf{E}$-identity $1_E : E \rightarrow E$ is cartesian for $(1_{\mathsf{P}(E)}, E)$.
The composition of cartesian morphisms is cartesian;
that is,
if $e_2 : E_3 \rightarrow E_2$ is cartesian for $(b_2, E_2)$ 
and $e_1 : E_2 \rightarrow E_1$ is cartesian for $(b_1, E_1)$,
then $e_2 \cdot_{\mathsf{E}} e_1 : E_3 \rightarrow E_1$ is cartesian for $(b_2 \cdot_{\mathsf{B}} b_1, E_1)$.
The functor
$\mathsf{P} : \mathsf{E} \rightarrow \mathsf{B}$
is a \emph{fibration} when there is a cartesian $\mathsf{E}$-morphism
for any fiber pair $(b, E_1)$.
Then we say that 
$\mathsf{E}$ is fibered over $\mathsf{B}$,
with $\mathsf{B}$ the base category and $\mathsf{E}$ the total category of the fibration $\mathsf{P}$.

When the cartesian morphisms for a fibration 
$\mathsf{P} : \mathsf{E} \rightarrow \mathsf{B}$
are chosen,
$\mathsf{P}$ is said to have a \emph{cleavage}.
Hence,
a cleavage for $\mathsf{P}$ maps each fiber pair $(b, E_1)$
to a cartesian $\mathsf{E}$-morphism 
$\gamma(b, E_1) : E_2 \rightarrow E_1$\footnote{A more compact notation for the cleavage of a fiber pair $(b, E_1)$ is $b^\ast_{E_1} : b^\ast(E_1) \rightarrow E_1$.}.
The cleavage $\gamma$ is a \emph{splitting} of the fibration $\mathsf{P}$
when
it satisfies the following two conditions:
(1) for any $\mathsf{B}$-object $B$
with $E \in \mathsf{P}^{-1}B$,
the cleavage of $E$ along the $\mathsf{B}$-identity $1_B$ is the $\mathsf{E}$-identity 
$\gamma(1_B, E) = 1_E \in \mathsf{P}^{-1}B$;
and
(2) for any $\mathsf{B}$-composable pair  
$b_2 : B_3 \rightarrow B_2$ and $b_1 : B_2 \rightarrow B_1$
with $\mathsf{E}$-object $E_1 \in \mathsf{P}^{-1}B_1$,
if $E_2 \in \mathsf{P}^{-1}B_2$ is the source of $\gamma(b_1, E_1)$
then the cleavage of $E_1$ along the $\mathsf{B}$-composite $b_2 \cdot_{\mathsf{B}} b_1$ 
is the $\mathsf{E}$-composite
$\gamma(b_2 \cdot_{\mathsf{B}} b_1, E_1) 
= \gamma(b_2, E_2) \cdot_{\mathsf{E}} \gamma(b_1, E_1)$.
A fibration $\mathsf{P}$ is said to be \emph{split} when it has a splitting.

Some special concepts and notation are useful (Figure~\ref{special-notation-concepts}).
For any $\mathsf{E}$-morphism $e : E_2 \rightarrow E_1$,
the \emph{lift} of $e$ is the cleavage
$\sharp_{e} \doteq \gamma(\mathsf{P}(e),E_1) : \Delta(e) \rightarrow E_1$
for the fiber pair $(\mathsf{P}(e),E_1)$,
the \emph{apex} of $e$ is the cleavage source $\Delta(e)$,
and the \emph{gap} of $e$ is the unique vertical $\mathsf{E}$-morphism
$\flat_{e} : E_2 \rightarrow \Delta(e)$
such that $e = \flat_{e} \cdot_{\mathsf{E}} \sharp_{e}$.
We can think of the gap as the embedding of the source $E_2$ into the apex $\Delta(e)$.
Thus,
any $\mathsf{E}$-morphism $e : E_2 \rightarrow E_1$
factors as 
$e = \flat_{e} \cdot_{\mathsf{E}} \sharp_{e} 
: E_2 \rightarrow \Delta(e) \rightarrow E_1$
where $\flat_{e}$ is a vertical $\mathsf{E}$-morphism 
and $\sharp_{e}$ is a cartesian $\mathsf{E}$-morphism.
This is called the \emph{cleavage factorization} of $e$.
The cleavage factorization of the gap $\flat_{e}$ has
apex $\Delta(\flat_{e}) = \Delta(e)$,
gap $\flat_{\flat_{e}} = \flat_{e} : E_2 \rightarrow \Delta(e)$ and
lift $\sharp_{\flat_{e}} = 1_{\Delta(e)} : \Delta(e) \rightarrow \Delta(e)$.
The cleavage factorization of the lift $\sharp_{e}$ has
apex $\Delta(\sharp_{e}) = \Delta(e)$,
gap $\flat_{\sharp_{e}} = 1_{\Delta(e)} : \Delta(e) \rightarrow \Delta(e)$ and
lift $\sharp_{\sharp_{e}} = \sharp_{e} : \Delta(e) \rightarrow E_1$.
For any $\mathsf{E}$-object $E$,
the cleavage factorization of the identity $1_E$ has
apex $\Delta(1_E) = E$, and
gap and lift $\flat_{1_E} = \sharp_{1_E} = 1_E : E \rightarrow E$.
For any composable pair of $\mathsf{E}$-morphisms
$e_2 : E_3 \rightarrow E_2$ and
$e_1 : E_2 \rightarrow E_1$,
the \emph{diagonal} of $(e_1,e_2)$ is the unique $\mathsf{E}$-morphism
$\delta_{e_1,e_2} : E_3 \rightarrow \Delta(e_1)$
such that
$\delta_{e_1,e_2} \cdot_{\mathsf{E}} \sharp_{e_1} = e_2 \cdot_{\mathsf{E}} e_1$
and $\mathsf{P}(\delta_{e_1,e_2}) = \mathsf{P}(e_2)$,
Clearly,
$\delta_{e_1,e_2} = e_2 \cdot_{\mathsf{E}} \flat_{e_1}$.
When $\mathsf{P}$ is split,
$\delta_{e_1,e_2} 
= \flat_{e_2 {\cdot} \flat_{e_1}} \cdot_{\mathsf{E}} \sharp_{e_2 {\cdot} \flat_{e_1}}
= \flat_{e_2 {\cdot} e_1} \cdot_{\mathsf{E}} \sharp_{e_2 {\cdot} \flat_{e_1}}$,
$\flat_{e_2 {\cdot}  e_1} = \flat_{e_2 {\cdot} \flat_{e_1}}$ and
$\sharp_{e_2 {\cdot} e_1} = \sharp_{e_2 {\cdot} \flat_{e_1}} \cdot_{\mathsf{E}} \sharp_{e_1}$.
If an $\mathsf{E}$-morphism $e : E_3 \rightarrow E_1$
factors as $e = e_2 \cdot_{\mathsf{E}} e_1 : E_3 \rightarrow E_2 \rightarrow E_1$
where $e_1$ is cartesian,
then $\Delta(e) \cong \Delta(e_2)$.

\begin{figure}
\begin{center}
\begin{tabular}{c@{\hspace{30pt}}c@{\hspace{50pt}}c}
\setlength{\unitlength}{1.0pt}
\begin{picture}(60,95)(0,0)
\put(0,30){\begin{picture}(120,60)(0,0)
\put(-30,-15){\makebox(60,30){$E_2$}}
\put(30,-15){\makebox(60,30){$E_1$}}
\put(-30,30){\makebox(60,30){$\Delta(e)$}}
\put(15,-10){\makebox(30,30){\footnotesize{$e$}}}
\put(-22,7.5){\makebox(30,30){\footnotesize{$\flat_{e}$}}}
\put(22,13){\makebox(30,30){\footnotesize{$\sharp_{e}$}}}
\put(15,0){\vector(1,0){30}}
\put(0,10){\vector(0,1){25}}
\put(12,36){\vector(4,-3){36}}
\end{picture}}
\qbezier[50](-15,15)(30,15)(75,15)
\put(-30,-15){\makebox(60,30){$B_2$}}
\put(30,-15){\makebox(60,30){$B_1$}}
\put(15,-10){\makebox(30,30){\footnotesize{$b$}}}
\put(15,0){\vector(1,0){30}}
\end{picture}
&
\setlength{\unitlength}{1.0pt}
\begin{picture}(60,95)(0,0)
\put(0,30){\begin{picture}(120,60)(0,0)
\put(-30,-15){\makebox(60,30){$E$}}
\put(30,-15){\makebox(60,30){$E$}}
\put(-30,30){\makebox(60,30){$E = \Delta(1_E)$}}
\put(15,-10){\makebox(30,30){\footnotesize{$1_E$}}}
\put(-13,7.5){\makebox(30,30){\footnotesize{$1_E = \flat_{1_E}$}}}
\put(10,13){\makebox(30,30){\footnotesize{$1_E = \sharp_{1_E}$}}}
\put(15,0){\vector(1,0){30}}
\put(0,10){\vector(0,1){25}}
\put(12,36){\vector(4,-3){36}}
\end{picture}}
\qbezier[50](-15,15)(30,15)(75,15)
\put(-30,-15){\makebox(60,30){$B$}}
\put(30,-15){\makebox(60,30){$B$}}
\put(15,-10){\makebox(30,30){\footnotesize{$1_B$}}}
\put(15,0){\vector(1,0){30}}
\end{picture}
&
\setlength{\unitlength}{1.0pt}
\begin{picture}(120,85)(0,0)
\put(0,30){\begin{picture}(120,60)(0,0)
\put(-30,-15){\makebox(60,30){$E_3$}}
\put(30,-15){\makebox(60,30){$E_2$}}
\put(90,-15){\makebox(60,30){$E_1$}}
\put(-30,30){\makebox(60,30){$\Delta(e_2 {\cdot} e_1)$}}
\put(-34,20){\makebox(60,30){$= \Delta(e_2 {\cdot} \flat_{e_1})$}}
\put(30,30){\makebox(60,30){$\Delta(e_1)$}}
\put(15,-10){\makebox(30,30){\footnotesize{$e_2$}}}
\put(75,-10){\makebox(30,30){\footnotesize{$e_1$}}}
\put(21,7.5){\makebox(30,30){\footnotesize{$\delta_{e_1,e_2}$}}}
\put(-28,7.5){\makebox(30,30){\footnotesize{$\flat_{e_2 {\cdot} e_1}$}}}
\put(-31,-1){\makebox(30,30){\footnotesize{$= \flat_{e_2 {\cdot} \flat_{e_1}}$}}}
\put(55,7.5){\makebox(30,30){\footnotesize{$\flat_{e_1}$}}}
\put(20,37){\makebox(30,30){\footnotesize{$\sharp_{e_2 {\cdot} \flat_{e_1}}$}}}
\put(82,13){\makebox(30,30){\footnotesize{$\sharp_{e_1}$}}}
\put(55,55){\makebox(30,30){\footnotesize{$\sharp_{e_2{\cdot}e_1}$}}}
\put(24,45){\vector(1,0){20}}
\put(15,0){\vector(1,0){30}}
\put(12,9){\vector(4,3){36}}
\put(75,0){\vector(1,0){30}}
\put(0,10){\vector(0,1){18}}
\put(60,10){\vector(0,1){25}}
\put(72,36){\vector(4,-3){36}}
\qbezier(8,55)(70,90)(116,12)
\put(116,12){\vector(1,-2){0}}
\end{picture}}
\qbezier[80](-33,15)(50,15)(135,15)
\put(-30,-15){\makebox(60,30){$B_3$}}
\put(30,-15){\makebox(60,30){$B_2$}}
\put(90,-15){\makebox(60,30){$B_1$}}
\put(15,-10){\makebox(30,30){\footnotesize{$b_2$}}}
\put(75,-10){\makebox(30,30){\footnotesize{$b_1$}}}
\put(15,0){\vector(1,0){30}}
\put(75,0){\vector(1,0){30}}
\end{picture}
\end{tabular}
\end{center}
\caption{Special Notation and Concepts}
\label{special-notation-concepts}
\end{figure}
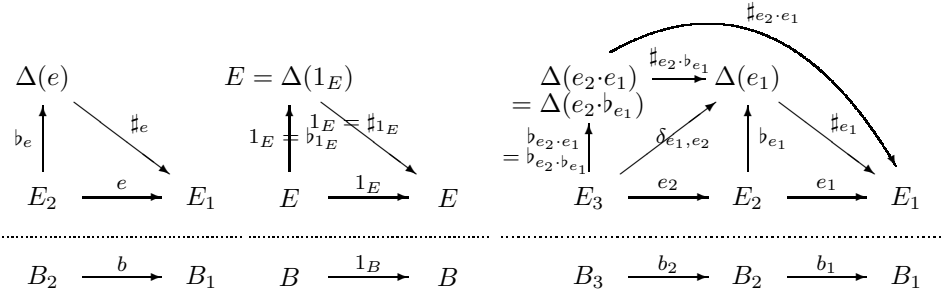

\subsection{Factorization Systems\label{subsec:factorization:systems}}

Let $\mathsf{C}$ be an arbitrary category.
An ordinary \emph{factorization system} in $\mathsf{C}$ is a pair 
$\langle \mathsf{E}, \mathsf{M} \rangle$ of classes of $\mathsf{C}$-morphisms satisfying the following conditions. 
{\bfseries Subcategories:}
All $\mathsf{C}$-isomorphisms are in $\mathsf{E} \,\cap\, \mathsf{M}$. 
Both $\mathsf{E}$ and $\mathsf{M}$ are closed under $\mathsf{C}$-composition.
Hence,
$\mathsf{E}$ and $\mathsf{M}$ are $\mathsf{C}$-subcategories with the same objects as $\mathsf{C}$.
{\bfseries Existence:} 
Every $\mathsf{C}$-morphism $f : A \rightarrow B$ has an $\langle \mathsf{E}, \mathsf{M} \rangle$-factorization\footnote{An $\langle \mathsf{E}, \mathsf{M} \rangle$-factorization is a quadruple $(A, e, C, m, B)$ where 
$e : A \rightarrow C$ and $m : C \rightarrow B$ 
is a composable pair of $\mathsf{C}$-morphisms 
with $e \in \mathsf{E}$ and $m \in \mathsf{M}$.};
that is,
there is an $\langle \mathsf{E}, \mathsf{M} \rangle$-factorization $(A, e, C, m, B)$ 
and $f$ is its composition\footnote{In this paper, all compositions are written in diagrammatic form.} $f = e \cdot m$.
{\bfseries Diagonalization:}
For every commutative square $e \cdot s = r \cdot m$ of $\mathsf{C}$-morphisms, with $e \in \mathsf{E}$ and $m \in \mathsf{M}$, 
there is a unique $\mathsf{C}$-morphism $d$ with $e \cdot d = r$ and $d \cdot m = s$.
This diagonalization condition implies the following condition.
{\bfseries Uniqueness:}
Any two $\langle \mathsf{E}, \mathsf{M} \rangle$-factorizations of a $\mathsf{C}$-morphism are isomorphic; that is, if $(A, e, C, m, B)$ and $(A, e^\prime, C^\prime, m^\prime, B)$ are two $\langle \mathsf{E}, \mathsf{M} \rangle$-factorizations of $f : A \rightarrow B$, then there is a unique $\mathsf{C}$-isomorphism $h : C \cong C^\prime$ with $e \cdot h = e^\prime$ and $h \cdot m^\prime = m$.
An \emph{epi-mono factorization system} is one 
where $\mathsf{E}$ is contained in the class of $\mathsf{C}$-epimorphisms
and $\mathsf{M}$ is contained in the class of $\mathsf{C}$-monomorphisms.
For an epi-mono factorization system,
the diagonalization condition is equivalent to the uniqueness condition.

Let $\mathsf{C}^{\mathsf{2}}$ denote the arrow category\footnote{Recall that $\mathsf{2}$ is the two-object category, pictured as $\bullet \rightarrow \bullet$, with one non-trivial morphism. The arrow category $\mathsf{C}^{\mathsf{2}}$ is (isomorphic to) the functor category $[\mathsf{2}, \mathsf{C}]$.} of $\mathsf{C}$.
An object of $\mathsf{C}^{\mathsf{2}}$ is a triple $(A, f, B)$,
where $f : A \rightarrow B$ is a $\mathsf{C}$-morphism.
A morphism of $\mathsf{C}^{\mathsf{2}}$,
$(a, b) : (A_1, f_1, B_1) \rightarrow (A_2, f_2, B_2)$,
is a pair of $\mathsf{C}$-morphisms $a : A_1 \rightarrow A_2$ and $b : B_1 \rightarrow B_2$ that form a commuting square $a \cdot f_2 = f_1 \cdot b$.
There are source and target projection functors
$\partial_0^{\mathsf{C}}, \partial_1^{\mathsf{C}} : \mathsf{C}^{\mathsf{2}} \rightarrow \mathsf{C}$
and an arrow natural transformation
$\alpha_{\mathsf{C}} : \partial_0^{\mathsf{C}} \Rightarrow \partial_1^{\mathsf{C}} : \mathsf{C}^{\mathsf{2}} \rightarrow \mathsf{C}$
with component
$\alpha_{\mathsf{C}}(A, f, B) = f : A \rightarrow B$
(background of Fig.~\ref{factorization-equivalence}).
Let $\mathsf{E}^{\mathsf{2}}$ denote the full subcategory of $\mathsf{C}^{\mathsf{2}}$
whose objects are the morphisms in $\mathsf{E}$.
Make the same definitions for $\mathsf{M}^{\mathsf{2}}$.
Just as for $\mathsf{C}^{\mathsf{2}}$,
the category $\mathsf{E}^{\mathsf{2}}$ has source and target projection functors
$\partial_0^\mathsf{E}, \partial_1^\mathsf{E} : \mathsf{E}^{\mathsf{2}} \rightarrow \mathsf{C}$
and an arrow natural transformation
$\alpha_{\mathsf{E}} : \partial_0^\mathsf{E} \Rightarrow \partial_1^\mathsf{E} : \mathsf{E}^{\mathsf{2}} \rightarrow \mathsf{C}$
(foreground of Fig.~\ref{factorization-equivalence}).
The same is true for $\mathsf{M}^{\mathsf{2}}$.
Let $\mathsf{E} \odot \mathsf{M}$ denote the category
of $\langle \mathsf{E}, \mathsf{M} \rangle$-factorizations
(top foreground of Fig.~\ref{factorization-equivalence}),
whose objects are $\langle \mathsf{E}, \mathsf{M} \rangle$-factorizations $(A, e, C, m, B)$,
and whose morphisms $(a, c, b) : (A_1, e_1, C_1, m_1, B_1) \rightarrow (A_2, e_2, C_2, m_2, B_2)$
are $\mathsf{C}$-morphism triples
where $(a, c) : (A_1, e_1, C_1) \rightarrow (A_2, e_2, C_2)$ 
is an $\mathsf{E}^{\mathsf{2}}$-morphism
and $(c, b) : (C_1, m_1, B_1) \rightarrow (C_2, m_2, B_2)$ 
is an $\mathsf{M}^{\mathsf{2}}$-morphism.
$\mathsf{E} \odot \mathsf{M}
= \mathsf{E}^{\mathsf{2}} \times_{\mathsf{C}} \mathsf{M}^{\mathsf{2}}$
is the pullback (in the category of categories)
of the $1^{\mathrm{st}}$-projection of $\mathsf{E}^{\mathsf{2}}$
and the $0^{\mathrm{th}}$-projection of $\mathsf{M}^{\mathsf{2}}$.
There is a composition functor
$\circ_{\mathsf{C}} : \mathsf{E} {\odot} \mathsf{M} \rightarrow \mathsf{C}^{\mathsf{2}}$
that commutes with projections:
on objects $\circ_{\mathsf{C}}(A, e, C, m, B) = (A, e \circ_{\mathsf{C}} m, B)$,
and on morphisms $\circ_{\mathsf{C}}(a, c, b) = (a, b)$.

An $\langle \mathsf{E}, \mathsf{M} \rangle$-factorization system with choice has a specified factorization for each $\mathsf{C}$-morphism;
that is, there is a choice function from the class of $\mathsf{C}$-morphisms to the class of $\langle \mathsf{E}, \mathsf{M} \rangle$-factorizations mapping each $\mathsf{C}$-morphism to one of its factorizations.
With this choice,
diagonalization is uniquely determined.
When choice is specified,
there is a factorization functor
$\div_{\mathsf{C}} : \mathsf{C}^{\mathsf{2}} \rightarrow \mathsf{E} {\odot} \mathsf{M}$,
which is defined on objects as the chosen $\langle \mathsf{E}, \mathsf{M} \rangle$-factorization $\div_{\mathsf{C}}(A, f, B) = (A, e, C, m, B)$
and on morphisms as
$\div_{\mathsf{C}}(a, b) = (a, c, b)$
where $c$ is defined by diagonalization
($\div_{\mathsf{C}}$ is functorial by uniqueness of diagonalization).
Clearly,
factorization followed by composition is the identity
$\div_{\mathsf{C}} \circ\, \circ_{\mathsf{C}} 
= \mathsf{id}_{\mathsf{C}}$.
By uniqueness of factorization (up to isomorphism)
composition followed by factorization is an isomorphism
$\circ_{\mathsf{C}} \,\circ \div_{\mathsf{C}} 
\cong \mathsf{id}_{\mathsf{E} {\odot} \mathsf{M}}$.

\begin{theorem} [General Equivalence] \label{general-equivalence}
When a category $\mathsf{C}$ has an $\langle \mathsf{E}, \mathsf{M} \rangle$-factorization system with choice,
the $\mathsf{C}$-arrow category
is equivalent (Fig.~\ref{factorization-equivalence}) to
the $\langle \mathsf{E}, \mathsf{M} \rangle$-factorization category
\[\mathsf{C}^{\mathsf{2}} \equiv \mathsf{E} {\odot} \mathsf{M}.\]
\end{theorem}
This equivalence is mediated by factorization and composition.

\begin{figure}
\begin{center}
\setlength{\unitlength}{0.78pt}
\begin{tabular}{c}
\begin{picture}(200,110)(-83,0)
\put(0,0){\begin{picture}(100,100)(0,0)
\put(5,75){\makebox(100,50){$\mathsf{E} {\odot} \mathsf{M}$}}
\put(-50,25){\makebox(100,50){$\mathsf{E}^{\mathsf{2}}$}}
\put(50,25){\makebox(100,50){$\mathsf{M}^{\mathsf{2}}$}}
\put(-100,-25){\makebox(100,50){$\mathsf{C}$}}
\put(0,-25){\makebox(100,50){$\mathsf{C}$}}
\put(100,-25){\makebox(100,50){$\mathsf{C}$}}
\put(-32,56){\makebox(100,50){\footnotesize{$\pi_\mathsf{E}$}}}
\put(33,56){\makebox(100,50){\footnotesize{$\pi_\mathsf{M}$}}}
\put(-80,8){\makebox(100,50){\footnotesize{$\partial_0^\mathsf{E}$}}}
\put(-23,12){\makebox(100,50){\footnotesize{$\partial_1^\mathsf{E}$}}}
\put(23,12){\makebox(100,50){\footnotesize{$\partial_0^\mathsf{M}$}}}
\put(82,8){\makebox(100,50){\footnotesize{$\partial_1^\mathsf{M}$}}}
\put(-50,-35){\makebox(100,50){\footnotesize{$\mathsf{id}$}}}
\put(50,-35){\makebox(100,50){\footnotesize{$\mathsf{id}$}}}
\put(-48,-4){\makebox(100,50){\footnotesize{$\alpha_{\mathsf{E}}$}}}
\put(-48,-14){\makebox(100,50){\Large{$\Rightarrow$}}}
\put(52,-4){\makebox(100,50){\footnotesize{$\alpha_{\mathsf{M}}$}}}
\put(52,-14){\makebox(100,50){\Large{$\Rightarrow$}}}
\put(50,25){\begin{picture}(30,15)(0,-15)
\put(0,0){\line(-1,-1){15}}
\put(0,0){\line(1,-1){15}}
\end{picture}}
\thicklines
\put(40,90){\vector(-1,-1){30}}
\put(60,90){\vector(1,-1){30}}
\put(-10,40){\vector(-1,-1){30}}
\put(10,40){\vector(1,-1){30}}
\put(-30,0){\vector(1,0){60}}
\put(90,40){\vector(-1,-1){30}}
\put(110,40){\vector(1,-1){30}}
\put(70,0){\vector(1,0){60}}
\end{picture}}
\thicklines
\put(-21,108){\vector(1,0){40}}
\put(-49,90){\makebox(100,50){\footnotesize{$\div_{\mathsf{C}}$}}}
\put(-52,75.5){\makebox(100,50){\footnotesize{$\equiv$}}}
\put(-49,62){\makebox(100,50){\footnotesize{$\circ_{\mathsf{C}}$}}}
\put(19,94){\vector(-1,0){40}}
\put(-35,90){\line(2,-1){44}}
\put(21,62){\line(2,-1){43}}
\put(79,33){\line(2,-1){15}}
\put(109,18){\vector(2,-1){28}}
\put(-50,85){\vector(0,-1){70}}
\put(-98,75){\makebox(100,50){$\mathsf{C}^{\mathsf{2}}$}}
\put(-109,24){\makebox(100,50){\footnotesize{$\partial_0^{\mathsf{C}}$}}}
\put(-7,36){\makebox(100,50){\footnotesize{$\partial_1^{\mathsf{C}}$}}}
\put(-75,45){\makebox(100,50){\footnotesize{$\alpha_{\mathsf{C}}$}}}
\put(-75,35){\makebox(100,50){\Large{$\Rightarrow$}}}
\end{picture}
\\ \\ \\
$\begin{array}{rcl}
\div_{\mathsf{C}} \circ\, \circ_{\mathsf{C}} & = & \mathsf{id}_{\mathsf{C}} 
\\
\circ_{\mathsf{C}} \,\circ \div_{\mathsf{C}} & \cong & \mathsf{id}_{\mathsf{E} {\odot} \mathsf{M}} 
\\
\div_{\mathsf{C}} 
\left(
\pi_\mathsf{E} \, \alpha_{\mathsf{E}}
\bullet
\pi_\mathsf{M} \, \alpha_{\mathsf{M}}
\right)
& = &
\alpha_{\mathsf{C}}
\end{array}$
\end{tabular}
\end{center}
\caption{Factorization Equivalence}
\label{factorization-equivalence}
\end{figure}
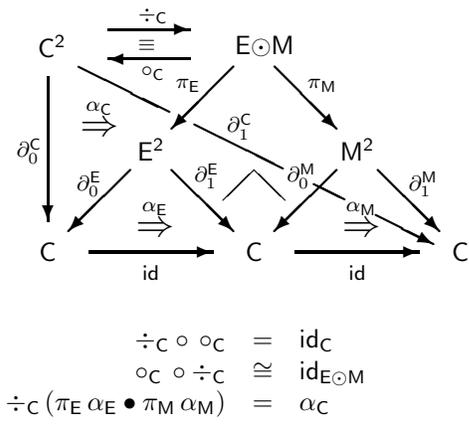


\section{Order-Enriched Categories\label{sec:order:enriched:categories}}

\subsection{Order-Enriched Categories\label{subsec:order:enriched:categories}}

An order-enriched category is a category whose ``hom-sets'' are ordered;
that is, whose hom-objects are in $\mathsf{Ord}$.
More precisely,
an \emph{order-enriched} category $\mathsf{C}$ 
consists of:
a set of objects $|\mathsf{C}| = \mathsf{obj}(\mathsf{C})$; 
for each pair of objects $A,B \in |\mathsf{C}|$,
a hom-preorder $\mathsf{C}(A,B) = \langle \mathsf{C}(A,B), \leq_{A,B} \rangle$;
for each triple of objects $A,B,C \in |\mathsf{C}|$,
a monotonic composition function
$\cdot_{A,B,C} : \mathsf{C}(A,B) \times \mathsf{C}(B,C) \rightarrow \mathsf{C}(A,C)$;
and
for each object $A \in |\mathsf{C}|$,
an monotonic identity element $1_A \in |\mathsf{C}(A,A)|$.
A $\mathsf{C}$-morphism $f : A \rightarrow B$ 
from $\mathsf{C}$-object $A$ to $\mathsf{C}$-object $B$
is an element in the underlying set $f \in |\mathsf{C}(A,B)|$.
This data is subject to the associative law (using infix notation for composition)
$f \cdot_{A,B,D} (g \cdot_{B,C,D} h) = (f \cdot_{A,B,C} g) \cdot_{A,C,D} h$
for all composable pairs $f : A \rightarrow B$, $g : B \rightarrow C$ and $h : C \rightarrow D$, 
and the identity laws on composition
$1_A \cdot_{A,A,B} f = f$ and $f \cdot_{A,B,B} 1_B = f$
for all morphisms $f : A \rightarrow B$. 
For each pair of objects $A,B \in |\mathsf{C}|$,
the order $\leq_{A,B}$ defines an equivalence relation $\equiv_{A,B}$
on the hom-set $|\mathsf{C}(A,B)|$.
Composition preserves equivalence:
if $f_1 \equiv f_2 : A \rightarrow B$ and $g_1 \equiv g_2 : B \rightarrow C$,
then $f_1 \cdot g_1 \equiv f_2 \cdot g_2 : A \rightarrow C$. 
The preorder of $\mathsf{C}$-morphisms is the disjoint union
$\mathsf{mor}(\mathsf{C}) = \coprod_{A,B \in |\mathsf{C}|} |\mathsf{C}(A,B)|$
with the induced order.
The \emph{opposite} order-enriched category of $\mathsf{C}$ 
is another order-enriched category $\mathsf{C}^{\mathrm{op}}$
consisting of:
the same class of objects $|\mathsf{C}^{\mathrm{op}}| = |\mathsf{C}|$; 
for each pair of objects $B,A \in |\mathsf{C}|$,
the hom-preorder 
$\mathsf{C}^{\mathrm{op}}(B,A) 
= \mathsf{C}(A,B)^{\mathrm{op}}
= \langle \mathsf{C}(A,B), \geq_{A,B} \rangle$
so that
a $\mathsf{C}^{\mathrm{op}}$-morphism $f : B \rightarrow A$
is a $\mathsf{C}$-morphism $f : A \rightarrow B$;
for each triple of objects $A,B,C \in |\mathsf{C}|$,
the monotonic composition function
$g \cdot^{\mathrm{op}}_{C,B,A} f = f \cdot_{A,B,C} g$
for any $\mathsf{C}$-morphisms $f : A \rightarrow B$ and $g : B \rightarrow C$; 
and for each object $A \in |\mathsf{C}|$,
the monotonic identity element $1^{\mathrm{op}}_A = 1_A$.
A $\mathsf{C}$-morphism $f : A \rightarrow B$ is an \emph{isomorphism} $f : A \cong B$
when there is an oppositely-directed $\mathsf{C}$-morphism $f^{-1} : B \rightarrow A$
called its inverse
such that $f \cdot f^{-1} = 1_A$ and $f^{-1} \cdot f = 1_B$.
A $\mathsf{C}$-morphism $f : A \rightarrow B$ is an \emph{equivalence} $f : A \equiv B$
when there is an oppositely-directed $\mathsf{C}$-morphism $f^\prime : B \rightarrow A$
called its pseudo-inverse
such that $f \cdot f^\prime \equiv 1_A$ and $f^\prime \cdot f \equiv 1_B$.
A $\mathsf{C}$-morphism $e : A \rightarrow B$ is a \emph{pseudo-epimorphism}
when for any parallel pair of $\mathsf{C}$-morphisms $f, g : B \rightarrow C$,
if $e \cdot_{\mathsf{C}} f \equiv e \cdot_{\mathsf{C}} g$ then $f \equiv g$.
There is a dual definition for a \emph{pseudo-monomorphism}.

An \emph{order-enriched functor}
$\mathsf{F} : \mathsf{A} \rightarrow \mathsf{B}$ 
between two order-enriched categories
$\mathsf{A}$ and $\mathsf{B}$,
consists of:
an object function $|\mathsf{F}| : |\mathsf{A}| \rightarrow |\mathsf{B}|$; 
and for each pair of objects $A_1,A_2 \in |\mathsf{A}|$,
a monotonic function between hom-preorders 
$\mathsf{F}(A_1,A_2) 
: \mathsf{A}(A_1,A_2) \rightarrow \mathsf{B}(|\mathsf{F}|(A_1),|\mathsf{F}|(A_2))$.
This data is subject to 
the compatibility law for composition
$\mathsf{F}(A_1,A_3)(f \cdot_{A_1,A_2,A_3} g) 
= \mathsf{F}(A_1,A_2)(f) \cdot_{|\mathsf{F}|(A_1),|\mathsf{F}|(A_2),|\mathsf{F}|(A_3)}
\mathsf{F}(A_2,A_3)(g)$
for each composable pair $f : A_1 \rightarrow A_2$ and $g : A_2 \rightarrow A_3$,
and
the compatibility law for identity
$\mathsf{F}(A,A)(1_A) = 1_{|\mathsf{F}|(A)}$
for each object $A \in |\mathsf{A}|$.
An \emph{involution} $\propto$ on an order-enriched category $\mathsf{A}$ is an order-enriched functor that is an isomorphism between categories.
An involution consists of
an object bijection $(\mbox{-})^\propto : |\mathsf{A}| \rightarrow |\mathsf{A}|$; 
and 
for each pair of objects $A_1,A_2 \in |\mathsf{A}|$,
an order isomorphism 
$(\mbox{-})^\propto_{A_1,A_2} 
: \mathsf{A}^{\mathrm{op}}(A_2,A_1) \rightarrow \mathsf{A}(A_1^\propto,A_2^\propto)$
that respects composition and identities:
such that
$(g \cdot^{\mathrm{op}}_{A_3,A_2,A_1} f)^\propto_{A_1,A_3}
= f^\propto_{A_1,A_2} \cdot_{A_1,A_2,A_3} g^\propto_{A_2,A_3})$
and 
$(1_{A^\propto})^\propto_{A,A} = 1_A$.
An object $A$ of an order-enriched category $\mathsf{C}$ is a \emph{posetal object} 
when the hom-order $\mathsf{C}(B,A)$ is a poset for any objects $B \in |\mathsf{C}|$.
Let $\mathsf{C}_{=}$ denote the full replete subcategory of all posetal objects of $\mathsf{C}$,
and let $\mathsf{incl} : \mathsf{C}_{=} \rightarrow \mathsf{C}$ denote the (order-enriched) inclusion functor.

An \emph{order-enriched pseudo-natural transformation}
$\alpha : \mathsf{F} \Rightarrow \mathsf{G} : \mathsf{A} \rightarrow \mathsf{B}$ 
between two order-enriched functors $\mathsf{F}$ and $\mathsf{G}$,
is an $|\mathsf{A}|$-indexed collection of $\mathsf{B}$-morphisms
$\{ \alpha_A \in \mathsf{B}(|\mathsf{F}|(A),|\mathsf{G}|(A)) \mid A \in |\mathsf{A}| \}$
such that
$\alpha_{A_1} \cdot_{\mathsf{B}} \mathsf{G}(A_1,A_2)(h)
\equiv \mathsf{F}(A_1,A_2)(h) \cdot_{\mathsf{B}} \alpha_{A_2}$
for every $\mathsf{A}$-morphism $h : A_1 \rightarrow A_2$.
This equivalence expresses the pseudo-naturality of $\alpha$.
The vertical composite 
$\alpha \bullet \beta : \mathsf{F} \Rightarrow \mathsf{H} : \mathsf{A} \rightarrow \mathsf{B}$ 
of two pseudo-natural transformations
$\alpha : \mathsf{F} \Rightarrow \mathsf{G} : \mathsf{A} \rightarrow \mathsf{B}$ 
and
$\beta : \mathsf{G} \Rightarrow \mathsf{H} : \mathsf{A} \rightarrow \mathsf{B}$ 
has the component
$(\alpha \bullet \beta)_A \doteq \alpha_A \cdot_{\mathsf{B}} \beta_A$
for each $A \in |\mathsf{A}|$. 
The composite 
$\mathsf{K}\alpha : \mathsf{K} \circ \mathsf{F} \Rightarrow \mathsf{K} \circ \mathsf{G} : \mathsf{D} \rightarrow \mathsf{B}$ 
of $\mathsf{K} : \mathsf{D} \rightarrow \mathsf{A}$
with $\alpha : \mathsf{F} \Rightarrow \mathsf{G} : \mathsf{A} \rightarrow \mathsf{B}$ 
has component
$(\mathsf{K}\alpha)_D \doteq \alpha_{|\mathsf{K}|(D)}$
for each $D \in |\mathsf{D}|$. 
The composite 
$\alpha\mathsf{H} : \mathsf{F} \circ \mathsf{H} \Rightarrow \mathsf{G} \circ \mathsf{H} : \mathsf{A} \rightarrow \mathsf{C}$ 
of $\alpha : \mathsf{F} \Rightarrow \mathsf{G} : \mathsf{A} \rightarrow \mathsf{B}$ 
with $\mathsf{H} : \mathsf{B} \rightarrow \mathsf{C}$
has component
$(\alpha\mathsf{H})_A \doteq \mathsf{H}(|\mathsf{F}|(A),|\mathsf{G}|(A))(\alpha_A)$
for each $A \in |\mathsf{A}|$. 
Because of the pseudo-naturality of $\beta$,
we have the equivalence
$\mathsf{F}\beta \bullet \alpha\mathsf{K} \equiv \alpha\mathsf{H} \bullet \mathsf{G}\beta$ 
for two pseudo-natural transformations
$\alpha : \mathsf{F} \Rightarrow \mathsf{G} : \mathsf{A} \rightarrow \mathsf{B}$ 
and
$\beta : \mathsf{H} \Rightarrow \mathsf{K} : \mathsf{B} \rightarrow \mathsf{C}$.
We could choose either one of these for
the horizontal composite 
$\alpha \circ \beta : \mathsf{F} \circ \mathsf{H} \Rightarrow \mathsf{G} \circ \mathsf{K} : \mathsf{A} \rightarrow \mathsf{C}$.

\paragraph{Adjunctions}
A $\mathsf{C}$-\emph{adjunction} 
$g 
= \langle \mathsf{left}(g), \mathsf{right}(g) \rangle
= \langle \check{g}, \hat{g} \rangle
: A \rightleftharpoons B$
in an order-enriched category $\mathsf{C}$ consists of 
a left adjoint $\mathsf{C}$-morphism in the forward direction
$\mathsf{left}(g) = \check{g} : A \rightarrow B$
and a right adjoint $\mathsf{C}$-morphism in the reverse direction
$\mathsf{right}(g) = \hat{g} : B \rightarrow A$
satisfying the adjointness conditions:
$a \circ \check{g} \leq b$
\underline{iff}
$a \leq b \circ \hat{g}$
for every $\mathsf{C}$-object $C$
and every pair of $\mathsf{C}$-morphisms
$a : C \rightarrow A$ and $b : C \rightarrow B$
(or equivalently,
$1_A \leq \check{g} \circ \hat{g}$
and
$\hat{g} \circ \check{g} \leq 1_B$).
For any two adjunctions
$g_1, g_2 : A \rightleftharpoons B$,
$g_1 \leq g_2$
when
$\check{g}_1 \leq \check{g}_2$
(or equivalently, when $\hat{g}_2 \leq \hat{g}_1$).
Let $\mathsf{Adj}_{\mathsf{C}}(A,B)$
denote the preorder of all $\mathsf{C}$-adjunctions from $A$ to $B$.
Composition and identities of adjunctions are defined componentwise.
Let $\mathsf{Adj}_{\mathsf{C}}$ denote the order-enriched category of 
$\mathsf{C}$-objects and $\mathsf{C}$-adjunctions.
Posetal objects and adjunctions form the full subcategory $\mathsf{Adj}_{\mathsf{C},=} \subset \mathsf{Adj}_{\mathsf{C}}$. 
Projecting to the left and right gives rise to two component order-enriched functors.
The left functor 
$\mathsf{left}_{\mathsf{C}} : \mathsf{Adj}_{\mathsf{C}} \rightarrow \mathsf{C}$
is the identity of objects and maps an adjunction
$g : A \rightleftharpoons B$
to its left component
$\mathsf{left}_{\mathsf{C}}(g) = \check{g} : A \rightarrow B$.
The right functor 
$\mathsf{right}_{\mathsf{C}} : \mathsf{Adj}^{\mathrm{op}}_{\mathsf{C}} \rightarrow \mathsf{C}$
is the identity of objects and maps an adjunction
$g : A \rightleftharpoons B$
to its right component
$\mathsf{right}_{\mathsf{C}}(g) = \hat{g} : B \rightarrow A$.
The order-enriched \emph{involution} isomorphism
${(\mbox{-})}^{\propto} : \mathsf{Adj}^{\mathrm{op}} \rightarrow \mathsf{Adj}$
flips source/target and left/right:
${(\mbox{-})}^{\propto} \!\circ\, \mathsf{left} = \mathsf{right}$ 
and
${(\mbox{-})}^{\propto} \!\circ\, \mathsf{right}^{\mathrm{op}} = \mathsf{left}^{\mathrm{op}}$.

Let $g : A \rightleftharpoons B$ be any $\mathsf{C}$-adjunction.
The \emph{closure} of $g$ is the $\mathsf{C}$-endomorphism
$(\mbox{-})^{\bullet_{g}}
= \check{g} \circ \hat{g}
: A \rightarrow A$.
Closure is increasing $1_A \leq (\mbox{-})^{\bullet_{g}}$ and idempotent $(\mbox{-})^{\bullet_{g}} \circ (\mbox{-})^{\bullet_{g}} \equiv (\mbox{-})^{\bullet_{g}}$.
Idempotency is implied by the fact that
$\check{g} \circ \hat{g} \circ \check{g} \equiv \check{g}$.
The closure of any $A$-element $a : 1 \rightarrow A$
is the $A$-element 
$a^{\bullet_{g}} = a \circ (\mbox{-})^{\bullet_{g}} : 1 \rightarrow A$.
The closure equalizer diagram is the parallel pair
$1_A, (\mbox{-})^{\bullet_{g}} : A \rightarrow A$.
The subobject of closed elements of $g$
is defined to be the equalizer 
$\mathrm{incl}_0^{g} : \mathsf{clo}(g) \rightarrow A$
of this diagram.
Being part of a limiting cone, 
$\mathrm{incl}_0^{g} \cdot (\mbox{-})^{\bullet_{g}} 
= \mathrm{incl}_0^{g}$.
An $A$-element $a : 1 \rightarrow A$ is a closed element of $g$
when
it factors through $\mathsf{clo}(g)$;
that is,
there is an element $\bar{a} : 1 \rightarrow \mathsf{clo}(g)$ 
such that $a$ is equal to its inclusion $a = \bar{a} \circ \mathrm{incl}_0^{g}$
(or equivalently,
when $a$ is equal to its closure $a = a^{\bullet_{g}}$).
An $A$-element $a : 1 \rightarrow A$ is a pseudo-closed element of $g$
when $a$ is equivalent to its closure 
$a \equiv_{A} a^{\bullet_{g}}$
(or equivalently,
when $a$ is equivalent $a \equiv_{A} b \circ \hat{g}$ 
to the image of some target element $b \in B$).
Any closed element of $g$ is a pseudo-closed element of $g$.
When $A$ is a poset, 
the closed and pseudo-closed elements of $g$ coincide.

Let $g : A \rightleftharpoons B$ be any $\mathsf{C}$-adjunction.
The \emph{interior} of $g$ is the $\mathsf{C}$-endomorphism
$(\mbox{-})^{\circ_{g}}
= \hat{g} \circ \check{g}
: B \rightarrow B$.
Interior is decreasing $1_B \geq (\mbox{-})^{\circ_{g}}$ and idempotent $(\mbox{-})^{\circ_{g}} \circ (\mbox{-})^{\circ_{g}} \equiv (\mbox{-})^{\circ_{g}}$.
Idempotency is implied by the fact that
$\hat{g} \circ \check{g} \circ \hat{g} \equiv \hat{g}$.
The interior of any $B$-element $b : 1 \rightarrow B$
is the $B$-element 
$b^{\circ_{g}} = b \circ (\mbox{-})^{\circ_{g}} : 1 \rightarrow B$.
The interior equalizer diagram is the parallel pair
$1_B, (\mbox{-})^{\circ_{g}} : B \rightarrow B$.
The subobject of open elements of $g$
is defined to be the equalizer 
$\mathrm{incl}_1^{g} : \mathsf{open}(g) \rightarrow B$
of this diagram.
Being part of a limiting cone, 
$\mathrm{incl}_1^{g} \cdot (\mbox{-})^{\circ_{g}} 
= \mathrm{incl}_1^{g}$.
A $B$-element $b : 1 \rightarrow B$ is an open element of $g$
when
it factors through $\mathsf{open}(g)$;
that is,
there is an element $\tilde{b} : 1 \rightarrow \mathsf{open}(g)$ 
such that $b$ is equal to its inclusion $b = \tilde{b} \circ \mathrm{incl}_1^{g}$
(or equivalently,
when $b$ is equal to its interior $b = b^{\circ_{g}}$).
A $B$-element $b : 1 \rightarrow B$ is a pseudo-open element of $g$
when $b$ is equivalent to its interior 
$b \equiv_{B} b^{\circ_{g}}$
(or equivalently,
when $b$ is equivalent $b \equiv_{B} a \circ \check{g}$ 
to the image of some source element $a \in A$).
Any open element of $g$ is a pseudo-open element of $g$.
When $B$ is a poset, 
the open and pseudo-open elements of $g$ coincide.

A $\mathsf{C}$-\emph{pseudo-reflection} is a $\mathsf{C}$-adjunction 
$g : A \rightleftharpoons B$
that satisfies the equivalence
$1_{B} \equiv (\mbox{-})^{\circ_{g}}$.
Any pseudo-reflection is a pseudo-epimorphism. 
Let $\widetilde{\mathsf{Ref}}(\mathsf{C})$ denote the morphism subclass of all $\mathsf{C}$-pseudo-reflections.
A $\mathsf{C}$-\emph{reflection} is a $\mathsf{C}$-pseudo-reflection that is strict: 
it satisfies the identity
$1_{B} = (\mbox{-})^{\circ_{g}}$. 
If $g : A \rightleftharpoons B$ is a (strict) $\mathsf{C}$-reflection
and the source $A$ is posetal,
then the target $B$ is also posetal.
Let $\mathsf{Ref}(\mathsf{C})$ denote the morphism subclass of all $\mathsf{C}$-reflections.
A $\mathsf{C}$-\emph{pseudo-coreflection} is a $\mathsf{C}$-adjunction 
$g : A \rightleftharpoons B$
that satisfies the equivalence
$1_{A} \equiv (\mbox{-})^{\bullet_{g}}$. 
Any pseudo-coreflection is a pseudo-monomorphism. 
Let $\widetilde{\mathsf{Ref}}^\propto(\mathsf{C})$ denote the morphism subclass of all $\mathsf{C}$-pseudo-coreflections.
A $\mathsf{C}$-\emph{coreflection} is a $\mathsf{C}$-pseudo-coreflection that is strict:
it satisfies the identity
$1_{A} = (\mbox{-})^{\bullet_{g}}$. 
If $g : A \rightleftharpoons B$ is a (strict) $\mathsf{C}$-coreflection
and the target $B$ is posetal,
then the source $A$ is also posetal.
Let $\mathsf{Ref}^\propto(\mathsf{C})$ denote the morphism subclass of all $\mathsf{C}$-coreflections.
The involution of a $\mathsf{C}$-pseudo-reflection is a $\mathsf{C}$-pseudo-coreflection, and vice-versa.
A $\mathsf{C}$-isomorphism $f : A \rightarrow B$ forms an adjunction with its inverse
$\langle f, f^{\mbox{-1}} \rangle : A \rightleftharpoons B$,
which is also called a $\mathsf{C}$-isomorphism.
The transpose $\langle f^{\mbox{-1}}, f \rangle : B \rightleftharpoons A$
of a $\mathsf{C}$-isomorphism is also a $\mathsf{C}$-isomorphism. 
Let $\mathsf{iso}(\mathsf{C})$ denote the morphism subclass of all $\mathsf{C}$-isomorphisms.
A $\mathsf{C}$-equivalence $f : A \rightarrow B$ forms an adjunction with its pseudo-inverse
$\langle f, f^\prime \rangle : A \rightleftharpoons B$,
which is also called a $\mathsf{C}$-equivalence.
The transpose $\langle f^\prime, f \rangle : B \rightleftharpoons A$
of a $\mathsf{C}$-equivalence is also a $\mathsf{C}$-equivalence. 
Let $\mathsf{equ}(\mathsf{C})$ denote the morphism subclass of all $\mathsf{C}$-equivalences.
Any $\mathsf{C}$-isomorphism is a $\mathsf{C}$-equivalence.
A $\mathsf{C}$-morphism is a $\mathsf{C}$-isomorphism iff it is both a $\mathsf{C}$-reflection and a $\mathsf{C}$-coreflection.
A $\mathsf{C}$-morphism is a $\mathsf{C}$-equivalence iff it is both a $\mathsf{C}$-pseudo-reflection and a $\mathsf{C}$-pseudo-coreflection.
The connection between these morphism classes is summarized as follows.
\begin{center}
$\begin{array}{rcl}
\mathsf{iso}(\mathsf{C}) & \subseteq & \mathsf{equ}(\mathsf{C}) \\
\mathsf{Ref}(\mathsf{C}) & \subseteq & \widetilde{\mathsf{Ref}}(\mathsf{C}) \\
\mathsf{Ref}^\propto(\mathsf{C}) & \subseteq & \widetilde{\mathsf{Ref}}^\propto(\mathsf{C})
\\
\mathsf{iso}(\mathsf{C}) & = & \mathsf{Ref}(\mathsf{C}) \cap \mathsf{Ref}^\propto(\mathsf{C}) \\
\mathsf{equ}(\mathsf{C}) & = & \widetilde{\mathsf{Ref}}(\mathsf{C}) \cap \widetilde{\mathsf{Ref}}^\propto(\mathsf{C})
\end{array}$
\end{center}

\paragraph{Involutions}
An \emph{involution} $\propto$ on an order-enriched category $\mathsf{C}$ 
is an order-enriched functor 
${(\mbox{-})}^{\propto} : \mathsf{C}^{\mathrm{op}} \rightarrow \mathsf{C}$
that is an isomorphism between categories.
An involution consists of
an object bijection $(\mbox{-})^\propto : |\mathsf{A}| \rightarrow |\mathsf{A}|$; 
and 
for each pair of objects $A_1,A_2 \in |\mathsf{A}|$,
an order isomorphism 
$(\mbox{-})^\propto_{A_1,A_2} 
: \mathsf{A}^{\mathrm{op}}(A_2,A_1) \rightarrow \mathsf{A}(A_1^\propto,A_2^\propto)$
that respects composition and identities:
such that
$(g \circ^{\mathrm{op}}_{A_3,A_2,A_1} f)^\propto_{A_1,A_3}
= f^\propto_{A_1,A_2} \circ_{A_1,A_2,A_3} g^\propto_{A_2,A_3})$
and $(1_{A^\propto})^\propto_{A,A} = 1_A$.

There is an \emph{involution} on $\mathsf{C}$-adjunctions
${(\mbox{-})}_{\mathsf{C}}^{\propto} : \mathsf{Adj}(\mathsf{C})^{\mathrm{op}} \rightarrow \mathsf{Adj}(\mathsf{C})$
that flips source/target and left/right:
${(\mbox{-})}_{\mathsf{C}}^{\propto} \!\circ\, \mathsf{left}_{\mathsf{C}} = \mathsf{right}_{\mathsf{C}}$ 
and
${(\mbox{-})}_{\mathsf{C}}^{\propto} \!\circ\, \mathsf{right}_{\mathsf{C}}^{\mathrm{op}} = \mathsf{left}_{\mathsf{C}}^{\mathrm{op}}$.
The involution of a $\mathsf{C}$-reflection is a $\mathsf{C}$-coreflection, and vice-versa.
There is a \emph{direct} image functor
$\mathsf{dir}_{\mathsf{C}} : \mathsf{B} \rightarrow \mathsf{Adj}(\mathsf{C})$,
where $\mathsf{dir}_{\mathsf{C}} \circ \mathsf{left}_{\mathsf{C}} 
= \exists_{\mathsf{C}} : \mathsf{B} \rightarrow \mathsf{C}$
and $\mathsf{dir}_{\mathsf{C}} \circ \mathsf{right}_{\mathsf{C}}^{\mathrm{op}} 
= {{(-)}_{\mathsf{C}}^{-1}}^{\mathrm{op}} : \mathsf{B} \rightarrow \mathsf{C}^{\mathrm{op}}$.
There is a \emph{inverse} image functor
$\mathsf{inv}_{\mathsf{C}} : \mathsf{B}^{\mathrm{op}} \rightarrow \mathsf{Adj}(\mathsf{C})$,
where $\mathsf{inv}_{\mathsf{C}} \circ \mathsf{left}_{\mathsf{C}} 
= {(-)}_{\mathsf{C}}^{-1} : \mathsf{B}^{\mathrm{op}} \rightarrow \mathsf{C}$
and $\mathsf{inv}_{\mathsf{C}} \circ \mathsf{right}_{\mathsf{C}}^{\mathrm{op}} 
= \exists_{\mathsf{C}}^{\mathrm{op}} : \mathsf{B}^{\mathrm{op}} \rightarrow \mathsf{C}^{\mathrm{op}}$.
Direct is the involution of inverse $\mathsf{dir}_{\mathsf{C}} = \mathsf{inv}_{\mathsf{C}}^{\mathrm{op}} \circ {(\mbox{-})}_{\mathsf{C}}^{\propto}$.

An involution ${(\mbox{-})}^\propto$ respects a factorization system
$\langle \mathsf{E}, \mathsf{M} \rangle$ when
it maps morphisms in $\mathsf{E}$ to morphisms in $\mathsf{M}$
and vice-versa.

\begin{center}
$\begin{array}{l}
\begin{array}{c@{\hspace{5pt}\Rightarrow\hspace{5pt}}c}
\mathsf{E} & \mathsf{M} \\
\mathsf{M} & \mathsf{E} \\ 
\end{array} \\ \hline
{\propto}_{\mathsf{C}} : \mathsf{mor}(\mathsf{C}) \rightarrow \mathsf{mor}(\mathsf{C}) \\
{\propto}_{\mathsf{E}\mathsf{M}} : \mathsf{E} \rightarrow \mathsf{M} \\
{\propto}_{\mathsf{M}\mathsf{E}} : \mathsf{M} \rightarrow \mathsf{E} \\ \hline
{\propto}_{\mathsf{C}} \cdot \mathsf{src}_{\mathsf{C}} = \mathsf{tgt}_{\mathsf{C}} \\
{\propto}_{\mathsf{C}} \cdot \mathsf{tgt}_{\mathsf{C}} = \mathsf{src}_{\mathsf{C}} \\
{\propto}_{\mathsf{E}\mathsf{M}} \cdot {\propto}_{\mathsf{M}\mathsf{E}} = 1 \\
{\propto}_{\mathsf{M}\mathsf{E}} \cdot {\propto}_{\mathsf{E}\mathsf{M}} = 1 \\
\mathrm{incl}_{\mathsf{E}} \cdot {\propto}_{\mathsf{C}} 
= {\propto}_{\mathsf{E}\mathsf{M}} \cdot \mathsf{incl}_{M} \\
\mathsf{incl}_\mathsf{M} \cdot {\propto}_{\mathsf{C}} 
= {\propto}_{\mathsf{M}\mathsf{E}} \cdot \mathrm{incl}_{E}
\end{array}$
\end{center}

The involution
${\propto}_{\mathsf{C}} : \mathsf{Adj}(\mathsf{C})^{\mathrm{op}} \rightarrow \mathsf{Adj}(\mathsf{C})$
restricts to isomorphisms,
restricts to equivalence,
maps morphisms in $\mathsf{E}$ to morphisms in $\mathsf{M}$,
and
maps morphisms in $\mathsf{M}$ to morphisms in $\mathsf{E}$.

Dual correspondents on $\mathsf{Adj}(\mathsf{C})$ for an order-enriched category $\mathsf{C}$:

\begin{center}
$\begin{array}{l}
\begin{array}{|c@{\hspace{5pt}\Rightarrow\hspace{5pt}}c|} \hline
A_0 & A_1 \\
\mathsf{left}(g) : A_0 \rightarrow A_1 & \mathsf{right}(g) : A_1 \rightarrow A_0 \\
(\mbox{-})^{\bullet_g} : A_0 \rightarrow A_0 & (\mbox{-})^{\circ_g} : A_1 \rightarrow A_1 \\
\mathsf{clo}(g) \subseteq A_0	& \mathsf{open}(g) \subseteq A_1 \\
\mathsf{axis}(g) & \mathsf{axis}(g) \\
\mathsf{ref}(g) : A_0 \rightleftharpoons \mathsf{axis}(g)		
&	\mathsf{ref}^\propto(g) : \mathsf{axis}(g) \rightleftharpoons A_1	\\
(A_0,\mathsf{ref}(g),\mathsf{axis}(g),\mathsf{ref}^\propto(g),A_1) 
& (A_1,\mathsf{ref}^\propto(g),\mathsf{axis}(g),\mathsf{ref}(g),A_0) \\ \hline
  \mathsf{src}_\mathsf{C}
&	\mathsf{tgt}_\mathsf{C}            \\
  \mathsf{left}_\mathsf{C}  : \mathsf{Adj}(\mathsf{C}) \rightarrow \mathsf{C}
&	\mathsf{right}_\mathsf{C} : \mathsf{Adj}(\mathsf{C})^{\mathrm{op}} \rightarrow \mathsf{C}      \\
\mathsf{clo}_\mathsf{C} & \mathsf{open}_\mathsf{C} \\
  \mathsf{ref}_\mathsf{C} 
  : \mathsf{C}^{\mathsf{2}} \rightarrow	\mathsf{Ref}(\mathsf{C})^{\mathsf{2}}
& \mathsf{ref}^\propto_\mathsf{C} 
  : \mathsf{C}^{\mathsf{2}} \rightarrow \mathsf{Ref}^\propto(\mathsf{C})^{\mathsf{2}} \\
\mathsf{Ref}(\mathsf{C})		&	\mathsf{Ref}^\propto(\mathsf{C}) \\
\mathsf{Iso}(\mathsf{C})		&	\mathsf{Iso}(\mathsf{C})         \\
\mathsf{Equiv}(\mathsf{C})	&	\mathsf{Equiv}(\mathsf{C})       \\
\mathsf{axis}_\mathsf{C}		&	\mathsf{axis}_\mathsf{C}         \\ \hline
\end{array}
\\ \\
\mathsf{Iso}(\mathsf{C}) 
\subseteq \mathsf{Equiv}(\mathsf{C})
\subseteq \mathsf{Ref}(\mathsf{C}) \bigcap \mathsf{Ref}^\propto(\mathsf{C})
\\
\mathsf{Ref}(\mathsf{C}), \mathsf{Ref}^\propto(\mathsf{C}) \subseteq \mathsf{Adj}(\mathsf{C})
\end{array}$
\end{center}

\subsection{Order-Enriched Fibrations\label{subsec:order:enriched:fibrations}}

An \emph{order-enriched fibration} 
$\mathsf{P} : \mathsf{E} \rightarrow \mathsf{B}$
is an order-enriched functor,
which is a split fibration 
that satisfies the following additional conditions\footnote{In a sense,
in this definition we are characterizing the fibration 
$|\mbox{-}| : \mathsf{Ord} \rightarrow \mathsf{Set}$.}.
\begin{itemize}
\item
The functor $\mathsf{P}$ is faithful, replete and preserves limits.
\item {\bfseries [cleavage uniqueness]}
The cleavage $\gamma(b, E_1)$ is the only cartesian morphism for a fiber pair $(b, E_1)$; 
that is,
if $e : E_2 \rightarrow E_1$ is cartesian
then $\Delta(e) = E_2$, $\flat_{e} = 1_{E_2}$ and $\sharp_{e} = e$\footnote{Equivalently (assuming the fibration splits),
if an $\mathsf{E}$-morphism $e : E_3 \rightarrow E_1$
factors as $e = e_2 \cdot_{\mathsf{E}} e_1 : E_3 \rightarrow E_2 \rightarrow E_1$
where $e_1$ is cartesian,
then $\Delta(e) = \Delta(e_2)$,
$\flat_{e} = \flat_{e_2} : E_2 \rightarrow \Delta(e)$
and $\sharp_{e} = \sharp_{e_2} \cdot_{\mathsf{E}} e_1 : \Delta(e) \rightarrow E_1$.}.
\item {\bfseries [right cancellation]}
For any cartesian $\mathsf{E}$-morphism $e : E_2 \rightarrow E_1$
and any parallel pair of $\mathsf{E}$-morphisms $g,h : E_3 \rightarrow E_2$,
if $g \cdot_{\mathsf{E}} e \leq h \cdot_{\mathsf{E}} e$ then $g \leq h$\footnote{It follows that,
(1) any cartesian $\mathsf{E}$-morphism is a pseudo-monomorphism:
for any cartesian $\mathsf{E}$-morphism $e : E_2 \rightarrow E_1$ 
and any parallel pair of $\mathsf{E}$-morphisms $g,h : E_3 \rightarrow E_2$,
if $g \cdot_{\mathsf{E}} e \equiv h \cdot_{\mathsf{E}} e$ then $g \equiv h$;
(2) for any $\mathsf{E}$-morphism $e : E_2 \rightarrow E_1$ 
and any parallel pair of $\mathsf{E}$-morphisms $g,h : E_3 \rightarrow E_2$,
if $g \cdot_{\mathsf{E}} e \leq h \cdot_{\mathsf{E}} e$ then $\delta_{g,e} \leq \delta_{h,e}$.}.
\item {\bfseries [left cancellation]}
For any $\mathsf{E}$-morphism $e : E_2 \rightarrow E_1$
and any parallel pair of $\mathsf{E}$-morphisms $g, h : \Delta(e) \rightarrow E$,
if $\flat_{e} \cdot_{\mathsf{E}} g \leq \flat_{e} \cdot_{\mathsf{E}} h$ then $g \leq h$\footnote{It follows that any gap is a pseudo-epimorphism:
for any $\mathsf{E}$-morphism $e : E_2 \rightarrow E_1$
and any parallel pair of $\mathsf{E}$-morphisms $g, h : \Delta(e) \rightarrow E$,
if $\flat_{e} \cdot_{\mathsf{E}} g \equiv \flat_{e} \cdot_{\mathsf{E}} h$ then $g \equiv h$.}.
\item {\bfseries [equivalence factorization]} 
An equivalent pair of $\mathsf{E}$-morphisms 
$f \equiv g : E_2 \rightarrow E_1$
have the same apex $\Delta(f) = \Delta(g)$
and gap $\flat_{f} = \flat_{g}$
and equivalent lifts $\sharp_{f} \equiv \sharp_{g} : \Delta(f) \rightarrow E_1$.
\end{itemize}

\begin{fact}
Equalizing morphisms are cartesian.
\end{fact}

\noindent
{\bfseries Proof:}
Let $f, g : B \rightarrow A$ be a parallel pair of $\mathsf{E}$-morphisms,
and let $e : E \rightarrow B$ be its equalizing $\mathsf{E}$-morphism. 
Suppose $h : C \rightarrow B$ is a $\mathsf{E}$-morphism
and $k : \mathsf{P}(C) \rightarrow \mathsf{P}(E)$ is a $\mathsf{B}$-morphism
such that $k \cdot_{\mathsf{B}} \mathsf{P}(e) = \mathsf{P}(h)$.
Since $\mathsf{P}$ preserves limits,
$\mathsf{P}(e) : \mathsf{P}(E) \rightarrow \mathsf{P}(B)$ 
is the equalizing $\mathsf{B}$-morphism of the parallel pair of $\mathsf{B}$-morphisms
$\mathsf{P}(f), \mathsf{P}(g) : \mathsf{P}(B) \rightarrow \mathsf{P}(A)$.
Hence,
$\mathsf{P}(h \cdot_{\mathsf{E}} f)
= \mathsf{P}(h) \cdot_{\mathsf{E}} \mathsf{P}(f)
= k \cdot_{\mathsf{B}} \mathsf{P}(e) \cdot_{\mathsf{E}} \mathsf{P}(f)
= k \cdot_{\mathsf{B}} \mathsf{P}(e) \cdot_{\mathsf{E}} \mathsf{P}(g)
= \mathsf{P}(h) \cdot_{\mathsf{E}} \mathsf{P}(g)
= \mathsf{P}(h \cdot_{\mathsf{E}} g)$,
and
$h \cdot_{\mathsf{E}} f = h \cdot_{\mathsf{E}} g$
since $\mathsf{P}$ is faithful.
Therefore,
there is a unique $\mathsf{E}$-morphism
$\tilde{k} : C \rightarrow E$
such that $\tilde{k} \cdot_{\mathsf{E}} e = h$. 
But 
$\mathsf{P}(\tilde{k}) = k$,
since $\mathsf{P}(e)$ is a $\mathsf{B}$-monomorphism
and
$\mathsf{P}(\tilde{k}) \cdot_{\mathsf{E}} \mathsf{P}(e)
= \mathsf{P}(\tilde{k} \cdot_{\mathsf{E}} e)
= \mathsf{P}(h)
= k \cdot_{\mathsf{B}} \mathsf{P}(e)$.
\rule{2pt}{6pt}

\begin{fact}
The right adjoint of a reflection is cartesian.
The left adjoint of a coreflection is cartesian.
\end{fact}

\noindent
{\bfseries Proof:}
Let
$g
= \langle \check{g}, \hat{g} \rangle 
: A_0 \rightleftharpoons A_1$ 
denote a reflection,
so that $1_{A_0} \leq \check{g} \cdot_{\mathsf{E}} \hat{g}$
and $\hat{g} \cdot_{\mathsf{E}} \check{g} = 1_{A_1}$.
We show that $\hat{g}$ is a cartesian $\mathsf{E}$-morphism. 
Let $h : B \rightarrow A_0$ be an $\mathsf{E}$-morphism
and let $f : \mathsf{P}(B) \rightarrow \mathsf{P}(A_1)$ be a $\mathsf{B}$-morphism
such that $\mathsf{P}(h) = f \cdot_{\mathsf{B}} \mathsf{P}(\hat{g})$.
{\bfseries [Existence]}
Consider the $\mathsf{E}$-morphism
$h \cdot_{\mathsf{E}} \check{g} : B \rightarrow A_1$.
The underlying $\mathsf{B}$-morphism is $f$,
since
$\mathsf{P}(h \cdot_{\mathsf{E}} \check{g})
= \mathsf{P}(h) \cdot_{\mathsf{B}} \mathsf{P}(\check{g})
= f \cdot_{\mathsf{B}} \mathsf{P}(\hat{g}) \cdot_{\mathsf{B}}  \mathsf{P}(\check{g})
= f \cdot_{\mathsf{B}} \mathsf{P}(\hat{g} \cdot_{\mathsf{E}} \check{g})
= f \cdot_{\mathsf{B}} \mathsf{P}(1_{A_1})
= f$.
The composition with $\hat{g}$ is $h$,
since
$\mathsf{P}(h \cdot_{\mathsf{E}} \check{g} \cdot_{\mathsf{E}} \hat{g})
= \mathsf{P}(h) \cdot_{\mathsf{B}} \mathsf{P}(\check{g} \cdot_{\mathsf{E}} \hat{g})
= f \cdot_{\mathsf{B}} \mathsf{P}(\hat{g}) \cdot_{\mathsf{B}} \mathsf{P}(\check{g} \cdot_{\mathsf{E}} \hat{g})
= f \cdot_{\mathsf{B}} \mathsf{P}(\hat{g} \cdot_{\mathsf{E}} \check{g} \cdot_{\mathsf{E}} \hat{g})
= f \cdot_{\mathsf{B}} \mathsf{P}(\hat{g})
= \mathsf{P}(h)$
and the fibration $\mathsf{P}$ is faithful.
{\bfseries [Uniqueness]}
Suppose $k : B \rightarrow A_1$ is an $\mathsf{E}$-morphism such that
$\mathsf{P}(k) = f$ and $k \cdot_{\mathsf{E}} \hat{g} = h$.
Then
$k 
= k \cdot_{\mathsf{E}} 1_{A_1}
= k \cdot_{\mathsf{E}} \hat{g} \cdot_{\mathsf{E}} \check{g}
= h \cdot_{\mathsf{E}} \check{g}$.
The proof for coreflections is dual.
\rule{2pt}{6pt}

\begin{center}
SPACE
\end{center}
\begin{center}
SPACE
\end{center}
\begin{center}
SPACE
\end{center}
\begin{center}
SPACE
\end{center}
\begin{center}
SPACE
\end{center}
\begin{center}
SPACE
\end{center}
\begin{center}
SPACE
\end{center}
\begin{center}
SPACE
\end{center}
\begin{center}
SPACE
\end{center}

\subsection{Pseudo Factorization Systems\label{subsec:pseudo:factorization:systems}}

A pseudo-factorization system in an order-enriched category
relaxes the notion of 
a factorization system in a category
by using equivalences instead of isomorphisms,
and only requires uniqueness up to equivalence.
Let $\mathsf{C}$ be an arbitrary order-enriched category.
A \emph{pseudo-factorization system} in $\mathsf{C}$ is a pair 
$\langle \mathsf{E}, \mathsf{M} \rangle$ of classes of $\mathsf{C}$-morphisms satisfying the following conditions\footnote{Pseudo-factorization systems are often compatible with ordinary factorization systems. 
A \emph{tiered factorization system} in $\mathsf{C}$ is a quadruple 
$\langle \tilde{\mathsf{E}}, \mathsf{E}, \mathsf{M}, \tilde{\mathsf{M}} \rangle$,
consisting of an ordinary factorization system $\langle \mathsf{E}, \mathsf{M} \rangle$
within a pseudo-factorization system $\langle \tilde{\mathsf{E}}, \tilde{\mathsf{M}} \rangle$;
so that $\mathsf{E} \subseteq \tilde{\mathsf{E}}$ and $\mathsf{M} \subseteq \tilde{\mathsf{M}}$.}. 
{\bfseries Subcategories:}
All $\mathsf{C}$-equivalences (and hence also all $\mathsf{C}$-isomorphisms) are in $\mathsf{E} \,\cap\, \mathsf{M}$. 
Both $\mathsf{E}$ and $\mathsf{M}$ are closed under $\mathsf{C}$-composition.
Hence,
$\mathsf{E}$ and $\mathsf{M}$ are $\mathsf{C}$-subcategories with the same objects as $\mathsf{C}$.
{\bfseries Existence:} 
Every $\mathsf{C}$-morphism $f : A \rightarrow B$ has an $\langle \mathsf{E}, \mathsf{M} \rangle$-factorization\footnote{An $\langle \mathsf{E}, \mathsf{M} \rangle$-factorization is a quadruple $(A, e, C, m, B)$ where 
$e : A \rightarrow C$ and $m : C \rightarrow B$ 
is a composable pair of $\mathsf{C}$-morphisms 
with $e \in \mathsf{E}$ and $m \in \mathsf{M}$.};
that is,
there is an $\langle \mathsf{E}, \mathsf{M} \rangle$-factorization $(A, e, C, m, B)$ 
and $f$ is its composition\footnote{In this paper, all compositions are written in diagrammatic form.} $f = e \cdot m$.
{\bfseries Diagonalization:}
For every diagonalization square there is a diagonal,
unique up to equivalence;
that is, 
for every pseudo-commutative square $e \cdot s \equiv r \cdot m$ of $\mathsf{C}$-morphisms 
with $e \in \mathsf{E}$ and $m \in \mathsf{M}$, 
there is a $\mathsf{C}$-morphism $d$, 
unique up to equivalence, 
with $e \cdot d \equiv r$ and $d \cdot m \equiv s$.
This diagonalization condition implies the following condition.
{\bfseries Uniqueness:}
Any two $\langle \mathsf{E}, \mathsf{M} \rangle$-factorizations are equivalent; 
that is, 
if $(A, e, C, m, B)$ and $(A, e^\prime, C^\prime, m^\prime, B)$ 
are two $\langle \mathsf{E}, \mathsf{M} \rangle$-factorizations of $f : A \rightarrow B$, 
then there is a $\mathsf{C}$-equivalence $h : C \equiv C^\prime$, 
unique up to equivalence, 
with $e \cdot h \equiv e^\prime$ and $h \cdot m^\prime \equiv m$.
An \emph{epi-mono pseudo-factorization system} is one 
where $\mathsf{E}$ is contained in the class of $\mathsf{C}$-pseudo-epimorphisms
and $\mathsf{M}$ is contained in the class of $\mathsf{C}$-pseudo-monomorphisms.
For an epi-mono pseudo-factorization system,
the diagonalization condition is equivalent to the uniqueness condition.

Let $\mathsf{C}^{\tilde{\mathsf{2}}}$ denote the pseudo-arrow category\footnote{Recall that $\mathsf{2}$ is the two-object category, pictured as $\bullet \rightarrow \bullet$, with one non-trivial morphism. The pseudo-arrow category $\mathsf{C}^{\tilde{\mathsf{2}}}$ of the order-enriched category $\mathsf{C}$ is (isomorphic to) the pseudo-functor category $\tilde{[}\mathsf{2}, \mathsf{C}\tilde{]}$ of order-enriched functors and pseudo-natural transformations.} of $\mathsf{C}$.
An object of $\mathsf{C}^{\tilde{\mathsf{2}}}$ is a triple $(A, f, B)$,
where $f : A \rightarrow B$ is a $\mathsf{C}$-morphism.
A morphism of $\mathsf{C}^{\tilde{\mathsf{2}}}$,
$(a, b) : (A_1, f_1, B_1) \rightarrow (A_2, f_2, B_2)$,
is a pair of $\mathsf{C}$-morphisms $a : A_1 \rightarrow A_2$ and $b : B_1 \rightarrow B_2$ that form a pseudo-commutative square $a \cdot f_2 \equiv f_1 \cdot b$.
There are source and target projection functors
$\partial_0^{\mathsf{C}}, \partial_1^{\mathsf{C}} : \mathsf{C}^{\tilde{\mathsf{2}}} \rightarrow \mathsf{C}$
and an arrow natural transformation
$\alpha_{\mathsf{C}} : \partial_0^{\mathsf{C}} \Rightarrow \partial_1^{\mathsf{C}} : \mathsf{C}^{\tilde{\mathsf{2}}} \rightarrow \mathsf{C}$
with component
$\alpha_{\mathsf{C}}(A, f, B) = f : A \rightarrow B$
(background of Figure~\ref{pseudo-factorization-equivalence}).
Let $\mathsf{E}^{\tilde{\mathsf{2}}}$ denote the full subcategory of $\mathsf{C}^{\tilde{\mathsf{2}}}$
whose objects are the morphisms in $\mathsf{E}$.
Make the same definitions for $\mathsf{M}^{\tilde{\mathsf{2}}}$.
Just as for $\mathsf{C}^{\tilde{\mathsf{2}}}$,
the category $\mathsf{E}^{\tilde{\mathsf{2}}}$ has source and target projection functors
$\partial_0^\mathsf{E}, \partial_1^\mathsf{E} : \mathsf{E}^{\tilde{\mathsf{2}}} \rightarrow \mathsf{C}$
and an arrow natural transformation
$\alpha_{\mathsf{E}} : \partial_0^\mathsf{E} \Rightarrow \partial_1^\mathsf{E} : \mathsf{E}^{\tilde{\mathsf{2}}} \rightarrow \mathsf{C}$
(foreground of Figure~\ref{pseudo-factorization-equivalence}).
The same is true for $\mathsf{M}^{\tilde{\mathsf{2}}}$.
Let $\mathsf{E} \tilde{\odot} \mathsf{M}$ denote the category
of $\langle \mathsf{E}, \mathsf{M} \rangle$-pseudo-factorizations
(top foreground of Figure~\ref{pseudo-factorization-equivalence}),
whose objects are $\langle \mathsf{E}, \mathsf{M} \rangle$-factorizations $(A, e, C, m, B)$,
and whose morphisms $(a, c, b) : (A_1, e_1, C_1, m_1, B_1) \rightarrow (A_2, e_2, C_2, m_2, B_2)$
are $\mathsf{C}$-morphism triples
where $(a, c) : (A_1, e_1, C_1) \rightarrow (A_2, e_2, C_2)$ 
is an $\mathsf{E}^{\tilde{\mathsf{2}}}$-morphism
and $(c, b) : (C_1, m_1, B_1) \rightarrow (C_2, m_2, B_2)$ 
is an $\mathsf{M}^{\tilde{\mathsf{2}}}$-morphism.
$\mathsf{E} \tilde{\odot} \mathsf{M}
= \mathsf{E}^{\tilde{\mathsf{2}}} \times_{\mathsf{C}} \mathsf{M}^{\tilde{\mathsf{2}}}$
is the pullback (in the category of categories)
of the $1^{\mathrm{st}}$-projection of $\mathsf{E}^{\tilde{\mathsf{2}}}$
and the $0^{\mathrm{th}}$-projection of $\mathsf{M}^{\tilde{\mathsf{2}}}$.
There is a composition functor
$\circ_{\mathsf{C}} : \mathsf{E} \tilde{\odot} \mathsf{M} \rightarrow \mathsf{C}^{\tilde{\mathsf{2}}}$
that commutes with projections:
on objects $\circ_{\mathsf{C}}(A, e, C, m, B) = (A, e \circ_{\mathsf{C}} m, B)$,
and on morphisms $\circ_{\mathsf{C}}(a, c, b) = (a, b)$.

An $\langle \mathsf{E}, \mathsf{M} \rangle$-factorization system with choice has 
(1) a specified factorization for each $\mathsf{C}$-morphism,
and
(2) a specified diagonal for each diagonalization square in $\mathsf{C}$;
that is,
(1) there is a choice function 
from the class of $\mathsf{C}$-morphisms 
to the class of $\langle \mathsf{E}, \mathsf{M} \rangle$-factorizations, 
mapping each triple $(A,f,B)$ corresponding to a $\mathsf{C}$-morphism $f : A \rightarrow B$ 
to one of its factorizations $\gamma(A,f,B) = (A,e,C,m,B)$,
and
(2) there is a choice function 
from the class of diagonalization squares in $\mathsf{C}$
to the class of $\mathsf{C}$-morphisms, 
mapping each diagonalization square $(A,(e,C,s),(r,D,m),B)$ 
where $e \cdot s \equiv r \cdot m$ with $e \in \mathsf{E}$ and $m \in \mathsf{M}$,
to one of its diagonals 
$\gamma(A,(e,C,s),(r,D,m),B) = d$ where $d : C \rightarrow D$ is a $\mathsf{C}$-morphism
satisfying $e \cdot d \equiv r$ and $d \cdot m \equiv s$.
To reiterate,
diagonalization must be chosen in addition to factorization.
When choice is specified,
there is a factorization \underline{pseudo}-functor
$\div_{\mathsf{C}} 
: \mathsf{C}^{\tilde{\mathsf{2}}} \rightarrow \mathsf{E} \tilde{\odot} \mathsf{M}$,
which is defined on objects 
$\div_{\mathsf{C}}(A,f,B) = \gamma(A,f,B) = (A,e,C,m,B)$
as the chosen $\langle \mathsf{E}, \mathsf{M} \rangle$-factorization, 
and 
which is defined on morphisms 
$\div_{\mathsf{C}}\left((a,b):(A_1,f_1,B_1)\rightarrow(A_2,f_2,B_2)\right) 
= (a, c, b) : \div_{\mathsf{C}}(A_1,f_1,B_1) \rightarrow \div_{\mathsf{C}}(A_2,f_2,B_2)$
via the chosen diagonal $c = \gamma(A_1,(e_1,C_1,m_1{\cdot}b),(q{\cdot}e_2,C_2,m_2),B_2)$.

\emph{Show that $\div_{\mathsf{C}}$ is functorial.}

Clearly,
factorization followed by composition is the identity
$\div_{\mathsf{C}} \circ\, \circ_{\mathsf{C}} 
= \mathsf{id}_{\mathsf{C}}$.
By uniqueness of factorization (up to isomorphism)
composition followed by factorization is an isomorphism
$\circ_{\mathsf{C}} \,\circ \div_{\mathsf{C}} 
\cong \mathsf{id}_{\mathsf{E} \tilde{\odot} \mathsf{M}}$.

\begin{theorem} [General Equivalence] \label{pseudo-general-equivalence}
When an order-enriched category $\mathsf{C}$ has an $\langle \mathsf{E}, \mathsf{M} \rangle$-pseudo-factorization system with choice,
the $\mathsf{C}$-pseudo-arrow category
is equivalent (Figure~\ref{pseudo-factorization-equivalence}) to
the $\langle \mathsf{E}, \mathsf{M} \rangle$-pseudo-factorization category
\[\mathsf{C}^{\tilde{\mathsf{2}}} \equiv \mathsf{E} \tilde{\odot} \mathsf{M}.\]
\end{theorem}
This equivalence is mediated by factorization and composition.

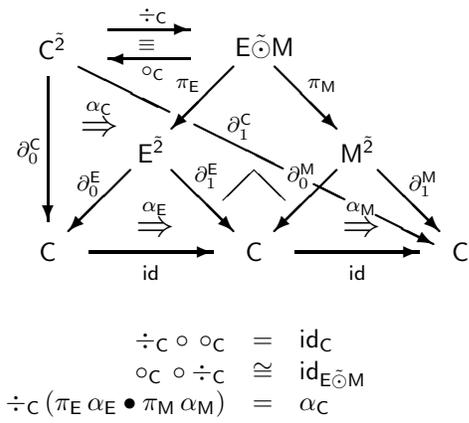
\begin{figure}
\begin{center}
\setlength{\unitlength}{0.78pt}
\begin{tabular}{c}
\begin{picture}(200,110)(-83,0)
\put(0,0){\begin{picture}(100,100)(0,0)
\put(5,75){\makebox(100,50){$\mathsf{E} \tilde{\odot} \mathsf{M}$}}
\put(-50,25){\makebox(100,50){$\mathsf{E}^{\tilde{\mathsf{2}}}$}}
\put(50,25){\makebox(100,50){$\mathsf{M}^{\tilde{\mathsf{2}}}$}}
\put(-100,-25){\makebox(100,50){$\mathsf{C}$}}
\put(0,-25){\makebox(100,50){$\mathsf{C}$}}
\put(100,-25){\makebox(100,50){$\mathsf{C}$}}
\put(-32,56){\makebox(100,50){\footnotesize{$\pi_\mathsf{E}$}}}
\put(33,56){\makebox(100,50){\footnotesize{$\pi_\mathsf{M}$}}}
\put(-80,8){\makebox(100,50){\footnotesize{$\partial_0^\mathsf{E}$}}}
\put(-23,12){\makebox(100,50){\footnotesize{$\partial_1^\mathsf{E}$}}}
\put(23,12){\makebox(100,50){\footnotesize{$\partial_0^\mathsf{M}$}}}
\put(82,8){\makebox(100,50){\footnotesize{$\partial_1^\mathsf{M}$}}}
\put(-50,-35){\makebox(100,50){\footnotesize{$\mathsf{id}$}}}
\put(50,-35){\makebox(100,50){\footnotesize{$\mathsf{id}$}}}
\put(-48,-4){\makebox(100,50){\footnotesize{$\alpha_{\mathsf{E}}$}}}
\put(-48,-14){\makebox(100,50){\Large{$\Rightarrow$}}}
\put(52,-4){\makebox(100,50){\footnotesize{$\alpha_{\mathsf{M}}$}}}
\put(52,-14){\makebox(100,50){\Large{$\Rightarrow$}}}
\put(50,25){\begin{picture}(30,15)(0,-15)
\put(0,0){\line(-1,-1){15}}
\put(0,0){\line(1,-1){15}}
\end{picture}}
\thicklines
\put(40,90){\vector(-1,-1){30}}
\put(60,90){\vector(1,-1){30}}
\put(-10,40){\vector(-1,-1){30}}
\put(10,40){\vector(1,-1){30}}
\put(-30,0){\vector(1,0){60}}
\put(90,40){\vector(-1,-1){30}}
\put(110,40){\vector(1,-1){30}}
\put(70,0){\vector(1,0){60}}
\end{picture}}
\thicklines
\put(-21,108){\vector(1,0){40}}
\put(-49,90){\makebox(100,50){\footnotesize{$\div_{\mathsf{C}}$}}}
\put(-52,75.5){\makebox(100,50){\footnotesize{$\equiv$}}}
\put(-49,62){\makebox(100,50){\footnotesize{$\circ_{\mathsf{C}}$}}}
\put(19,94){\vector(-1,0){40}}
\put(-35,90){\line(2,-1){44}}
\put(21,62){\line(2,-1){43}}
\put(79,33){\line(2,-1){15}}
\put(109,18){\vector(2,-1){28}}
\put(-50,85){\vector(0,-1){70}}
\put(-98,75){\makebox(100,50){$\mathsf{C}^{\tilde{\mathsf{2}}}$}}
\put(-109,24){\makebox(100,50){\footnotesize{$\partial_0^{\mathsf{C}}$}}}
\put(-7,36){\makebox(100,50){\footnotesize{$\partial_1^{\mathsf{C}}$}}}
\put(-75,45){\makebox(100,50){\footnotesize{$\alpha_{\mathsf{C}}$}}}
\put(-75,35){\makebox(100,50){\Large{$\Rightarrow$}}}
\end{picture}
\\ \\ \\
$\begin{array}{rcl}
\div_{\mathsf{C}} \circ\, \circ_{\mathsf{C}} & = & \mathsf{id}_{\mathsf{C}} 
\\
\circ_{\mathsf{C}} \,\circ \div_{\mathsf{C}} & \cong & \mathsf{id}_{\mathsf{E} \tilde{\odot} \mathsf{M}} 
\\
\div_{\mathsf{C}} 
\left(
\pi_\mathsf{E} \, \alpha_{\mathsf{E}}
\bullet
\pi_\mathsf{M} \, \alpha_{\mathsf{M}}
\right)
& = &
\alpha_{\mathsf{C}}
\end{array}$
\end{tabular}
\end{center}
\caption{Pseudo-factorization Equivalence}
\label{pseudo-factorization-equivalence}
\end{figure}


\section{Conceptual Structures Categories\label{sec:conceptual:structures:categories}}

An (abstract) conceptual structures (CS) category is a category 
in which an abstract version of the concept lattice construction can be defined.
\begin{definition}
An (abstract) conceptual structures (CS) category $\mathsf{C}$ is an order-enriched category that possesses finite limits.
\end{definition}
Recall that all finite limits are unique up to isomorphism,
and that the limiting cone for all finite limits is collectively monomorphic.
Here are some special finite limits.
There is a terminal object $1$ in $\mathsf{C}$.
For any two $\mathsf{C}$-objects $A_0,A_1 \in |\mathsf{C}|$,
there is a binary product $A_0 {\times} A_1$ with product projections $\pi_0 : A_0 {\times} A_1 \rightarrow A_0$ and $\pi_1 : A_0 {\times} A_1 \rightarrow A_1$.
For any parallel pair of $\mathsf{C}$-morphisms $f,g : A \rightarrow B$,
there is an equalizer $m : E \rightarrow A$.
For any opspan of $\mathsf{C}$-morphisms $f_0 : A_0 \rightarrow B \leftarrow A_1 : f_1$,
there is a pullback $A_0 {\times_{B}} A_1$ with pullback projections $\pi_0 : A_0 {\times} A_1 \rightarrow A_0$ and $\pi_1 : A_0 {\times} A_1 \rightarrow A_1$
satisfying $\pi_0 \circ f_0 = \pi_1 \circ f_1$.
For any parallel pair of $\mathsf{C}$-morphisms $f,g : A \rightarrow B$,
there is an equalizer $m : E \rightarrow A$.

In this section let $\mathsf{C}$ be a fixed CS category,
and assume that $g = \langle \check{g}, \hat{g} \rangle : A_0 \rightarrow A_1$ 
is a $\mathsf{C}$-adjunction with posetal source and target.

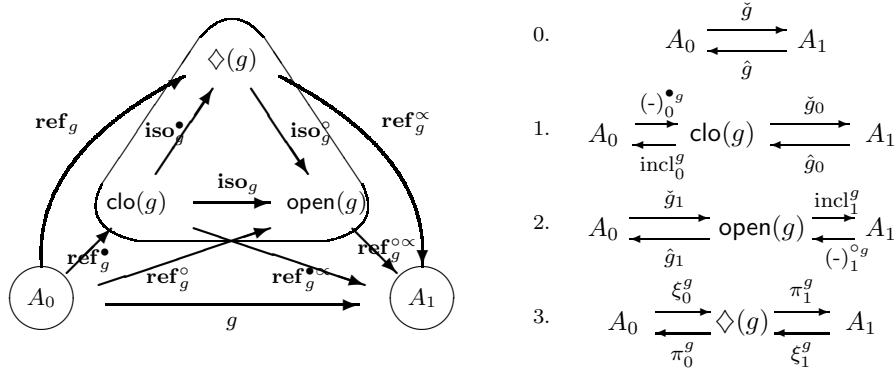
\begin{figure}
\begin{center}
\begin{tabular}{c@{\hspace{30pt}}c}
\begin{tabular}{c}
\setlength{\unitlength}{1.2pt}
\begin{picture}(120,70)(-30,-25)
\put(0,30){\makebox(60,30){\small{$\diamondsuit(g)$}}}
\put(-30,-15){\makebox(60,30){\small{$\mathsf{clo}(g)$}}}
\put(30,-15){\makebox(60,30){\small{$\mathsf{open}(g)$}}}
\put(-60,-45){\makebox(60,30){\small{$A_0$}}}
\put(60,-45){\makebox(60,30){\small{$A_1$}}}
\put(15,-53){\makebox(30,30){\footnotesize{$g$}}}
\put(-40,10){\makebox(30,30){\footnotesize{$\mathbf{ref}_{g}$}}}
\put(70,10){\makebox(30,30){\footnotesize{$\mathbf{ref}^\propto_{g}$}}}
\put(-30,-33){\makebox(30,30){\footnotesize{$\mathbf{ref}^\bullet_{g}$}}}
\put(64,-30){\makebox(30,30){\footnotesize{$\mathbf{ref}^{\circ \propto}_{g}$}}}
\put(-5,-39){\makebox(30,30){\footnotesize{$\mathbf{ref}^\circ_{g}$}}}
\put(37,-39){\makebox(30,30){\footnotesize{$\mathbf{ref}^{\bullet \propto}_{g}$}}}
\put(17,-9){\makebox(30,30){\footnotesize{$\mathbf{iso}_{g}$}}}
\put(-6,7){\makebox(30,30){\footnotesize{$\mathbf{iso}^\bullet_{g}$}}}
\put(40,7){\makebox(30,30){\footnotesize{$\mathbf{iso}^\circ_{g}$}}}
\put(-30,-30){\circle{20}}
\put(90,-30){\circle{20}}
\qbezier(-11,6)(-18,-8)(0,-12)
\qbezier(20,52.2)(30,64)(40,52.2)
\qbezier(71,6)(78,-8)(60,-12)
\put(-11,6){\line(2,3){31}}
\put(0,-12){\line(1,0){60}}
\put(71,6){\line(-2,3){31}}
\thicklines
\put(-10,-32){\vector(1,0){80}}
\put(6,9){\vector(2,3){18}}
\put(18,0){\vector(1,0){24}}
\put(36,36){\vector(2,-3){18}}
\put(-22,-22){\vector(1,1){14}}
\put(18,-8){\vector(3,-1){54}}
\put(68,-8){\vector(1,-1){14}}
\put(-12,-26){\vector(3,1){54}}
\qbezier(-30,-20)(-33,15)(15,40)
\put(15,40){\vector(3,2){0}}
\qbezier(90,-20)(93,15)(45,40)
\put(90,-20){\vector(0,-1){0}}
\end{picture}
\end{tabular}
&
\begin{tabular}{r@{\hspace{20pt}}c}
\footnotesize{0.}
&
\setlength{\unitlength}{1.0pt}
\begin{picture}(60,20)(0,0)
\put(-25,-15){\makebox(60,30){$A_0$}}
\put(25,-15){\makebox(60,30){$A_1$}}
\put(15,-4){\makebox(30,30){\footnotesize{$\check{g}$}}}
\put(15,-26){\makebox(30,30){\footnotesize{$\hat{g}$}}}
\put(15,4){\vector(1,0){30}}
\put(45,-4){\vector(-1,0){30}}
\end{picture}
\\ \\
\footnotesize{1.}
&
\begin{picture}(105,20)(0,0)
\put(-30,-15){\makebox(60,30){$A_0$}}
\put(15,-15){\makebox(60,30){$\mathsf{clo}(g)$}}
\put(75,-15){\makebox(60,30){$A_1$}}
\put(7.5,-3){\makebox(30,30){\footnotesize{$(\mbox{-})_0^{\bullet_{g}}$}}}
\put(7.5,-27){\makebox(30,30){\footnotesize{$\mathrm{incl}_0^{g}$}}}
\put(65,-4){\makebox(30,30){\footnotesize{$\check{g}_0$}}}
\put(65,-26){\makebox(30,30){\footnotesize{$\hat{g}_0$}}}
\put(12,4){\vector(1,0){16}}
\put(27,-4){\vector(-1,0){16}}
\put(63,4){\vector(1,0){30}}
\put(93,-4){\vector(-1,0){30}}
\end{picture}
\\ \\
\footnotesize{2.}
&
\begin{picture}(105,20)(0,0)
\put(-30,-15){\makebox(60,30){$A_0$}}
\put(30,-15){\makebox(60,30){$\mathsf{open}(g)$}}
\put(75,-15){\makebox(60,30){$A_1$}}
\put(73.5,-5){\makebox(30,30){\footnotesize{$\mathrm{incl}_1^{g}$}}}
\put(77.5,-26){\makebox(30,30){\footnotesize{$(\mbox{-})_1^{\circ_{g}}$}}}
\put(12,-4){\makebox(30,30){\footnotesize{$\check{g}_1$}}}
\put(12,-26){\makebox(30,30){\footnotesize{$\hat{g}_1$}}}
\put(10,4){\vector(1,0){30}}
\put(40,-4){\vector(-1,0){30}}
\put(79,4){\vector(1,0){16}}
\put(95,-4){\vector(-1,0){16}}
\end{picture}
\\ \\
\footnotesize{3.}
&
\begin{picture}(90,20)(0,0)
\put(-30,-15){\makebox(60,30){$A_0$}}
\put(15,-15){\makebox(60,30){$\diamondsuit(g)$}}
\put(60,-15){\makebox(60,30){$A_1$}}
\put(7.5,-3){\makebox(30,30){\footnotesize{$\xi_0^{g}$}}}
\put(7.5,-27){\makebox(30,30){\footnotesize{$\pi_0^{g}$}}}
\put(52.5,-3){\makebox(30,30){\footnotesize{$\pi_1^{g}$}}}
\put(52.5,-27){\makebox(30,30){\footnotesize{$\xi_1^{g}$}}}
\put(12,4){\vector(1,0){21}}
\put(33,-4){\vector(-1,0){21}}
\put(57,4){\vector(1,0){21}}
\put(78,-4){\vector(-1,0){21}}
\end{picture}
\\ \\
\end{tabular}
\end{tabular}
\end{center}
\caption{Factorizations of an Adjunction}
\label{adjunction-factorization}
\end{figure}

\subsection{Objects\label{subsec:objects:cs}}

Let $(\mathbf{A}_0,\mathbf{g},\mathbf{A}_1)$ be
an object in the arrow category $\mathsf{Ord}^{\mathsf{2}}$
consisting of two preorders and an adjunction 
$\mathbf{g} : \mathbf{A}_0 \rightleftharpoons \mathbf{A}_1$
between them.
The \emph{adjunction diagram} of $\mathbf{g}$ 
is the diagram in the category $\mathsf{Ord}$
consisting of the two preorders $\mathbf{A}_0$ and $\mathbf{A}_1$
and the two monotonic functions
$\check{\mathbf{g}} : \mathbf{A}_0 \rightarrow \mathbf{A}_1$
and $\hat{\mathbf{g}} : \mathbf{A}_1 \rightarrow \mathbf{A}_0$.
A \emph{cone} over the adjunction diagram of $\mathbf{g}$ consists of 
a vertex preorder $\mathbf{C}$ and two component monotonic functions
$\mathbf{c}_0 : \mathbf{C} \rightarrow \mathbf{A}_0$
and $\mathbf{c}_1 : \mathbf{C} \rightarrow \mathbf{A}_1$
that satisfy the bipolar pair constraints
$\mathbf{c}_0 \cdot \check{\mathbf{g}} = \mathbf{c}_1$
and $\mathbf{c}_1 \cdot \hat{\mathbf{g}} = \mathbf{c}_0$.
The \emph{axis} order $\diamondsuit(\mathbf{g})$
is the finite limit in $\mathsf{Ord}$ of the adjunction diagram of $\mathbf{g}$.
More precisely,
the axis order is the vertex preorder of a limiting cone for the adjunction diagram of $\mathbf{g}$.
It comes equipped with two projection monotonic functions
$\pi_0^{\mathbf{g}} : \diamondsuit(\mathbf{g}) \rightarrow \mathbf{A}_0$
and $\pi_1^{\mathbf{g}} : \diamondsuit(\mathbf{g}) \rightarrow \mathbf{A}_1$
that satisfy the bipolar pair constraints
$\pi_0^{\mathbf{g}} \cdot \check{\mathbf{g}} = \pi_1^{\mathbf{g}}$
and $\pi_1^{\mathbf{g}} \cdot \hat{\mathbf{g}} = \pi_0^{\mathbf{g}}$.
The limiting cone is optimal.
For any cone over the adjunction diagram of $\mathbf{g}$ as above,
there is a unique mediating monotonic function
$\mathbf{c} : \mathbf{C} \rightarrow \diamondsuit(\mathbf{g})$
satisfying the projection constraints
$\mathbf{c} \cdot \pi_0^{\mathbf{g}} = \mathbf{c}_0$
and
$\mathbf{c} \cdot \pi_1^{\mathbf{g}} = \mathbf{c}_1$.

Since the closure monotonic function
$(\mbox{-})^{\bullet_{\mathbf{g}}} : \mathbf{A}_0 \rightarrow \mathbf{A}_0$
and the left adjoint monotonic function $\check{\mathbf{g}} : \mathbf{A}_0 \rightarrow \mathbf{A}_1$
have a common source and satisfy the bipolar pair constraints
$(\mbox{-})^{\bullet_{\mathbf{g}}} \cdot \check{\mathbf{g}} = \check{\mathbf{g}}$
and $\check{\mathbf{g}} \cdot \hat{\mathbf{g}} = (\mbox{-})^{\bullet_{\mathbf{g}}}$,
they form a cone over the adjunction diagram of $\mathbf{g}$ 
called the \emph{source embedding cone}.
The \emph{source embedding} monotonic function 
$\xi_0^{\mathbf{g}} : \mathbf{A}_0 \rightarrow \diamondsuit(\mathbf{g})$
is the mediating monotonic function for the source embedding cone.
It satisfies the projection constraints
$\xi_0^{\mathbf{g}} \cdot \pi_0^{\mathbf{g}} 
= (\mbox{-})^{\bullet_{\mathbf{g}}}
\geq \mathrm{id}_{\mathbf{A}_0}$
and $\xi_0^{\mathbf{g}} \cdot \pi_1^{\mathbf{g}} = \check{\mathbf{g}}$.
Consider the composite monotonic function
$\pi_0^{\mathbf{g}} \cdot \xi_0^{\mathbf{g}} 
: \diamondsuit(\mathbf{g}) \rightarrow \diamondsuit(\mathbf{g})$.
Since this satisfies the projection constraints
$\pi_0^{\mathbf{g}} \cdot \xi_0^{\mathbf{g}} \cdot \pi_0^{\mathbf{g}}
= \pi_0^{\mathbf{g}} \cdot (\mbox{-})^{\bullet_{\mathbf{g}}}
= \pi_0^{\mathbf{g}} \cdot \check{\mathbf{g}} \cdot \hat{\mathbf{g}}
= \pi_1^{\mathbf{g}} \cdot \hat{\mathbf{g}}
= \pi_0^{\mathbf{g}}$
and
$\pi_0^{\mathbf{g}} \cdot \xi_0^{\mathbf{g}} \cdot \pi_1^{\mathbf{g}}
= \pi_0^{\mathbf{g}} \cdot \check{\mathbf{g}}
= \pi_1^{\mathbf{g}}$,
by uniqueness of limit mediators,
it is the identity
$\pi_0^{\mathbf{g}} \cdot \xi_0^{\mathbf{g}} = \mathrm{id}_{\diamondsuit(\mathbf{g})}$.
Source embedding and projection form the \emph{extent reflection}
$\mathsf{ref}_{\mathbf{g}} 
= \langle \xi_0^{\mathbf{g}}, \pi_0^{\mathbf{g}} \rangle
: \mathbf{A}_0 \rightleftharpoons \diamondsuit(\mathbf{g})$.

Dually,
since the right adjoint monotonic function $\hat{\mathbf{g}} : \mathbf{A}_1 \rightarrow \mathbf{A}_0$
and the interior monotonic function
$(\mbox{-})^{\circ_{\mathbf{g}}} : \mathbf{A}_1 \rightarrow \mathbf{A}_1$
have a common source and satisfy the bipolar pair constraints
$\hat{\mathbf{g}} \cdot \check{\mathbf{g}} = (\mbox{-})^{\circ_{\mathbf{g}}}$
and $(\mbox{-})^{\circ_{\mathbf{g}}} \cdot \hat{\mathbf{g}} = \hat{\mathbf{g}}$,
they form a cone over the adjunction diagram of $\mathbf{g}$ 
called the \emph{target embedding cone}.
The \emph{target embedding} monotonic function 
$\xi_1^{\mathbf{g}} : \mathbf{A}_1 \rightarrow \diamondsuit(\mathbf{g})$
is the mediating monotonic function for the target embedding cone.
It satisfies the projection constraints
$\xi_1^{\mathbf{g}} \cdot \pi_0^{\mathbf{g}} = \hat{\mathbf{g}}$
and $\xi_0^{\mathbf{g}} \cdot \pi_1^{\mathbf{g}} 
= (\mbox{-})^{\circ_{\mathbf{g}}}
\leq \mathrm{id}_{\mathbf{A}_1}$.
Consider the composite monotonic function
$\pi_1^{\mathbf{g}} \cdot \xi_1^{\mathbf{g}} 
: \diamondsuit(\mathbf{g}) \rightarrow \diamondsuit(\mathbf{g})$.
Since this satisfies the projection constraints
$\pi_1^{\mathbf{g}} \cdot \xi_1^{\mathbf{g}} \cdot \pi_0^{\mathbf{g}}
= \pi_1^{\mathbf{g}} \cdot \hat{\mathbf{g}}
= \pi_0^{\mathbf{g}}$
and
$\pi_1^{\mathbf{g}} \cdot \xi_1^{\mathbf{g}} \cdot \pi_1^{\mathbf{g}}
= \pi_1^{\mathbf{g}} \cdot (\mbox{-})^{\circ_{\mathbf{g}}}
= \pi_1^{\mathbf{g}} \cdot \hat{\mathbf{g}} \cdot \check{\mathbf{g}}
= \pi_0^{\mathbf{g}} \cdot \check{\mathbf{g}}
= \pi_1^{\mathbf{g}}$,
by uniqueness of limit mediators,
it is the identity
$\pi_1^{\mathbf{g}} \cdot \xi_1^{\mathbf{g}} = \mathrm{id}_{\diamondsuit(\mathbf{g})}$.
Target projection and embedding form the \emph{intent coreflection} 
$\mathsf{ref}_{\mathbf{g}}^{\propto} 
= \langle \pi_1^{\mathbf{g}}, \xi_1^{\mathbf{g}} \rangle
: \diamondsuit(\mathbf{g}) \rightleftharpoons \mathbf{A}_1$.

Composition of the extent reflection and the intent coreflection is the original adjunction,
$\mathsf{ref}_{\mathbf{g}} \circ \mathsf{ref}_{\mathbf{g}}^{\propto} 
= \mathbf{g}$. 
The \emph{polar factorization} of $(\mathbf{A}_0,\mathbf{g},\mathbf{A}_1)$
is the quintuple
$(\mathbf{A}_0,\mathsf{ref}_{\mathbf{g}},\diamondsuit(\mathbf{g}),\mathsf{ref}_{\mathbf{g}}^{\propto},\mathbf{A}_1)$
consisting of the $\mathsf{Rel}^{\mathsf{2}}$-object (reflection)
$(\mathbf{A}_0,\mathsf{ref}_{\mathbf{g}},\diamondsuit(\mathbf{g}))$ 
and the $\mathsf{Rel}^{{\propto}, \mathsf{2}}$-object (coreflection)
$(\diamondsuit(\mathbf{g}),\mathsf{ref}_{\mathbf{g}}^{\propto},\mathbf{A}_1)$.

\subsubsection{Existence of Factorization\label{subsubsec:existence:of:factorization:cs}}

\paragraph{The Closed Polar Factorization.}

Since closure is strictly idempotent 
$(\mbox{-})^{\bullet_{g}} \circ (\mbox{-})^{\bullet_{g}} 
= (\mbox{-})^{\bullet_{g}}$,
it forms a cone over the closure equalizer diagram consisting of 
$(\mbox{-})^{\bullet_{g}} : A_0 \rightarrow A_0$ 
and 
$(\mbox{-})^{\bullet_{g}} : A_0 \rightarrow A_0$. 
The closure target restriction 
$(\mbox{-})_0^{\bullet_{g}} : A_0 \rightarrow \mathsf{clo}(g)$ 
is defined to be the unique mediating $\mathsf{C}$-morphism for this cone,
and hence satisfies 
$(\mbox{-})_0^{\bullet_{g}} \cdot \mathrm{incl}_0^{g} 
= (\mbox{-})^{\bullet_{g}} \geq 1_{A_0}$.
Being an equalizer, 
$\mathrm{incl}_0^{g} : \mathsf{clo}(g) \rightarrow A_0$
is a monomorphic $\mathsf{C}$-monomorphism. 
Since 
$\mathrm{incl}_0^{g} 
= \mathrm{incl}_0^{g} \cdot (\mbox{-})^{\bullet_{g}} 
= \mathrm{incl}_0^{g} \cdot (\mbox{-})_0^{\bullet_{g}} \cdot \mathrm{incl}_0^{g}$,
by cancellation on the right ($\mathsf{C}$-monomorphism)
$1_{\mathsf{clo}(g)}
= \mathrm{incl}_0^{g} \cdot (\mbox{-})_0^{\bullet_{g}}$.
Hence,
there is a $\mathsf{C}$-reflection 
$\mathsf{ref}^\bullet_{g}
= (\mbox{-})_0^{\bullet_{g}} \dashv \mathrm{incl}_0^{g}
: A_0 \rightleftharpoons \mathsf{clo}(g)$ 
called the \emph{closed reflection} of $g$.
Define the left adjoint source restriction
$\check{g}_0 
\doteq \mathrm{incl}_0^{g} \circ \check{g}
: \mathsf{clo}(g) \rightarrow A_0 \rightarrow A_1$. 
The left adjoint morphism factors through this:
$\check{g}
= (\mbox{-})_0^{\bullet_{g}} \circ \check{g}_0$,
since
$(\mbox{-})_0^{\bullet_{g}} \cdot \check{g}_0
= (\mbox{-})_0^{\bullet_{g}} \cdot \mathrm{incl}_0^{g} \cdot \check{g}
= (\mbox{-})^{\bullet_{g}} \cdot \check{g}
= \check{g}$.
Define the right adjoint target restriction
$\hat{g}_0 
\doteq \hat{g} \cdot (\mbox{-})_0^{\bullet_{g}}
: A_1 \rightarrow A_0 \rightarrow \mathsf{clo}(g)$. 
The right adjoint morphism factors through this restriction:
$\hat{g}
= \hat{g}_0 \cdot \mathrm{incl}_0^{g}$,
since
$\hat{g}_0 \cdot \mathrm{incl}_0^{g}
= \hat{g} \cdot (\mbox{-})_0^{\bullet_{g}} \cdot \mathrm{incl}_0^{g}
= \hat{g} \cdot (\mbox{-})^{\bullet_{g}}
= \hat{g} \cdot \check{g} \cdot \hat{g}
= \hat{g}$.
Since
$\hat{g}_0 \cdot \check{g}_0 
= \hat{g} \cdot (\mbox{-})_0^{\bullet_{g}}
\cdot \mathrm{incl}_0^{g} \cdot \check{g}
= \hat{g} \cdot (\mbox{-})^{\bullet_{g}} \cdot \check{g}
= \hat{g} \cdot \check{g} \cdot \hat{g} \cdot \check{g}
= (\mbox{-})^{\circ_{g}} \cdot (\mbox{-})^{\circ_{g}}
= (\mbox{-})^{\circ_{g}}
\leq 1_{A_1}$
and
$\check{g}_0 \cdot \hat{g}_0
= \mathrm{incl}_0^{g} \cdot \check{g}
\cdot \hat{g} \cdot (\mbox{-})_0^{\bullet_{g}}
= \mathrm{incl}_0^{g} \cdot (\mbox{-})^{\bullet_{g}} \cdot (\mbox{-})_0^{\bullet_{g}}
= \mathrm{incl}_0^{g} \cdot (\mbox{-})_0^{\bullet_{g}}
= 1_{\mathsf{clo}(g)}$,
there is a $\mathsf{C}$-coreflection 
$\mathsf{ref}_{g}^{\bullet \propto}
= \check{g}_0 \dashv \hat{g}_0
: \mathsf{clo}(g) \rightleftharpoons B$ 
called the \emph{closed coreflection} of $g$.
The original adjunction factors in terms of its closed reflection and coreflection
$g 
= \mathsf{ref}_{g}^\bullet \circ \mathsf{ref}_{g}^{\bullet \propto}$
(Figure~\ref{adjunction-factorization}, row 1).
The \emph{closed polar factorization} of $g$ is the quintuple
$(A_0, \mathsf{ref}^{\circ}_{g}, \mathsf{open}(g), \mathsf{ref}^{\circ, \propto}_{g}, A_1)$.

\paragraph{The Open Polar Factorization.}

Since interior is strictly idempotent 
$(\mbox{-})^{\circ_{g}} \circ (\mbox{-})^{\circ_{g}} 
= (\mbox{-})^{\circ_{g}}$,
it forms a cone over the interior equalizer diagram consisting of 
$(\mbox{-})^{\circ_{g}} : A_1 \rightarrow A_1$ 
and 
$(\mbox{-})^{\circ_{g}} : A_1 \rightarrow A_1$. 
The interior target restriction 
$(\mbox{-})_1^{\circ_{g}} : A_1 \rightarrow \mathsf{open}(g)$ 
is defined to be the unique mediating $\mathsf{C}$-morphism for this cone,
and hence satisfies 
$(\mbox{-})_1^{\circ_{g}} \cdot \mathrm{incl}_1^{g} 
= (\mbox{-})^{\circ_{g}} \geq 1_{A_1}$.
Being an equalizer, 
$\mathrm{incl}_1^{g} : \mathsf{open}(g) \rightarrow A_1$
is a monomorphic $\mathsf{C}$-monomorphism. 
Since 
$\mathrm{incl}_1^{g} 
= \mathrm{incl}_1^{g} \cdot (\mbox{-})^{\circ_{g}} 
= \mathrm{incl}_1^{g} \cdot (\mbox{-})_1^{\circ_{g}} \cdot \mathrm{incl}_1^{g}$,
by cancellation on the right ($\mathsf{C}$-monomorphism)
$1_{\mathsf{open}(g)}
= \mathrm{incl}_1^{g} \cdot (\mbox{-})_1^{\circ_{g}}$.
Hence, 
there is a $\mathsf{C}$-coreflection
$\mathsf{ref}^{\circ \propto}_{g}
= \mathrm{incl}_1^{g} \dashv (\mbox{-})_1^{\circ_{g}}
: \mathsf{open}(g) \rightleftharpoons A_1$ 
called the \emph{open coreflection} of $g$
Define the right adjoint source restriction
$\hat{g}_1 
\doteq \mathrm{incl}_1^{g} \circ \hat{g}
: \mathsf{open}(g) \rightarrow A_1 \rightarrow A_0$. 
The right adjoint morphism factors through this:
$\hat{g}
= (\mbox{-})_1^{\circ_{g}} \circ \hat{g}_1$,
since
$(\mbox{-})_1^{\circ_{g}} \cdot \hat{g}_1
= (\mbox{-})_1^{\circ_{g}} \cdot \mathrm{incl}_1^{g} \cdot \hat{g}
= (\mbox{-})^{\circ_{g}} \cdot \hat{g}
= \hat{g}$.
Define the left adjoint target restriction
$\check{g}_1 
\doteq \check{g} \cdot (\mbox{-})_1^{\circ_{g}}
: A_0 \rightarrow A_1 \rightarrow \mathsf{open}(g)$. 
The left adjoint morphism factors through this restriction:
$\check{g}
= \check{g}_1 \cdot \mathrm{incl}_1^{g}$,
since
$\check{g}_1 \cdot \mathrm{incl}_1^{g}
= \check{g} \cdot (\mbox{-})_1^{\circ_{g}} \cdot \mathrm{incl}_1^{g}
= \check{g} \cdot (\mbox{-})^{\circ_{g}}
= \check{g} \cdot \hat{g} \cdot \check{g}
= \check{g}$.
Since
$\check{g}_1 \cdot \hat{g}_1 
= \check{g} \cdot (\mbox{-})_1^{\circ_{g}}
\cdot \mathrm{incl}_1^{g} \cdot \hat{g}
= \check{g} \cdot (\mbox{-})^{\circ_{g}} \cdot \hat{g}
= \check{g} \cdot \hat{g} \cdot \check{g} \cdot \hat{g}
= (\mbox{-})^{\bullet_{g}} \cdot (\mbox{-})^{\bullet_{g}}
= (\mbox{-})^{\bullet_{g}}
\geq 1_{A_0}$
and
$\hat{g}_1 \cdot \check{g}_1
= \mathrm{incl}_1^{g} \cdot \hat{g}
\cdot \check{g} \cdot (\mbox{-})_1^{\circ_{g}}
= \mathrm{incl}_1^{g} \cdot (\mbox{-})^{\circ_{g}} \cdot (\mbox{-})_1^{\circ_{g}}
= \mathrm{incl}_1^{g} \cdot (\mbox{-})_1^{\circ_{g}}
= 1_{\mathsf{open}(g)}$,
there is a $\mathsf{C}$-reflection
$\mathsf{ref}_{g}^\circ
= \check{g}_1 \dashv \hat{g}_1
: A_0 \rightleftharpoons \mathsf{open}(g)$ 
called the \emph{open reflection} of $g$
The original adjunction factors in terms of its open reflection and coreflection
$g 
= \mathsf{ref}_{g}^\circ \circ \mathsf{ref}_{g}^{\circ \propto}$
(Figure~\ref{adjunction-factorization}, row 2).
The \emph{open polar factorization} of $g$ is the quintuple
$(A_0, \mathsf{ref}^{\circ}_{g}, \mathsf{open}(g), \mathsf{ref}^{\circ, \propto}_{g}, A_1)$.

\paragraph{The Polar Factorization.}

Consider the axis diagram in $\mathsf{C}$
consisting of the two opspans
$\mathrm{incl}_0^{g} : \mathsf{clo}(g) \rightarrow A_0
\leftarrow \mathsf{open}(g) : \hat{g}_1$
and
$\check{g}_0 : \mathsf{clo}(g) \rightarrow A_1
\leftarrow \mathsf{open}(g) : \mathrm{incl}_1^{g}$.
The axis $\mathsf{C}$-object $\diamondsuit(g)$ is the pullback of this diagram.
It comes equipped with two pullback projections
$\tilde{\pi}_0^{g} : \diamondsuit(g) \rightarrow \mathsf{clo}(g)$
and
$\tilde{\pi}_1^{g} : \diamondsuit(g) \rightarrow \mathsf{open}(g)$.
These satisfy
$\tilde{\pi}_0^{g} \circ \mathrm{incl}_0^{g}
= \tilde{\pi}_1^{g} \circ \hat{g}_1$
and
$\tilde{\pi}_0^{g} \circ \check{g}_0
= \tilde{\pi}_1^{g} \circ \mathrm{incl}_1^{g}$.
Since
$\tilde{\pi}_0^{g} \circ \check{g}_{01} \circ \mathrm{incl}_1^{g}
= \tilde{\pi}_0^{g} \circ \check{g}_0
= \tilde{\pi}_1^{g} \circ \mathrm{incl}_1^{g}$
and $\mathrm{incl}_1^{g}$ is a $\mathbf{C}$-monomorphism,
$\tilde{\pi}_0^{g} \circ \check{g}_{01} = \tilde{\pi}_1^{g}$.
Dually,
$\tilde{\pi}_1^{g} \circ \hat{g}_{01} = \tilde{\pi}_0^{g}$.
Define the extended projections
$\pi_0^{g}
= \tilde{\pi}_0^{g} \circ \mathrm{incl}_0^{g} 
: \diamondsuit(g) \rightarrow \mathsf{clo}(g) \rightarrow A_0$
and
$\pi_1^{g}
= \tilde{\pi}_1^{g} \circ \mathrm{incl}_1^{g} 
: \diamondsuit(g) \rightarrow \mathsf{open}(g) \rightarrow A_1$.
The pair of $\mathsf{C}$-morphisms
$(\mbox{-})_0^{\bullet_{g}} : A_0 \rightarrow \mathsf{clo}(g)$
and
$\check{g}_1 : A_0 \rightarrow \mathsf{open}(g)$
forms a cone for the axis diagram,
since
$(\mbox{-})_0^{\bullet_{g}} \circ \mathrm{incl}_0^{g}
= (\mbox{-})^{\bullet_{g}}
= \check{g}_1 \circ \hat{g}_1$
and
$(\mbox{-})_0^{\bullet_{g}} \circ \check{g}_0
= \check{g}
= \check{g}_1 \circ \mathrm{incl}_1^{g}$.
Let
$\xi_0^{g} : A_0 \rightarrow \diamondsuit(g)$
denote the mediating $\mathsf{C}$-morphism for this cone;
so that
$\xi_0^{g}$ is the unique $\mathsf{C}$-morphism
such that
$\xi_0^{g} \circ \tilde{\pi}_0^{g}
= (\mbox{-})_0^{\bullet_{g}}$
and
$\xi_0^{g} \circ \tilde{\pi}_1^{g}
= \check{g}_1$.
Hence,
$\xi_0^{g} \circ \pi_0^{g}
= \xi_0^{g} \circ \tilde{\pi}_0^{g} \circ \mathrm{incl}_0^{g}
= (\mbox{-})_0^{\bullet_{g}} \circ \mathrm{incl}_0^{g}
= (\mbox{-})^{\bullet_{g}}$
and
$\xi_0^{g} \circ \pi_1^{g}
= \xi_0^{g} \circ \tilde{\pi}_1^{g} \circ \mathrm{incl}_1^{g}
= \check{g}_1 \circ \mathrm{incl}_1^{g}
= \check{g}$.
Since
$\pi_0^{g} \circ \xi_0^{g} \circ \tilde{\pi}_0^{g}
= \tilde{\pi}_0^{g} \circ 
\mathrm{incl}_0^{g} \circ (\mbox{-})_0^{\bullet_{g}}
= \tilde{\pi}_0^{g}$
and
$\pi_0^{g} \circ \xi_0^{g} \circ \tilde{\pi}_1^{g}
= \pi_0^{g} \circ \check{g}_1
= \tilde{\pi}_0^{g} \circ \mathrm{incl}_0^{g} \circ \check{g}_1
= \tilde{\pi}_0^{g} \circ \check{g}_{01}
= \tilde{\pi}_1^{g}$,
by uniqueness of mediating $\mathsf{C}$-morphisms
(in other words,
since limit projections are collectively monomorphic)
$\pi_0^{g} \circ \xi_0^{g}
= 1_{\diamondsuit(g)}$.
By direct calculation,
$\xi_0^{g} \circ \pi_0^{g}
= \xi_0^{g} \circ \tilde{\pi}_0^{g} \circ \mathrm{incl}_0^{g} 
= (\mbox{-})_0^{\bullet_{g}} \circ \mathrm{incl}_0^{g} 
= (\mbox{-})^{\bullet_{g}} \geq 1_{A_0}$.
Hence,
there is a $\mathsf{C}$-reflection
$\mathsf{ref}(g)
= \langle \xi_0^{g}, \pi_0^{g} \rangle
: A_0 \rightleftharpoons \diamondsuit(g)$
called the \emph{extent reflection} of $g$.
Dually,
the pair of $\mathsf{C}$-morphisms
$\hat{g}_0 : A_1 \rightarrow \mathsf{clo}(g)$
and
$(\mbox{-})_1^{\circ_{g}} : A_1 \rightarrow \mathsf{open}(g)$
is a cone for the axis diagram.
Let
$\xi_1^{g} : A_1 \rightarrow \diamondsuit(g)$
denote the unique mediating $\mathsf{C}$-morphism for this cone.
Hence,
there is a $\mathsf{C}$-coreflection
$\mathsf{ref}^\propto(g)
= \langle \pi_1^{g}, \xi_1^{g} \rangle
: \diamondsuit(g) \rightleftharpoons A_1$
called the \emph{intent reflection} of $g$.
The original $\mathsf{C}$-adjunction factors in terms of its extent reflection and intent coreflection
$g = \mathsf{ref}_{g} \circ \mathsf{ref}^\propto_{g}$
(Figure~\ref{adjunction-factorization}, row 3).
The quintuple
$( A_0, 
\mathsf{ref}_{g}, \diamondsuit(g), \mathsf{ref}^\propto_{g}, 
A_1 )$
is called the \emph{polar factorization} of $g$.
Since both source $A_0$ and target $A_1$ are posetal objects, 
the axis $\diamondsuit(g)$ is also a posetal object.

\paragraph{Linking the Factorizations.}

Consider the left adjoint source restriction
$\check{g}_0 : \mathsf{clo}(g) \rightarrow A_1$.
Since
$\check{g}_0 \circ (\mbox{-})^{\circ_{g}}
= \mathrm{incl}_0^{g} \circ \check{g} \circ \hat{g} \circ \check{g}
= \mathrm{incl}_0^{g} \circ \check{g}
= \check{g}_0
= \check{g}_0 \circ 1_{A_1}$,
there is a unique equalizer mediating $\mathsf{C}$-morphism
$\check{g}_{01} : \mathsf{clo}(g) \rightarrow \mathsf{open}(g)$
such that
$\check{g}_{01} \circ \mathrm{incl}_1^{g} = \check{g}_0$.
Call this the left adjoint source-target restriction.
The left adjoint target restriction factors through this 
$(\mbox{-})_0^{\bullet_{g}} \circ \check{g}_{01}
= \check{g}_1$,
since
$(\mbox{-})_0^{\bullet_{g}} \circ \check{g}_{01} \circ \mathrm{incl}_1^{g}
= (\mbox{-})_0^{\bullet_{g}} \circ \check{g}_0
= \check{g}
= \check{g}_1 \circ \mathrm{incl}_1^{g}$
and $\mathrm{incl}_1^{g}$ is a $\mathsf{C}$-monomorphism.
Consider the right adjoint source restriction
$\hat{g}_1 : \mathsf{open}(g) \rightarrow A_0$.
Since
$\hat{g}_1 \circ (\mbox{-})^{\bullet_{g}}
= \mathrm{incl}_1^{g} \circ \hat{g} \circ \check{g} \circ \hat{g}
= \mathrm{incl}_1^{g} \circ \hat{g}
= \hat{g}_1
= \hat{g}_1 \circ 1_{A_0}$,
there is a unique equalizer mediating $\mathsf{C}$-morphism
$\hat{g}_{01} : \mathsf{open}(g) \rightarrow \mathsf{clo}(g)$
such that
$\hat{g}_{01} \circ \mathrm{incl}_0^{g} = \hat{g}_1$.
Call this the right adjoint source-target restriction.
The right adjoint target restriction factors through this 
$(\mbox{-})_1^{\circ_{g}} \circ \hat{g}_{01}
= \hat{g}_0$,
since
$(\mbox{-})_1^{\circ_{g}} \circ \hat{g}_{01} \circ \mathrm{incl}_0^{g}
= (\mbox{-})_1^{\circ_{g}} \circ \hat{g}_1
= \hat{g}
= \hat{g}_0 \circ \mathrm{incl}_0^{g}$
and $\mathrm{incl}_0^{g}$ is a $\mathsf{C}$-monomorphism.
Since
$\check{g}_{01} \circ \hat{g}_{01} \circ \mathrm{incl}_0^{g}
= \check{g}_{01} \circ \hat{g}_1
= \check{g}_{01} \circ \mathrm{incl}_1^{g} \circ \hat{g}
= \check{g}_0 \circ \hat{g}
= \mathrm{incl}_0^{g} \circ \check{g} \circ \hat{g}
= \mathrm{incl}_0^{g} \circ (\mbox{-})^{\bullet_{g}}
= \mathrm{incl}_0^{g}
= 1_{\mathsf{clo}(g)} \circ \mathrm{incl}_0^{g}$
and
$\mathrm{incl}_0^{g}$ is a $\mathsf{C}$-monomorphism,
$\check{g}_{01} \circ \hat{g}_{01}
= 1_{\mathsf{clo}(g)}$.
Dually,
$\hat{g}_{01} \circ \check{g}_{01}
= 1_{\mathsf{open}(g)}$.
Hence,
there is a $\mathsf{C}$-isomorphism 
$\mathsf{iso}_{g}
= \check{g}_{01} \dashv \hat{g}_{01}
: \mathsf{clo}(g) \rightleftharpoons \mathsf{open}(g)$
called the \emph{central isomorphism} of $g$.
The closed and open reflections are isomorphic
$\mathsf{ref}_{g}^\bullet \circ \mathsf{iso}_{g}
= \mathsf{ref}_{g}^\circ$.
The closed and open coreflections are isomorphic
$\mathsf{ref}_{g}^{\bullet \propto} \circ \mathsf{iso}_{g}
= \mathsf{ref}_{g}^{\circ \propto}$.

The pair of $\mathsf{C}$-morphisms
$1_{\mathsf{clo}(g)} : \mathsf{clo}(g) \rightarrow \mathsf{clo}(g)$
and
$\check{g}_{01} : \mathsf{clo}(g) \rightarrow \mathsf{open}(g)$
forms a cone for the axis diagram,
since
$1_{\mathsf{clo}(g)} \circ \mathrm{incl}_0^{g}
= \mathrm{incl}_0^{g}
= \check{g}_{01} \circ \hat{g}_{01} \circ \mathrm{incl}_0^{g}
= \check{g}_{01} \circ \hat{g}_1$
and
$1_{\mathsf{clo}(g)} \circ \check{g}_0
= \check{g}_0
= \check{g}_{01} \circ \mathrm{incl}_1^{g}$.
Let
$\tilde{\xi}_0^{g} : \mathsf{clo}(g) \rightarrow \diamondsuit(g)$
denote the mediating $\mathsf{C}$-morphism for this cone;
so that
$\tilde{\xi}_0^{g}$ is the unique $\mathsf{C}$-morphism
such that
$\tilde{\xi}_0^{g} \circ \tilde{\pi}_0^{g}
= 1_{\mathsf{clo}(g)}$
and
$\tilde{\xi}_0^{g} \circ \tilde{\pi}_1^{g}
= \check{g}_{01}$.
Since
$\tilde{\pi}_0^{g} \circ \tilde{\xi}_0^{g} \circ \tilde{\pi}_0^{g}
= \tilde{\pi}_0^{g}$
and
$\tilde{\pi}_0^{g} \circ \tilde{\xi}_0^{g} \circ \tilde{\pi}_1^{g}
= \tilde{\pi}_0^{g} \circ \check{g}_{01}
= \tilde{\pi}_1^{g}$,
by uniqueness of mediating $\mathsf{C}$-morphisms
$\tilde{\pi}_0^{g} \circ \tilde{\xi}_0^{g}
= 1_{\diamondsuit(g)}$.
Hence,
there is a $\mathsf{C}$-isomorphism 
$\mathsf{iso}^\bullet_{g}
= \langle \tilde{\xi}_0^{g}, \tilde{\pi}_0^{g} \rangle
: \mathsf{clo}(g) \rightleftharpoons \diamondsuit(g)$
called the \emph{closed isomorphism} of $g$.
Since
$(\mbox{-})_0^{\bullet_{g}} \circ \tilde{\xi}_0^{g}
\circ \tilde{\pi}_0^{g}
= (\mbox{-})_0^{\bullet_{g}}
= \xi_0^{g} \circ \tilde{\pi}_0^{g}$
and
$(\mbox{-})_0^{\bullet_{g}} \circ \tilde{\xi}_0^{g}
\circ \tilde{\pi}_1^{g}
= (\mbox{-})_0^{\bullet_{g}} \circ \check{g}_{01}
= \check{g}_1
= \xi_0^{g} \circ \tilde{\pi}_1^{g}$,
by uniqueness of mediating $\mathsf{C}$-morphisms
$(\mbox{-})_0^{\bullet_{g}} \circ \tilde{\xi}_0^{g}
= \xi_0^{g}$.
Since in addition
$\tilde{\pi}_0^{g} \circ \mathrm{incl}_0^{g}
= \pi_0^{g}$,
the extent and closed reflections are isomorphic
$\mathsf{ref}_{g} 
= \mathsf{ref}^\bullet_{g} \circ \mathsf{iso}^\bullet_{g}$.
Dually,
since the pair of $\mathsf{C}$-morphisms
$\hat{g}_{01} : \mathsf{open}(g) \rightarrow \mathsf{clo}(g)$
and
$1_{\mathsf{open}(g)} : \mathsf{open}(g) \rightarrow \mathsf{open}(g)$
forms a cone for the axis diagram,
by defining
$\tilde{\xi}_1^{g} : \mathsf{open}(g) \rightarrow \diamondsuit(g)$
to be the mediating $\mathsf{C}$-morphism for this cone,
so that
$\tilde{\xi}_1^{g}$ is the unique $\mathsf{C}$-morphism
such that
$\tilde{\xi}_1^{g} \circ \tilde{\pi}_0^{g}
= \hat{g}_{01}$
and
$\tilde{\xi}_1^{g} \circ \tilde{\pi}_1^{g}
= 1_{\mathsf{open}(g)}$,
there is a $\mathsf{C}$-isomorphism 
$\mathsf{iso}^\circ_{g}
= \langle \tilde{\pi}_1^{g}, \tilde{\xi}_1^{g} \rangle
: \diamondsuit(g) \rightleftharpoons \mathsf{open}(g)$
called the \emph{open isomorphism} of $g$.
Furthermore,
the intent and open coreflections are isomorphic
$\mathsf{ref}^\propto_{g} 
= \mathsf{iso}^\circ_{g} \circ \mathsf{ref}^{\circ \propto}_{g}$.
Also,
the central isomorphism factors in terms of the closed and open isomorphisms
$\mathsf{iso}_{g} 
= \mathsf{iso}^\bullet_{g} \circ \mathsf{iso}^\circ_{g}$.

\begin{figure}
\begin{center}
\setlength{\unitlength}{0.9pt}
\begin{picture}(80,105)(0,0)
\put(20,70){\makebox(40,30){$\diamondsuit(g)$}}
\put(-20,30){\makebox(40,30){$\mathsf{clo}(g)$}}
\put(60,30){\makebox(40,30){$\mathsf{open}(g)$}}
\put(-20,-15){\makebox(40,30){$A_0$}}
\put(60,-15){\makebox(40,30){$A_1$}}
\qbezier(-9,8)(-13,13)(-15.75,18)
\qbezier(-20.25,27)(-35,65)(22,84)
\put(22.5,84){\vector(3,1){0}}
\qbezier(89,8)(93,13)(95.75,18)
\qbezier(101.25,29)(115,65)(58,84)
\put(57.5,84){\vector(-3,1){0}}
\put(-43,50){\makebox(40,30){\footnotesize{$\xi_0^{g}$}}}
\put(84.5,50){\makebox(40,30){\footnotesize{$\xi_1^{g}$}}}
\put(-10,55){\makebox(40,30){\footnotesize{$\tilde{\pi}_0^{g}$}}}
\put(50,55){\makebox(40,30){\footnotesize{$\tilde{\pi}_1^{g}$}}}
\put(24,41){\makebox(40,30){\footnotesize{$\check{g}_{01}$}}}
\put(20,30){\makebox(40,30){\footnotesize{$\cong$}}}
\put(24,19){\makebox(40,30){\footnotesize{$\hat{g}_{01}$}}}
\put(20,-3){\makebox(40,30){\footnotesize{$\check{g}$}}}
\put(20,-27){\makebox(40,30){\footnotesize{$\hat{g}$}}}
\put(4,11){\makebox(40,30){\footnotesize{$\check{g}_0$}}}
\put(37,11){\makebox(40,30){\footnotesize{$\hat{g}_1$}}}
\put(-35,7.5){\makebox(40,30){\footnotesize{$\mathrm{incl}_0^{g}$}}}
\put(75,7.5){\makebox(40,30){\footnotesize{$\mathrm{incl}_1^{g}$}}}
\put(20,50){\vector(1,0){40}}
\put(60,40){\vector(-1,0){40}}
\put(0,35){\vector(0,-1){25}}
\put(80,35){\vector(0,-1){25}}
\put(16,36){\vector(2,-1){52}}
\put(64,36){\vector(-2,-1){52}}
\put(25,75){\vector(-1,-1){20}}
\put(55,75){\vector(1,-1){20}}
\put(20,5){\vector(1,0){40}}
\put(60,-5){\vector(-1,0){40}}
\thinlines
\put(15,16){\line(0,1){15}}
\put(15,31){\line(-2,-1){12}}
\put(65,16){\line(0,1){15}}
\put(65,31){\line(2,-1){12}}
\end{picture}
\end{center}
\caption{The Axis of an Adjunction}
\label{axis}
\end{figure}
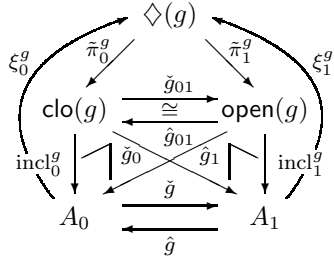

\subsubsection{Uniqueness of Factorization\label{subsubsec:uniqueness:of:factorization:cs}}

\begin{lemma} [Diagonalization]
Assume that we are given a commutative square
$e \circ s = r \circ m$
\begin{center}
\setlength{\unitlength}{0.85pt}
\begin{picture}(50,50)
\put(-50,25){\makebox(100,50){$A_0$}}
\put(0,25){\makebox(100,50){$B$}}
\put(-50,-25){\makebox(100,50){$C$}}
\put(0,-25){\makebox(100,50){$A_1$}}
\put(0,30){\makebox(50,50){\footnotesize{$e$}}}
\put(0,-31){\makebox(50,50){\footnotesize{$m$}}}
\put(-30,0){\makebox(50,50){\footnotesize{$r$}}}
\put(31,0){\makebox(50,50){\footnotesize{$s$}}}
\put(-5,5){\makebox(50,50){\footnotesize{$d$}}}
\thicklines
\put(12,0){\vector(1,0){26}}
\put(12,50){\vector(1,0){26}}
\put(0,38){\vector(0,-1){26}}
\put(50,38){\vector(0,-1){26}}
\put(38,38){\vector(-1,-1){26}}
\end{picture}
\end{center}
of adjunctions,
with reflection $e$ and coreflection $m$.
Then there is a unique adjunction 
$d : B \rightleftharpoons C$
with $e \circ d = r$ and $d \circ m = s$.
\end{lemma}

\noindent
{\bfseries Proof:} 
The necessary conditions give the definitions
$\check{d}
\doteq 
\check{s} \cdot \hat{m} 
= \hat{e} \cdot \check{r}$
and
$\hat{d}
\doteq 
\hat{r} \cdot \check{e} 
= \check{m} \cdot \hat{s}$.
Existence follows from these definitions.

In more detail,
the fundamental adjointness property,
the special conditions for (co) reflections
and the above commutative diagram,
resolve into the following identities and inequalities:
$\hat{e} \cdot \check{e} = 1_{B}$,
$1_{A_0} \leq \check{e} \cdot \hat{e}$,
$\hat{s} \cdot \check{s} \leq 1_{A_1}$,
$1_{B} \leq \check{s} \cdot \hat{s}$,
$\hat{r} \cdot \check{r} \leq 1_{C}$,
$1_{A_0} \leq \check{r} \cdot \hat{r}$,
$\hat{m} \cdot \check{m} \leq 1_{A_1}$,
$1_{C} = \check{m} \cdot \hat{m}$,
$\check{e} \cdot \check{s} =
\check{r} \cdot \check{m}$,
and
$\hat{m} \cdot \hat{r} =
\hat{s} \cdot \hat{e}$.
By suitable pre- and post-composition we can prove the identities:
$\check{e} \cdot \check{s} \cdot \hat{m} 
= \check{r}$,
$\check{m} \cdot \hat{s} \cdot \hat{e}
= \hat{r}$,
$\hat{m} \cdot \hat{r} \cdot \check{e} 
= \hat{s}$
and
$\hat{e} \cdot \check{r} \cdot \check{m} 
= \check{s}$,
(and then)
$\check{s} \cdot \hat{m} 
= \hat{e} \cdot \check{r}$ and
$\hat{r} \cdot \check{e} 
= \check{m} \cdot \hat{s}$.

[{\bfseries Existence}]
Define the $\mathsf{C}$-morphisms
\underline{$\check{d}
\doteq 
\check{s} \cdot \hat{m} 
= \hat{e} \cdot \check{r}$}
and
\underline{$\hat{d}
\doteq 
\hat{r} \cdot \check{e} 
= \check{m} \cdot \hat{s}$}.
The properties
$\hat{d} \cdot \check{d}
= \check{m} \cdot \hat{s} \cdot \hat{e} \cdot \check{r}
= \hat{r} \cdot \check{r}
\leq 1_{C}$ and
$\check{d} \cdot \hat{d}
= \check{s} \cdot \hat{m} \cdot \hat{r} \cdot \check{e}
= \check{s} \cdot \hat{s}
\geq 1_{B}$ 
show that
$d = \langle \check{d}, \hat{d} \rangle : B \rightleftharpoons C$
is a $\mathsf{C}$-adjunction.
The properties
$\check{d} \cdot \check{m}
= \hat{e} \cdot \check{r} \cdot \check{m}
= \check{s}$
and
$\hat{m} \cdot \hat{d}
= \hat{m} \cdot \hat{r} \cdot \check{e}
= \hat{s}$
show that
$d$
satisfies the required identity
$d \circ m = s$.
The properties
$\check{e} \cdot \check{d}
= \check{e} \cdot \check{s} \cdot \hat{m}
= \check{r}$
and
$\hat{d} \cdot \hat{e}
= \check{m} \cdot \hat{s} \cdot \hat{e}
= \hat{r}$
show that
$d$
satisfies the required identity
$e \circ d = r$.

[{\bfseries Uniqueness}]
Suppose $b = \langle \check{b}, \hat{b} \rangle
: B \rightleftharpoons C$
is another $\mathsf{C}$-adjunction satisfying the require identities
$e \circ b = r$
and $b \circ m = s$.
These identities resolve to the identities
$\check{e} \cdot \check{b} = \check{r}$,
$\hat{b} \cdot \hat{e} = \hat{r}$,
$\check{b} \cdot \check{m} = \check{s}$,
and
$\hat{m} \cdot \hat{b} = \hat{s}$.
Hence,
$\check{b}
= \check{b} \cdot \check{m} \cdot \hat{m}
= \check{s} \cdot \hat{m}
= \check{d}$,
$\hat{b}
= \hat{b} \cdot \hat{e} \cdot \check{e}
= \hat{r} \cdot \check{e}
= \hat{d}$
and thus
$b = d$.
\rule{2pt}{6pt}

\begin{lemma} [Polar Factorization]
The classes $\mathsf{Ref}(\mathsf{C})$ and $\mathsf{Ref}(\mathsf{C})^\propto$ of reflections and coreflections form a factorization system for $\mathsf{Adj}(\mathsf{C})_{=}$.
The (open, closed or full) polar factorization makes this a factorization system with choice.
\end{lemma}

\begin{proof} 
The previous discussion and lemma.
$\Box$
\end{proof}

\noindent
With this result,
we can specialize the discussion of subsection~\ref{subsec:factorization:systems}
to the case $\mathsf{C} = \mathsf{Adj}(\mathsf{C})_{=}$.
The arrow category $\mathsf{Adj}(\mathsf{C})_{=}^{\mathsf{2}}$
has adjunctions
$(A, g, B)$
as objects
and pairs of adjunctions
$(a, b) : (A_1, g_1, B_1) \rightarrow (A_2, g_2, B_2)$ 
forming a commutative diagram $a \circ g_2 = g_1 \circ b$ as morphisms.
The factorization category 
$\mathsf{Ref}(\mathsf{C}) \odot \mathsf{Ref}(\mathsf{C})^\propto$
has reflection-coreflection factorizations 
$(A, e, C, m, B)$
as objects
and triples of adjunctions
$(a, c, b) : (A_1, e_1, C_1, m_1, B_1) \rightarrow (A_2, e_2, C_2, m_2, B_2)$
forming commutative diagrams 
$a \circ e_2 = e_1 \circ c$ and $c \circ m_2 = m_1 \circ b$ as morphisms.
The polar factorization functor
$\div_{\mathsf{Adj}(\mathsf{C})_{=}} 
: \mathsf{Adj}(\mathsf{C})_{=}^{\mathsf{2}} \rightarrow
\mathsf{Ref}(\mathsf{C}) \odot \mathsf{Ref}(\mathsf{C})^\propto$
maps an adjunction $(A, g, B)$ 
to its polar factorization 
$\div_{\mathsf{Adj}(\mathsf{C})_{=}}(A, g, B)
= (A, \mathsf{ref}_{\mathsf{C}}(g), \diamondsuit(g), \mathsf{ref}_{\mathsf{C}}^\propto(g), B)$,
and maps a morphism of adjunctions
$(a, b) : (A_1, g_1, B_1) \rightarrow (A_2, g_2, B_2)$ 
to a morphism of polar factorizations
$\div_{\mathsf{Adj}_{=}}(a, b)
= (a, \diamondsuit_{(a, b)}, b) 
: \div_{\mathsf{Adj}(\mathsf{C})_{=}}(A_1, g_1, B_1)
\rightarrow 
\div_{\mathsf{Adj}(\mathsf{C})_{=}}(A_2, g_2, B_2)$,
where the axis adjunction
$\diamondsuit_{(a, b)}
: \diamondsuit(g_1) \rightleftharpoons \diamondsuit(g_2)$
is given by diagonalization of the commutative square
$\mathsf{ref}_{\mathsf{C}}(g_1) 
\circ
\left( \mathsf{ref}_{\mathsf{C}}^\propto(g_1) \circ b \right)
= 
\left( a \circ \mathsf{ref}_{\mathsf{C}}(g_2) \right)
\circ
\mathsf{ref}_{\mathsf{C}}^\propto(g_2)$.
The axis 
$\diamondsuit_{(a, b)}
= \langle \check{\diamondsuit}_{(a, b)}, \hat{\diamondsuit}_{(a, b)} \rangle$ 
is defined as follows.
\begin{center}
$\begin{array}{r@{\hspace{5pt}\doteq\hspace{5pt}}c@{\hspace{5pt}=\hspace{5pt}}c}
\check{\diamondsuit}_{(a, b)}
& \pi_1^{g_1} \cdot \check{b} \cdot \xi_1^{g_2}
& \pi_0^{g_1} \cdot \check{a} \cdot \xi_0^{g_2}
: \diamondsuit(g_1) \rightarrow \diamondsuit(g_2)
\\
\hat{\diamondsuit}_{(a, b)} 
& \pi_0^{g_2} \cdot \hat{a} \cdot \xi_0^{g_1} 
& \pi_1^{g_2} \cdot \hat{b} \cdot \xi_1^{g_1}
: \diamondsuit(g_2) \rightarrow \diamondsuit(g_1)
\end{array}$
\end{center}
Hence,
to compute either adjoint, 
first project to either source or target order,
next use the corresponding component adjoint,
and finally embed from the corresponding order.
Likewise,
the open polar factorization functor
$\div_{\mathsf{Adj}(\mathsf{C})_{=}}^{\circ} 
: \mathsf{Adj}(\mathsf{C})_{=}^{\mathsf{2}} \rightarrow
\mathsf{Ref}(\mathsf{C}) \odot \mathsf{Ref}(\mathsf{C})^\propto$
maps an adjunction 
to its open polar factorization,
and the closed polar factorization functor is dual.
These three functors are mutually isomorphic.

\begin{theorem} [Special Equivalence] \label{special-equivalence}
The $\mathsf{Adj}(\mathsf{C})_{=}$-arrow category
is equivalent (Fig.~\ref{conceptual-structure}) to
the $\langle \mathsf{Ref}(\mathsf{C}), \mathsf{Ref}(\mathsf{C})^\propto \rangle$-factorization category
\[\mathsf{Adj}(\mathsf{C})_{=}^{\mathsf{2}} \equiv \mathsf{Ref}(\mathsf{C}) \odot \mathsf{Ref}(\mathsf{C})^\propto.\]
\end{theorem}
This equivalence,
mediated by (open, closed or full) polar factorization and composition,
is a special case for adjunctions of the general equivalence
(Thm.~\ref{general-equivalence}).

\subsection{Morphisms\label{subsec:morphisms:cs}}

Let 
$\mathbf{f}
= (\mathbf{f}_0,\mathbf{f}_1) 
: (\mathbf{A}_0,\mathbf{g},\mathbf{A}_1) \rightarrow (\mathbf{B}_0,\mathbf{h},\mathbf{B}_1)$
be a morphism in the arrow category $\mathsf{Ord}^{\mathsf{2}}$.
This consists of four order adjunctions
$\mathbf{g} : \mathbf{A}_0 \rightleftharpoons \mathbf{A}_1$,
$\mathbf{h} : \mathbf{B}_0 \rightleftharpoons \mathbf{B}_1$,
$\mathbf{f}_0 : \mathbf{A}_0 \rightleftharpoons \mathbf{B}_0$ and
$\mathbf{f}_1 : \mathbf{A}_1 \rightleftharpoons \mathbf{B}_1$
which satisfy the commutative diagram
$\mathbf{g} \circ \mathbf{f}_1 = \mathbf{f}_0 \circ \mathbf{h}$.
This condition resolves into the constraints,
$\check{\mathbf{g}} \cdot \check{\mathbf{f}}_1 = \check{\mathbf{f}}_0 \cdot \check{\mathbf{h}}$,
which implies that 
$\check{\mathbf{f}}_1$ maps open elements of $\mathbf{A}_1$ to open elements of $\mathbf{B}_1$,
and $\hat{\mathbf{h}} \cdot \hat{\mathbf{f}}_0 = \hat{\mathbf{f}}_1 \cdot \hat{\mathbf{g}}$,
which implies that 
$\hat{\mathbf{f}}_0$ maps closed elements of $\mathbf{B}_0$ to closed elements of $\mathbf{A}_0$.

The monotonic functions
$\pi_0^{\mathbf{g}} \cdot \check{\mathbf{f}}_0 \cdot (\mbox{-})^{\bullet_{\mathbf{h}}}
: \diamondsuit(\mathbf{g}) \rightarrow \mathbf{B}_0$
and
$\pi_1^{\mathbf{g}} \cdot \check{\mathbf{f}}_1
: \diamondsuit(\mathbf{g}) \rightarrow \mathbf{B}_1$
form a cone over the target adjunction diagram,
satisfing the bipolar pair
$(\pi_0^{\mathbf{g}} \cdot \check{\mathbf{f}}_0 
\cdot (\mbox{-})^{\bullet_{\mathbf{h}}}) \cdot \check{\mathbf{h}}
= \pi_0^{\mathbf{g}} \cdot \check{\mathbf{f}}_0 \cdot \check{\mathbf{h}}
= \pi_0^{\mathbf{g}} \cdot \check{\mathbf{g}} \cdot \check{\mathbf{f}}_1
= \pi_1^{\mathbf{g}} \cdot \check{\mathbf{f}}_1$
and
$(\pi_1^{\mathbf{g}} \cdot \check{\mathbf{f}}_1) \cdot \hat{\mathbf{h}}
= \pi_0^{\mathbf{g}} \cdot \check{\mathbf{g}} 
\cdot \check{\mathbf{f}}_1 \cdot \hat{\mathbf{h}}
= \pi_0^{\mathbf{g}} \cdot \check{\mathbf{f}}_0 
\cdot \check{\mathbf{h}} \cdot \hat{\mathbf{h}}
= \pi_0^{\mathbf{g}} \cdot \check{\mathbf{f}}_0 \cdot (\mbox{-})^{\bullet_{\mathbf{h}}}$.
The mediating monotonic function
$\check{\diamondsuit}_{\mathbf{f}}
: \diamondsuit(\mathbf{g}) \rightarrow \diamondsuit(\mathbf{h})$
for this cone is called the \emph{left adjoint axis function}.
It satisfies the projection constraints
$\check{\diamondsuit}_{\mathbf{f}} \cdot \pi_0^{\mathbf{h}}
= \pi_0^{\mathbf{g}} \cdot \check{\mathbf{f}}_0 \cdot (\mbox{-})^{\bullet_{\mathbf{h}}}$
and $\check{\diamondsuit}_{\mathbf{f}} \cdot \pi_1^{\mathbf{h}}
= \pi_1^{\mathbf{g}} \cdot \check{\mathbf{f}}_1$.
Dually,
the monotonic functions
$\pi_0^{\mathbf{h}} \cdot \hat{\mathbf{f}}_0
: \diamondsuit(\mathbf{h}) \rightarrow \mathbf{A}_0$
and
$\pi_1^{\mathbf{h}} \cdot \hat{\mathbf{f}}_1 \cdot (\mbox{-})^{\circ_{\mathbf{g}}}
: \diamondsuit(\mathbf{h}) \rightarrow \mathbf{A}_1$
form a cone over the source adjunction diagram,
satisfing the bipolar pair
$(\pi_0^{\mathbf{h}} \cdot \hat{\mathbf{f}}_0) \cdot \check{\mathbf{g}}
= \pi_1^{\mathbf{h}} \cdot \hat{\mathbf{h}} 
\cdot \hat{\mathbf{f}}_0 \cdot \check{\mathbf{g}}
= \pi_1^{\mathbf{h}} \cdot \hat{\mathbf{f}}_1 
\cdot \hat{\mathbf{g}} \cdot \check{\mathbf{g}}
= \pi_1^{\mathbf{h}} \cdot \hat{\mathbf{f}}_1 \cdot (\mbox{-})^{\circ_{\mathbf{g}}}$
and
$(\pi_1^{\mathbf{h}} \cdot \hat{\mathbf{f}}_1 
\cdot (\mbox{-})^{\circ_{\mathbf{g}}}) \cdot \hat{\mathbf{g}}
= \pi_1^{\mathbf{h}} \cdot \hat{\mathbf{f}}_1 \cdot \hat{\mathbf{g}}
= \pi_1^{\mathbf{h}} \cdot \hat{\mathbf{h}} \cdot \hat{\mathbf{f}}_0
= \pi_0^{\mathbf{h}} \cdot \hat{\mathbf{f}}_0$.
The mediating monotonic function
$\hat{\diamondsuit}_{\mathbf{f}}
: \diamondsuit(\mathbf{h}) \rightarrow \diamondsuit(\mathbf{g})$
for this cone is called the \emph{right adjoint axis function}.
It satisfies the projection constraints
$\hat{\diamondsuit}_{\mathbf{f}} \cdot \pi_0^{\mathbf{g}}
= \pi_0^{\mathbf{h}} \cdot \hat{\mathbf{f}}_0$
and 
$\hat{\diamondsuit}_{\mathbf{f}} \cdot \pi_1^{\mathbf{g}}
= \pi_1^{\mathbf{h}} \cdot \hat{\mathbf{f}}_1 \cdot (\mbox{-})^{\circ_{\mathbf{g}}}$.

The composite monotonic function
$\check{\diamondsuit}_{\mathbf{f}} \cdot \hat{\diamondsuit}_{\mathbf{f}}
: \diamondsuit(\mathbf{g}) \rightarrow \diamondsuit(\mathbf{g})$
satisfies
$\check{\diamondsuit}_{\mathbf{f}} \cdot \hat{\diamondsuit}_{\mathbf{f}}
\geq \mathrm{id}_{\diamondsuit(\mathbf{g})}$,
since
$(\check{\diamondsuit}_{\mathbf{f}} \cdot \hat{\diamondsuit}_{\mathbf{f}}) \cdot \pi_0^{\mathbf{g}}
= \check{\diamondsuit}_{\mathbf{f}} \cdot \pi_0^{\mathbf{h}} \cdot \hat{\mathbf{f}}_0
= \pi_0^{\mathbf{g}} \cdot \check{\mathbf{f}}_0 \cdot (\mbox{-})^{\bullet_{\mathbf{h}}} \cdot \hat{\mathbf{f}}_0 
\geq \pi_0^{\mathbf{g}} \cdot \check{\mathbf{f}}_0 \cdot \hat{\mathbf{f}}_0 
\geq \pi_0^{\mathbf{g}}$
and
$(\check{\diamondsuit}_{\mathbf{f}} \cdot \hat{\diamondsuit}_{\mathbf{f}}) \cdot \pi_1^{\mathbf{g}}
= \check{\diamondsuit}_{\mathbf{f}} \cdot \pi_1^{\mathbf{h}} 
\cdot \hat{\mathbf{f}}_1 \cdot (\mbox{-})^{\circ_{\mathbf{g}}}
= \pi_1^{\mathbf{g}} \cdot \check{\mathbf{f}}_1
\cdot \hat{\mathbf{f}}_1 \cdot (\mbox{-})^{\circ_{\mathbf{g}}}
= \pi_0^{\mathbf{g}} \cdot \check{\mathbf{g}} \cdot \check{\mathbf{f}}_1
\cdot \hat{\mathbf{f}}_1 \cdot \hat{\mathbf{g}} \cdot \check{\mathbf{g}}
= \pi_0^{\mathbf{g}} \cdot \check{\mathbf{f}}_0 \cdot \check{\mathbf{h}}
\cdot \hat{\mathbf{h}} \cdot \hat{\mathbf{f}}_0 \cdot \check{\mathbf{g}}
\geq \pi_0^{\mathbf{g}} \cdot \check{\mathbf{g}}
= \pi_1^{\mathbf{g}}$.
Dually,
the composite monotonic function
$\hat{\diamondsuit}_{\mathbf{f}} \cdot \check{\diamondsuit}_{\mathbf{f}}
: \diamondsuit(\mathbf{h}) \rightarrow \diamondsuit(\mathbf{h})$
satisfies
$\hat{\diamondsuit}_{\mathbf{f}} \cdot \check{\diamondsuit}_{\mathbf{f}}
\leq \mathrm{id}_{\diamondsuit(\mathbf{h})}$,
since
$(\hat{\diamondsuit}_{\mathbf{f}} \cdot \check{\diamondsuit}_{\mathbf{f}}) \cdot \pi_0^{\mathbf{g}}
...
\leq \pi_0^{\mathbf{h}}$
and
$(\hat{\diamondsuit}_{\mathbf{f}} \cdot \check{\diamondsuit}_{\mathbf{f}}) \cdot \pi_1^{\mathbf{g}}
...
\leq \pi_1^{\mathbf{h}}$.
Hence,
the left and right adjoint axis monotonic functions form the \emph{axis adjunction}
$\diamondsuit_{\mathbf{f}} 
= \langle \check{\diamondsuit}_{\mathbf{f}}, \hat{\diamondsuit}_{\mathbf{f}} \rangle
: \diamondsuit(\mathbf{g}) \rightleftharpoons \diamondsuit(\mathbf{h})$.
Since
$(\xi_0^{\mathbf{g}} \cdot \check{\diamondsuit}_{\mathbf{f}}) \cdot \pi_0^{\mathbf{h}}
= \xi_0^{\mathbf{g}} \cdot \pi_0^{\mathbf{g}} 
\cdot \check{\mathbf{f}}_0 
\cdot (\mbox{-})^{\bullet_{\mathbf{h}}}
= (\mbox{-})^{\bullet_{\mathbf{g}}} 
\cdot \check{\mathbf{f}}_0 
\cdot (\mbox{-})^{\bullet_{\mathbf{h}}}
= (\mbox{-})^{\bullet_{\mathbf{g}}} 
\cdot \check{\mathbf{f}}_0 \cdot \check{\mathbf{h}} 
\cdot \hat{\mathbf{h}}
= (\mbox{-})^{\bullet_{\mathbf{g}}} \cdot \check{\mathbf{g}} 
\cdot \check{\mathbf{f}}_1 \cdot \hat{\mathbf{h}}
= \check{\mathbf{g}} \cdot \check{\mathbf{f}}_1 \cdot \hat{\mathbf{h}}
= \check{\mathbf{f}}_0 \cdot \check{\mathbf{h}} \cdot \hat{\mathbf{h}} 
= \check{\mathbf{f}}_0 \cdot (\mbox{-})^{\bullet_{\mathbf{h}}}
= (\check{\mathbf{f}}_0 \cdot \xi_0^{\mathbf{h}}) \cdot \pi_0^{\mathbf{h}}$
and
$(\xi_0^{\mathbf{g}} \cdot \check{\diamondsuit}_{\mathbf{f}}) \cdot \pi_1^{\mathbf{h}}
= \xi_0^{\mathbf{g}} \cdot \pi_1^{\mathbf{g}} \cdot \check{\mathbf{f}}_1
= \check{\mathbf{g}} \cdot \check{\mathbf{f}}_1
= \check{\mathbf{f}}_0 \cdot \check{\mathbf{h}}
= (\check{\mathbf{f}}_0 \cdot \xi_0^{\mathbf{h}}) \cdot \pi_1^{\mathbf{h}}$,
by uniqueness of limit mediators
$\xi_0^{\mathbf{g}} \cdot \check{\diamondsuit}_{\mathbf{f}}
= \check{\mathbf{f}}_0 \cdot \xi_0^{\mathbf{h}}$.
Since
we already know that
$\hat{\diamondsuit}_{\mathbf{f}} \cdot \pi_0^{\mathbf{g}}
= \pi_0^{\mathbf{h}} \cdot \hat{\mathbf{f}}_0$,
the axis adjunction satisfies the commutative diagram
$\mathsf{ref}_{\mathbf{g}} \circ \diamondsuit_{\mathbf{f}}
= \mathbf{f}_0 \circ \mathsf{ref}_{\mathbf{h}}$.
Since
$(\xi_1^{\mathbf{h}} \cdot \hat{\diamondsuit}_{\mathbf{f}}) \cdot \pi_1^{\mathbf{g}}
= \xi_1^{\mathbf{h}} \cdot \pi_1^{\mathbf{h}} 
\cdot \hat{\mathbf{f}}_1 
\cdot (\mbox{-})^{\circ_{\mathbf{g}}}
= (\mbox{-})^{\circ_{\mathbf{h}}} 
\cdot \hat{\mathbf{f}}_1 
\cdot (\mbox{-})^{\circ_{\mathbf{g}}}
= (\mbox{-})^{\circ_{\mathbf{h}}} 
\cdot \hat{\mathbf{f}}_1 \cdot \hat{\mathbf{g}} 
\cdot \check{\mathbf{g}}
= (\mbox{-})^{\circ_{\mathbf{h}}} \cdot \hat{\mathbf{h}} 
\cdot \hat{\mathbf{f}}_0 \cdot \check{\mathbf{g}}
= \hat{\mathbf{h}} \cdot \hat{\mathbf{f}}_0 \cdot \check{\mathbf{g}}
= \hat{\mathbf{f}}_1 \cdot \hat{\mathbf{g}} \cdot \check{\mathbf{g}} 
= \hat{\mathbf{f}}_1 \cdot (\mbox{-})^{\circ_{\mathbf{g}}}
= (\hat{\mathbf{f}}_1 \cdot \xi_1^{\mathbf{g}}) \cdot \pi_1^{\mathbf{g}}$
and
$(\xi_1^{\mathbf{h}} \cdot \hat{\diamondsuit}_{\mathbf{f}}) \cdot \pi_0^{\mathbf{g}}
= \xi_1^{\mathbf{h}} \cdot \pi_0^{\mathbf{h}} \cdot \hat{\mathbf{f}}_0
= \hat{\mathbf{h}} \cdot \hat{\mathbf{f}}_0
= \hat{\mathbf{f}}_1 \cdot \hat{\mathbf{g}}
= (\hat{\mathbf{f}}_1 \cdot \xi_1^{\mathbf{g}}) \cdot \pi_0^{\mathbf{g}}$,
by uniqueness of limit mediators
$\xi_1^{\mathbf{h}} \cdot \hat{\diamondsuit}_{\mathbf{f}}
= \hat{\mathbf{f}}_1 \cdot \xi_1^{\mathbf{g}}$.
Since we already know that
$\check{\diamondsuit}_{\mathbf{f}} \cdot \pi_1^{\mathbf{h}}
= \pi_1^{\mathbf{g}} \cdot \check{\mathbf{f}}_1$,
the axis adjunction satisfies the commutative diagram
$\diamondsuit_{\mathbf{f}} \circ \mathsf{ref}_{\mathbf{h}}^{\propto} 
= \mathsf{ref}_{\mathbf{g}}^{\propto} \circ \mathbf{f}_1$.

The polar factorization of the $\mathsf{Ord}^{\mathsf{2}}$-morphism
$\mathbf{f}
= (\mathbf{f}_0,\mathbf{f}_1) 
: (\mathbf{A}_0,\mathbf{g},\mathbf{A}_1) \rightarrow (\mathbf{B}_0,\mathbf{h},\mathbf{B}_1)$
is the triple
$(\mathbf{f}_0,\diamondsuit_{\mathbf{f}},\mathbf{f}_1)$
consisting of 
the $\mathsf{Rel}^{\mathsf{2}}$-morphism
$(\mathbf{f}_0,\diamondsuit_{\mathbf{f}}) 
: (\mathbf{A}_0,\mathsf{ref}_{\mathbf{g}},\diamondsuit(\mathbf{g})) 
\rightarrow (\mathbf{B}_0,\mathsf{ref}_{\mathbf{h}},\diamondsuit(\mathbf{h}))$
and
the $\mathsf{Rel}^{{\propto}, \mathsf{2}}$-morphism
$(\diamondsuit_{\mathbf{f}},\mathbf{f}_1) 
: (\diamondsuit(\mathbf{g}),\mathsf{ref}_{\mathbf{g}}^{\propto},\mathbf{A}_1) 
\rightarrow (\diamondsuit(\mathbf{h}),\mathsf{ref}_{\mathbf{h}}^{\propto},\mathbf{B}_1)$.

\begin{assumption}
Limit projections are collectively lax monomorphic:
if $\mathcal{D}$ is a diagram in $\mathsf{Ord}$
and $f,g : \mathbf{C} \rightarrow \mathrm{lim}(\mathcal{D})$
are two parallel monotonic functions that satisfy $f \cdot \pi_i \leq g \cdot \pi_i$
for all limit projections $\pi_i : \mathrm{lim}(\mathcal{D}) \rightarrow \mathcal{D}_i$,
then $f \leq g$.
\end{assumption}


\section{Lattice of Theories Categories\label{sec:lattice:of:theories:categories}}

We want to build and verify the structure in the Diamond Diagram of Figure~\ref{diamond-diagram}.
For this we need to define some additional adjunctions.

\subsection{Objects\label{subsec:objects:lot}}

\subsubsection{Existence of Factorization\label{subsubsec:existence:of:factorization:lot}}

\begin{definition}
An (abstract) lattice of theories (LOT) category $\mathsf{C}$ 
is a conceptual structures category 
that is the complete category for an order-enriched fibration.
\end{definition}

Consider the composition
$(\mbox{-})^{\bullet_{g}} \cdot_{\mathsf{E}} \check{g}
= \check{g}$
for any $\mathsf{E}$-adjunction 
$g= \langle \check{g}, \hat{g} \rangle
: A_0 \rightleftharpoons A_1$ with posetal target.
The \emph{closure lift} of $g$ is the cartesian $\mathsf{E}$-morphism
$(\check{\mbox{-}})^{\bullet_{g}}
\doteq \sharp_{(\mbox{-})^{\bullet_{g}} \cdot \flat_{\check{g}}}
: \Delta(\check{g}) \rightarrow \Delta(\check{g})$.
This morphism lifts the closure,
since
$\flat_{\check{g}} \cdot_{\mathsf{E}} (\check{\mbox{-}})^{\bullet_{g}}
= (\mbox{-})^{\bullet_{g}} \cdot_{\mathsf{E}} \flat_{\check{g}}$.
It is idempotent
$(\check{\mbox{-}})^{\bullet_{g}} \cdot_{\mathsf{E}} (\check{\mbox{-}})^{\bullet_{g}}
= (\check{\mbox{-}})^{\bullet_{g}}$.
Also,
it is equivalent to the identity
$(\check{\mbox{-}})^{\bullet_{g}} \equiv 1_{\Delta(\check{g})}$,
since
$(\check{\mbox{-}})^{\bullet_{g}} \cdot_{\mathsf{E}} \sharp_{\check{g}}
= (\check{\mbox{-}})^{\bullet_{g}}$.

Since the left adjoint factors as
$\check{g} 
= (\mbox{-})^{\bullet_{g}}_{0} \cdot_{\mathsf{E}} \check{g}_{0}$
and the closure morphism factors as
$(\mbox{-})^{\bullet_{g}} 
= (\mbox{-})^{\bullet_{g}}_{0} \cdot_{\mathsf{E}} \mathrm{incl}^{g}_{0}$,
where
$\check{g}_{0}$ and $\mathrm{incl}^{g}_{0}$ are cartesian,
we have equality of the apexes
$\Delta((\mbox{-})^{\bullet_{g}}) 
= \Delta((\mbox{-})^{\bullet_{g}}_{0})
= \Delta(\check{g})$,
equality of the gaps
$\flat_{(\mbox{-})^{\bullet_{g}}} 
= \flat_{(\mbox{-})^{\bullet_{g}}_{0}}
= \flat_{\check{g}} 
: A_0 \rightarrow \Delta(\check{g})$
and the lift identities
$\sharp_{\check{g}} 
= \sharp_{(\mbox{-})^{\bullet_{g}}_{0}} \cdot_{\mathsf{E}} \check{g}_{0}
: \Delta(\check{g}) \rightarrow A_1$
and
$\sharp_{(\mbox{-})^{\bullet_{g}}} 
= \sharp_{(\mbox{-})^{\bullet_{g}}_{0}} \cdot_{\mathsf{E}} \mathrm{incl}^{g}_{0}
: \Delta(\check{g}) \rightarrow A_0$.

\paragraph{Polar Factorization through $\mathsf{ref}_{g}^\bullet$
and $\mathsf{ref}_{g}^{\bullet \propto}$.}

For any $\mathsf{E}$-adjunction
$g 
= \langle \check{g}, \hat{g} \rangle 
: A_0 \rightleftharpoons A_1$,
the closed polar factorization 
$g 
= \mathsf{ref}_{g}^\bullet \circ \mathsf{ref}_{g}^{\bullet \propto}$
is expressed in terms of 
the closed polar reflection
$\mathsf{ref}_{g}^\bullet
= \langle (\mbox{-})^{\bullet_{g}}_{0}, \mathrm{incl}^{g}_{0} \rangle 
: A_0 \rightleftharpoons \mathsf{clo}(g)$
and the closed polar coreflection
$\mathsf{ref}_{g}^{\bullet \propto}
= \langle \check{g}_{0}, \hat{g}_{0} \rangle 
: \mathsf{clo}(g) \rightleftharpoons A_1$.
Hence,
we have the properties:
\begin{flushleft}
{\small \begin{tabular}{l@{\hspace{5pt}}l@{\hspace{5pt}}l}
\begin{tabular}[t]{l}
$1_{A_0} \leq \check{g} \cdot_{\mathsf{E}} \hat{g}$
\\
$\hat{g} \cdot_{\mathsf{E}} \check{g} \leq 1_{A_1}$
\end{tabular}
&
\begin{tabular}[t]{l}
$1_{A_0} \leq (\mbox{-})^{\bullet_{g}} = (\mbox{-})^{\bullet_{g}}_{0} \cdot_{\mathsf{E}} \mathrm{incl}^{g}_{0}$
\\
$\mathrm{incl}^{g}_{0} \cdot_{\mathsf{E}} (\mbox{-})^{\bullet_{g}}_{0}
= 1_{\mathsf{clo}(g)}$
\\
$\mathrm{incl}^{g}_{0}$ is an $\mathsf{E}$-monomorphism
\\
$(\mbox{-})^{\bullet_{g}}_{0}$ is an $\mathsf{E}$-epimorphism
\\
$\mathrm{incl}^{g}_{0}$ is a cartesian $\mathsf{E}$-morphism 
\end{tabular}
&
\begin{tabular}[t]{l}
$\check{g}_{0} \cdot_{\mathsf{E}} \hat{g}_{0} = 1_{\mathsf{clo}(g)}$
\\
$\hat{g}_{0} \cdot_{\mathsf{E}} \check{g}_{0}
= (\mbox{-})^{\circ_{g}} \leq 1_{A_1}$
\\
$\check{g}_{0}$ is an $\mathsf{E}$-monomorphism
\\
$\hat{g}_{0}$ is an $\mathsf{E}$-epimorphism
\\
$\check{g}_{0}$ is a cartesian $\mathsf{E}$-morphism 
\end{tabular}
\end{tabular}}
\end{flushleft}

\paragraph{The Closure Reflection $\mathsf{clo}_{g}$.}

Consider the pair of $\mathsf{E}$-morphisms
$\flat_{\check{g}} : A_0 \rightarrow \Delta(\check{g})$
and
$\sharp_{(\mbox{-})^{\bullet_{g}}} : \Delta(\check{g}) \rightarrow A_0$.
Composing in one direction,
get the inequality
$\flat_{\check{g}} \cdot_{\mathsf{E}} \sharp_{(\mbox{-})^{\bullet_{g}}}
= \flat_{\check{g}} \cdot_{\mathsf{E}} \sharp_{(\mbox{-})^{\bullet_{g}}_{0}} \cdot_{\mathsf{E}} \mathrm{incl}^{g}_{0}
= (\mbox{-})^{\bullet_{g}}_{0} \cdot_{\mathsf{E}} \mathrm{incl}^{g}_{0}
= (\mbox{-})^{\bullet_{g}}
\geq 1_{A_0}$.
Composing in the other direction,
the identity
$\flat_{\check{g}}
 \cdot_{\mathsf{E}} \sharp_{(\mbox{-})^{\bullet_{g}}}
 \cdot_{\mathsf{E}} \flat_{\check{g}}
 \cdot_{\mathsf{E}} \sharp_{\check{g}}
= \flat_{\check{g}}
 \cdot_{\mathsf{E}} \sharp_{(\mbox{-})^{\bullet_{g}}}
 \cdot_{\mathsf{E}} \check{g}
= (\mbox{-})^{\bullet_{g}}
 \cdot_{\mathsf{E}} \check{g}
= \check{g}
= \flat_{\check{g}} \cdot_{\mathsf{E}} \sharp_{\check{g}}$
implies by right cancellation the equivalence
$\flat_{\check{g}}
 \cdot_{\mathsf{E}} \sharp_{(\mbox{-})^{\bullet_{g}}}
 \cdot_{\mathsf{E}} \flat_{\check{g}}
\equiv \flat_{\check{g}}$,
which in turn implies by left cancellation the equivalence
$\sharp_{(\mbox{-})^{\bullet_{g}}} \cdot_{\mathsf{E}} \flat_{\check{g}}
\equiv 1_{\Delta(\check{g})}$.
Hence,
the pair forms the \emph{closure reflection}
$\mathsf{clo}_{g}
= \langle \flat_{\check{g}}, \sharp_{(\mbox{-})^{\bullet_{g}}} \rangle
: A_0 \rightleftharpoons \Delta(\check{g})$.

\paragraph{The Lift Coreflection ${\mathsf{lift}}_{g}$.}

Consider the pair of $\mathsf{E}$-morphisms
$\sharp_{\check{g}} : \Delta(\check{g}) \rightarrow A_1$
and
$\delta_{\hat{g},\check{g}} : A_1 \rightarrow \Delta(\check{g})$.
Composing in one direction,
get the inequality
$\delta_{\hat{g},\check{g}} \cdot_{\mathsf{E}} \sharp_{\check{g}}
= \hat{g} \cdot_{\mathsf{E}} \check{g}
= (\mbox{-})^{\circ_{g}}
\leq 1_{A_1}$.
Composing in the other direction,
the identity
$\flat_{\check{g}}
 \cdot_{\mathsf{E}} \sharp_{\check{g}}
 \cdot_{\mathsf{E}} \delta_{\hat{g},\check{g}}
 \cdot_{\mathsf{E}} \sharp_{\check{g}}
= \check{g} \cdot_{\mathsf{E}} \hat{g} \cdot_{\mathsf{E}} \check{g}
= \check{g}
= \flat_{\check{g}} \cdot_{\mathsf{E}} \sharp_{\check{g}}$
implies by right cancellation the equivalence
$\flat_{\check{g}}
 \cdot_{\mathsf{E}} \sharp_{\check{g}}
 \cdot_{\mathsf{E}} \delta_{\hat{g},\check{g}}
\equiv \flat_{\check{g}}$,
which in turn implies by left cancellation the equivalence
$\sharp_{\check{g}} \cdot_{\mathsf{E}} \delta_{\hat{g},\check{g}}
\equiv 1_{\Delta(\check{g})}$.
Hence,
the pair forms the \emph{lift coreflection}
$\mathsf{clo}_{g}
= \langle \sharp_{\check{g}}, \delta_{\hat{g},\check{g}} \rangle
: \Delta(\check{g}) \rightleftharpoons A_1$.

\paragraph{The Lifted Closure Equivalence $\mathsf{equ}_{g}^\bullet$.}

Consider the pair of $\mathsf{E}$-morphisms
$\sharp_{(\mbox{-})^{\bullet_{g}}_{0}} 
: \Delta(\check{g}) \rightarrow \mathsf{clo}(g)$
and
$\delta_{\mathrm{incl}^{g}_{0},(\mbox{-})^{\bullet_{g}}_{0}}
: \mathsf{clo}(g) \rightarrow \Delta(\check{g})$.
Composing in one direction,
get the identity
$\delta_{\mathrm{incl}^{g}_{0},(\mbox{-})^{\bullet_{g}}_{0}}
 \cdot_{\mathsf{E}} \sharp_{(\mbox{-})^{\bullet_{g}}_{0}} 
= \mathrm{incl}^{g}_{0} \cdot_{\mathsf{E}} (\mbox{-})^{\bullet_{g}}_{0}
= 1_{\mathsf{clo}(g)}$.
Composing in the other direction,
the identity
$\flat_{\check{g}}
 \cdot_{\mathsf{E}} \sharp_{(\mbox{-})^{\bullet_{g}}_{0}}
 \cdot_{\mathsf{E}} \delta_{\mathrm{incl}^{g}_{0},(\mbox{-})^{\bullet_{g}}_{0}}
 \cdot_{\mathsf{E}} \sharp_{\check{g}}
= \flat_{\check{g}}
 \cdot_{\mathsf{E}} \sharp_{(\mbox{-})^{\bullet_{g}}_{0}}
 \cdot_{\mathsf{E}} \delta_{\mathrm{incl}^{g}_{0},(\mbox{-})^{\bullet_{g}}_{0}}
 \cdot_{\mathsf{E}} \sharp_{(\mbox{-})^{\bullet_{g}}_{0}}
 \cdot_{\mathsf{E}} \check{g}_{0}
= (\mbox{-})^{\bullet_{g}}_{0}
 \cdot_{\mathsf{E}} \mathrm{incl}^{g}_{0}
 \cdot_{\mathsf{E}} (\mbox{-})^{\bullet_{g}}_{0}
 \cdot_{\mathsf{E}} \check{g}_{0}
= (\mbox{-})^{\bullet_{g}}_{0}
 \cdot_{\mathsf{E}} \check{g}_{0}
= \check{g}
= \flat_{\check{g}}
 \cdot_{\mathsf{E}} \sharp_{\check{g}}$
implies by right cancellation the equivalence
$\flat_{\check{g}}
 \cdot_{\mathsf{E}} \sharp_{(\mbox{-})^{\bullet_{g}}_{0}}
 \cdot_{\mathsf{E}} \delta_{\mathrm{incl}^{g}_{0},(\mbox{-})^{\bullet_{g}}_{0}}
\equiv \flat_{\check{g}}$,
which in turn implies by left cancellation the equivalence
$\sharp_{(\mbox{-})^{\bullet_{g}}_{0}} \cdot_{\mathsf{E}} \delta_{\mathrm{incl}^{g}_{0},(\mbox{-})^{\bullet_{g}}_{0}}
\equiv 1_{\Delta(\check{g})}$.
Hence,
the pair forms the \emph{lifted closure equivalence}
$\mathsf{equ}_{g}^\bullet
= \langle \sharp_{(\mbox{-})^{\bullet_{g}}_{0}}, \delta_{\mathrm{incl}^{g}_{0},(\mbox{-})^{\bullet_{g}}_{0}} \rangle
: \Delta(\check{g}) \rightleftharpoons \mathsf{clo}(g)$.

\paragraph{Composing the Adjunctions.}

We have already seen that the original adjunction
is the composition of the closed polar reflection and the closed polar coreflection
\[g
= \mathsf{ref}_{g}^\bullet \circ \mathsf{ref}_{g}^{\bullet \propto}
: A_0 \rightleftharpoons \mathsf{clo}({g}) \rightleftharpoons A_1.\]
The identities
$\flat_{\check{g}}
 \cdot_{\mathsf{E}} \sharp_{(\mbox{-})^{\bullet_{g}}_{0}} 
= \flat_{(\mbox{-})^{\bullet_{g}}_{0}} 
 \cdot_{\mathsf{E}} \sharp_{(\mbox{-})^{\bullet_{g}}_{0}} 
= (\mbox{-})^{\bullet_{g}}_{0}$
and
$\delta_{\mathrm{incl}^{g}_{0},(\mbox{-})^{\bullet_{g}}_{0}}
 \cdot_{\mathsf{E}} 
\sharp_{(\mbox{-})^{\bullet_{g}}} 
= \delta_{\mathrm{incl}^{g}_{0},(\mbox{-})^{\bullet_{g}}_{0}}
 \cdot_{\mathsf{E}} \sharp_{(\mbox{-})^{\bullet_{g}}_{0}} 
 \cdot_{\mathsf{E}} \mathrm{incl}^{g}_{0}
= \mathrm{incl}^{g}_{0}
 \cdot_{\mathsf{E}} (\mbox{-})^{\bullet_{g}}_{0}
 \cdot_{\mathsf{E}} \mathrm{incl}^{g}_{0}
= \mathrm{incl}^{g}_{0}$
show that 
the closed polar reflection is the composition of the closure reflection and the lifted closure equivalence
\[\mathsf{ref}_{g}^\bullet
= \mathsf{clo}_{g} \circ {\mathsf{equ}}_{g}^\bullet
: A_0 \rightleftharpoons \Delta(\check{g}) \rightleftharpoons \mathsf{clo}({g}).\]
The identities
$\sharp_{(\mbox{-})^{\bullet_{g}}_{0}} 
 \cdot_{\mathsf{E}} \check{g}_{0}
= \sharp_{\check{g}}$
and
$\hat{g}_{0}
 \cdot_{\mathsf{E}} \delta_{\mathrm{incl}^{g}_{0},(\mbox{-})^{\bullet_{g}}_{0}}
= \hat{g}
 \cdot_{\mathsf{E}} (\mbox{-})^{\bullet_{g}}_{0}
 \cdot_{\mathsf{E}} \mathrm{incl}^{g}_{0}
 \cdot_{\mathsf{E}} \flat_{\check{g}}
= \hat{g}
 \cdot_{\mathsf{E}} (\mbox{-})^{\bullet_{g}}
 \cdot_{\mathsf{E}} \flat_{\check{g}}
= \hat{g}
 \cdot_{\mathsf{E}} \flat_{\check{g}}
= \delta_{\hat{g},\check{g}}$
show that 
the the lift coreflection is the composition of the lifted closure equivalence and the closed polar coreflection
\[{\mathsf{lift}}_{g}
= {\mathsf{equ}}_{g}^\bullet \circ \mathsf{ref}_{g}^{\bullet \propto}
: \Delta(\check{g}) \rightleftharpoons \mathsf{clo}({g}) \rightleftharpoons A_1.\]
The identities
$\delta_{\hat{g},\check{g}}
 \cdot_{\mathsf{E}} \sharp_{(\mbox{-})^{\bullet_{g}}}
= \delta_{\hat{g},\check{g}}
 \cdot_{\mathsf{E}} \sharp_{(\mbox{-})^{\bullet_{g}}_{0}}
 \cdot_{\mathsf{E}} \mathrm{incl}^{g}_{0}
= \delta_{\hat{g},\check{g}}
 \cdot_{\mathsf{E}} \sharp_{(\mbox{-})^{\bullet_{g}}_{0}}
 \cdot_{\mathsf{E}} \mathrm{incl}^{g}_{0}
 \cdot_{\mathsf{E}} (\mbox{-})^{\bullet_{g}}
= \delta_{\hat{g},\check{g}}
 \cdot_{\mathsf{E}} \sharp_{(\mbox{-})^{\bullet_{g}}_{0}}
 \cdot_{\mathsf{E}} \mathrm{incl}^{g}_{0}
 \cdot_{\mathsf{E}} \check{g}
 \cdot_{\mathsf{E}} \hat{g}
= \delta_{\hat{g},\check{g}}
 \cdot_{\mathsf{E}} \sharp_{(\mbox{-})^{\bullet_{g}}_{0}}
 \cdot_{\mathsf{E}} \check{g}_{0}
 \cdot_{\mathsf{E}} \hat{g}
= \delta_{\hat{g},\check{g}}
 \cdot_{\mathsf{E}} \sharp_{\check{g}}
 \cdot_{\mathsf{E}} \hat{g}
= (\mbox{-})^{\circ_{g}} \cdot_{\mathsf{E}} \hat{g}
= \hat{g}$
and
$\flat_{\check{g}} \cdot_{\mathsf{E}} \sharp_{\check{g}}
= \check{g}$
show that 
the original adjunction is the composition of the closure reflection and the lift coreflection
\[g
= \mathsf{clo}_{g} \circ {\mathsf{lift}}_{g}
: A_0 \rightleftharpoons \Delta(\check{g}) \rightleftharpoons A_1.\]
But this can also be computed by adjunction composition
$g
= \mathsf{ref}_{g}^\bullet
 \circ \mathsf{ref}_{g}^{\bullet \propto}
= \mathsf{clo}_{g}
 \circ {\mathsf{equ}}_{g}^\bullet 
 \circ \mathsf{ref}_{g}^{\bullet \propto}
= \mathsf{clo}_{g}
 \circ {\mathsf{lift}}_{g}$.

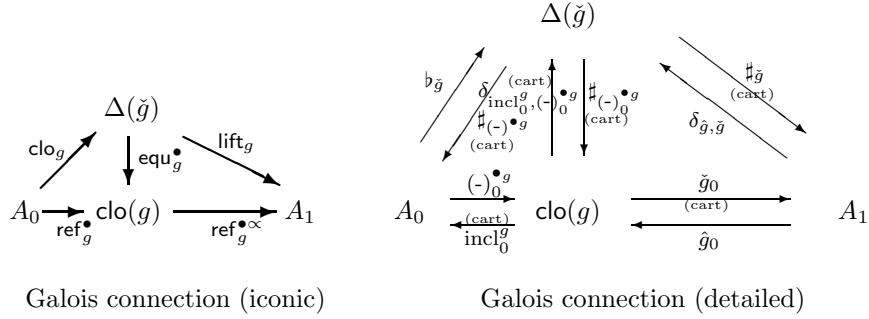
\begin{figure}
\begin{center}
\begin{tabular}{c@{\hspace{25pt}}c}

\setlength{\unitlength}{1.3pt}
\begin{picture}(90,30)(0,30)
\put(-30,15){\makebox(60,30){$A_0$}}
\put(0,15){\makebox(60,30){$\mathsf{clo}(g)$}}
\put(50,15){\makebox(60,30){$A_1$}}
\put(1,45){\makebox(60,30){$\Delta(\check{g})$}}
\put(9,9){\makebox(30,30)[l]{\footnotesize{$\mathsf{ref}_{g}^\bullet$}}}
\put(54,9){\makebox(30,30)[l]{\footnotesize{$\mathsf{ref}_{g}^{\bullet \propto}$}}}
\put(-18,33){\makebox(30,30)[r]{\footnotesize{$\mathsf{clo}_{g}$}}}
\put(56,35){\makebox(30,30)[l]{\footnotesize{${\mathsf{lift}}_{g}$}}}
\put(33,30){\makebox(30,30)[l]{\footnotesize{${\mathsf{equ}}_{g}^\bullet$}}}
\thicklines
\put(5,30){\vector(1,0){12}}
\put(43,30){\vector(1,0){30}}
\put(4.5,37){\vector(1,1){16}}
\put(30,52){\vector(0,-1){14}}
\put(46,52){\vector(2,-1){28}}
\end{picture}
&
\setlength{\unitlength}{1.2pt}
\begin{picture}(150,75)
\put(-50,-25){\makebox(100,50){$A_0$}}
\put(0,-25){\makebox(100,50){$\mathsf{clo}(g)$}}
\put(90,-25){\makebox(100,50){$A_1$}}
\put(0,37){\makebox(100,50){$\Delta(\check{g})$}}

\put(13,-3.5){\makebox(25,25){\footnotesize{$(\mbox{-})^{\bullet_{g}}_{0}$}}}
\put(12,-15){\makebox(25,25){\tiny{$(\mathrm{cart})$}}}
\put(12,-21.5){\makebox(25,25){\footnotesize{$\mathrm{incl}^{g}_{0}$}}}
\put(13,4){\vector(1,0){20}}
\put(33,-4){\vector(-1,0){20}}

\put(82,-3.5){\makebox(25,25){\footnotesize{$\check{g}_{0}$}}}
\put(81,-10.5){\makebox(25,25){\tiny{$(\mathrm{cart})$}}}
\put(82,-21.5){\makebox(25,25){\footnotesize{$\hat{g}_{0}$}}}
\put(70,4){\vector(1,0){50}}
\put(120,-4){\vector(-1,0){50}}

\put(-2,6){\begin{picture}(40,60)
\put(-5,25){\makebox(30,20){\footnotesize{$\flat_{\check{g}}$}}}
\put(17,12){\makebox(30,20){\footnotesize{$\sharp_{(\mbox{-})^{\bullet_{g}}}$}}}
\put(13,5){\makebox(30,20){\tiny{$(\mathrm{cart})$}}}
\put(5.5,16){\vector(2,3){20}}
\put(32.5,40){\vector(-2,-3){20}}
\end{picture}}
\put(35,6){\begin{picture}(30,60)
\put(-7,20){\makebox(20,20){\footnotesize{$\delta_{\mathrm{incl}^{g}_{0},(\mbox{-})^{\bullet_{g}}_{0}}$}}}
\put(-7,25){\makebox(20,20){\tiny{$(\mathrm{cart})$}}}
\put(20,20){\makebox(20,20){\footnotesize{$\sharp_{(\mbox{-})^{\bullet_{g}}_{0}}$}}}
\put(17,13){\makebox(20,20){\tiny{$(\mathrm{cart})$}}}
\put(10,12){\vector(0,1){30}}
\put(20,42){\vector(0,-1){30}}
\end{picture}}
\put(62,6){\begin{picture}(80,60)
\put(0,0){\makebox(80,60){}}
\put(35,22){\makebox(24,32){\footnotesize{$\sharp_{\check{g}}$}}}
\put(34,15){\makebox(24,32){\tiny{$(\mathrm{cart})$}}}
\put(20,6){\makebox(24,32){\footnotesize{$\delta_{\hat{g}, \check{g}}$}}}
\put(23,49){\vector(4,-3){40}}
\put(57,11){\vector(-4,3){40}}
\end{picture}}
\end{picture}
\\ \\ \\
Galois connection (iconic) & Galois connection (detailed)
\end{tabular}
\end{center}
\caption{Closure-Lift Factorization}
\label{closure-lift-factorization}
\end{figure}

\begin{figure}
\begin{center}
\begin{tabular}{c@{\hspace{25pt}}c}
\setlength{\unitlength}{1.3pt}
\begin{picture}(60,60)(0,-15)
\put(-30,15){\makebox(60,30){$A_0$}}
\put(1,15){\makebox(60,30){$\diamondsuit(g)$}}
\put(30,15){\makebox(60,30){$A_1$}}
\put(1,45){\makebox(60,30){$\mathsf{ker}(\check{g})$}}
\put(1,-15){\makebox(60,30){$\mathsf{ker}(\hat{g})$}}
\put(-18,33){\makebox(30,30)[r]{\footnotesize{$\mathsf{clo}_{g}$}}}
\put(49,33){\makebox(30,30)[l]{\footnotesize{${\mathsf{lift}}_{g}^{0}$}}}
\put(-18,-3){\makebox(30,30)[r]{\footnotesize{${\mathsf{lift}}_{g}^{1}$}}}
\put(50,-3){\makebox(30,30)[l]{\footnotesize{$\mathsf{int}_{g}$}}}
\put(9,20){\makebox(30,30)[l]{\footnotesize{$\mathsf{ref}_{g}$}}}
\put(40,20){\makebox(30,30)[l]{\footnotesize{$\mathsf{ref}_{g}^{\propto}$}}}
\put(32.5,30){\makebox(30,30)[l]{\footnotesize{${\mathsf{equ}}_{g}$}}}
\put(16,1){\makebox(30,30)[l]{\footnotesize{${\mathsf{equ}}_{g}^{\propto}$}}}
\thicklines
\put(7,30){\vector(1,0){12}}
\put(41,30){\vector(1,0){12}}
\put(30,52){\vector(0,-1){14}}
\put(30,22){\vector(0,-1){14}}
\put(4.5,37){\vector(1,1){16}}
\put(39.5,53){\vector(1,-1){16}}
\put(4.5,23){\vector(1,-1){16}}
\put(39.5,7){\vector(1,1){16}}
\end{picture}
&
\setlength{\unitlength}{1.2pt}
\begin{picture}(100,100)
\put(-8,42){\begin{picture}(16,16)
\put(2,8){\vector(1,0){0}}
\put(0,0){\oval(16,16)[tl]}
\put(0,0){\oval(16,16)[b]}
\put(-16,-7){\makebox(0,0){\footnotesize{$(\mbox{-})^{\bullet_{g}}$}}}
\end{picture}}
\put(36,92){\begin{picture}(28,16)
\put(14,14){\vector(0,-1){0}}
\put(0,16){\oval(28,16)[tr]}
\put(0,16){\oval(16,16)[l]}
\put(5,16){\makebox(0,0){\footnotesize{$(\mbox{-})^{\bullet_{g}}$}}}
\end{picture}}
\put(92,42){\begin{picture}(16,16)
\put(14,8){\vector(-1,0){0}}
\put(16,16){\oval(16,16)[br]}
\put(16,16){\oval(16,16)[t]}
\put(32,26){\makebox(0,0){\footnotesize{$(\mbox{-})^{\circ_{g}}$}}}
\end{picture}}
\put(36,-8){\begin{picture}(28,16)
\put(14,2){\vector(0,1){0}}
\put(28,0){\oval(16,16)[tr]}
\put(25,0){\oval(22,16)[b]}
\put(27,-1){\makebox(0,0){\footnotesize{$(\mbox{-})^{\circ_{g}}$}}}
\end{picture}}
\put(-50,25){\makebox(100,50){$A_0$}}
\put(0,25){\makebox(100,50){$\diamondsuit(g)$}}
\put(50,25){\makebox(100,50){$A_1$}}
\put(0,75){\makebox(100,50){$\mathsf{ker}(\check{g})$}}
\put(0,-25){\makebox(100,50){$\mathsf{ker}(\hat{g})$}}
\put(-5,50){\begin{picture}(50,50)
\put(3,21){\makebox(25,25){\footnotesize{$1_{A_0}$}}}
\put(20,9){\makebox(25,25){\footnotesize{$(\mbox{-})^{\bullet_{g}}$}}}
\put(16,3){\makebox(25,25){\tiny{$(\mathrm{iso})$}}}
\put(9,17){\vector(1,1){24}}
\put(39,35){\vector(-1,-1){24}}
\end{picture}}
\put(55,50){\begin{picture}(50,50)
\put(20,20){\makebox(25,25){\footnotesize{$\check{g}$}}}
\put(19,15){\makebox(25,25){\tiny{$(\mathrm{iso})$}}}
\put(8,8){\makebox(25,25){\footnotesize{$\hat{g}$}}}
\put(17,41){\vector(1,-1){24}}
\put(35,11){\vector(-1,1){24}}
\end{picture}}
\put(28,60){\makebox(25,25){\footnotesize{$\pi_0^{g}$}}}
\put(27,54){\makebox(25,25){\tiny{$(\mathrm{iso})$}}}
\put(48,60){\makebox(25,25){\footnotesize{$\xi_0^{g}$}}}
\put(47,54){\makebox(25,25){\tiny{$(\mathrm{iso})$}}}
\put(28.5,16){\makebox(25,25){\footnotesize{$\xi_1^{g}$}}}
\put(27.5,10){\makebox(25,25){\tiny{$(\mathrm{iso})$}}}
\put(48.5,16){\makebox(25,25){\footnotesize{$\pi_1^{g}$}}}
\put(47.5,10){\makebox(25,25){\tiny{$(\mathrm{iso})$}}}
\put(46,62){\vector(0,1){26}}
\put(54,88){\vector(0,-1){26}}
\put(46,12){\vector(0,1){26}}
\put(54,38){\vector(0,-1){26}}
\put(15,54){\vector(1,0){20}}
\put(35,46){\vector(-1,0){20}}
\put(65,54){\vector(1,0){20}}
\put(85,46){\vector(-1,0){20}}
\put(13,46.5){\makebox(25,25){\footnotesize{$\xi_0^{g}$}}}
\put(12,36){\makebox(25,25){\tiny{$(\mathrm{iso})$}}}
\put(13,28.5){\makebox(25,25){\footnotesize{$\pi_0^{g}$}}}
\put(63,46.5){\makebox(25,25){\footnotesize{$\pi_1^{g}$}}}
\put(62,39.5){\makebox(25,25){\tiny{$(\mathrm{iso})$}}}
\put(63,28.5){\makebox(25,25){\footnotesize{$\xi_1^{g}$}}}
\put(-5,0){\begin{picture}(50,50)
\put(20,19){\makebox(25,25){\footnotesize{$\check{g}$}}}
\put(8,7){\makebox(25,25){\footnotesize{$\hat{g}$}}}
\put(7,2){\makebox(25,25){\tiny{$(\mathrm{iso})$}}}
\put(17,41){\vector(1,-1){24}}
\put(35,11){\vector(-1,1){24}}
\end{picture}}
\put(55,0){\begin{picture}(50,50)
\put(6.5,19.5){\makebox(25,25){\footnotesize{$(\mbox{-})^{\circ_{g}}$}}}
\put(5.5,14.5){\makebox(25,25){\tiny{$(\mathrm{iso})$}}}
\put(21,8){\makebox(25,25){\footnotesize{$1_{A_1}$}}}
\put(9,17){\vector(1,1){24}}
\put(39,35){\vector(-1,-1){24}}
\end{picture}}
\end{picture}
\\ \\ \\
Adjunction (iconic) & Adjunction (detailed)
\\ \\
\fbox{$\begin{array}[b]{r@{\hspace{5pt}}c@{\hspace{5pt}}l}
g & \doteq & \langle \check{g} \dashv \hat{g} \rangle
: A_0 \rightleftharpoons A_1 \\ \\
\mathsf{ref}_{g} & \doteq & \langle \xi_0^{g} \dashv \pi_0^{g} \rangle : A_0 \rightleftharpoons \diamondsuit(g) \\
\mathsf{ref}_{g}^{\propto} & \doteq & \langle \pi_1^{g} \dashv \xi_1^{g} \rangle : \diamondsuit(g) \rightleftharpoons A_1 \\ \hline
\mathsf{equ}_{g} & \doteq & \langle \xi_0^{g} \dashv \pi_0^{g} \rangle : \mathsf{ker}(\check{g}) \rightleftharpoons \diamondsuit(g) \\
\mathsf{equ}_{g}^{\propto} & \doteq & \langle \pi_1^{g} \dashv \xi_1^{g} \rangle : \diamondsuit(g) \rightleftharpoons \mathsf{ker}(\hat{g}) \\ \hline
{\mathsf{lift}}^{0}_{g} & \doteq & \langle \check{g} \dashv \hat{g} \rangle : \mathsf{ker}(\check{g}) \rightleftharpoons A_1 \\
{\mathsf{lift}}^{1}_{g} & \doteq & \langle \check{g} \dashv \hat{g} \rangle : A_0 \rightleftharpoons \mathsf{ker}(\hat{g}) \\ \hline
\mathsf{clo}_{g} & \doteq & \langle 1_{A_0} \dashv (\mbox{-})^{\bullet_{g}} \rangle : A_0 \rightleftharpoons \mathsf{ker}(\check{g}) \\
\mathsf{int}_{g} & \doteq & \langle (\mbox{-})^{\bullet_{g}} \dashv 1_{A_1} \rangle : \mathsf{ker}(\hat{g}) \rightleftharpoons A_1
\end{array}$}
&
\fbox{$\begin{array}[b]{r@{\hspace{5pt}}c@{\hspace{5pt}}l}
\mathsf{clo}_{g} \circ \mathsf{equ}_{g} 
& = & \mathsf{ref}_{g} \\
\mathsf{equ}_{g} \circ \mathsf{ref}_{g}^{\propto} 
& = & {\mathsf{lift}}^{0}_{g} \\ \\
\mathsf{ref}_{g} \circ \mathsf{equ}_{g}^{\propto} 
& = & {\mathsf{lift}}^{1}_{g} \\
\mathsf{equ}_{g}^{\propto} \circ \mathsf{int}_{g} 
& = & \mathsf{ref}_{g}^{\propto} \\ \\
\mathsf{ref}_{g} \circ \mathsf{ref}_{g}^{\propto} & = & g \\
\mathsf{clo}_{g} \circ \mathsf{lift}^0_{g}        & = & g \\
\mathsf{lift}^1_{g} \circ \mathsf{int}_{g}        & = & g
\end{array}$}
\end{tabular}
\end{center}
\caption{The Diamond Diagram}
\label{diamond-diagram}
\end{figure}

\subsubsection{Uniqueness of Factorization\label{subsubsec:uniqueness:of:factorization:lot}}

\begin{lemma} [Diagonalization]
Assume that we are given a pseudo-commutative square
$e \circ s \equiv r \circ m$
\begin{center}
\setlength{\unitlength}{0.85pt}
\begin{picture}(50,50)
\put(-50,25){\makebox(100,50){$A_0$}}
\put(0,25){\makebox(100,50){$B$}}
\put(-50,-25){\makebox(100,50){$C$}}
\put(0,-25){\makebox(100,50){$A_1$}}
\put(0,30){\makebox(50,50){\footnotesize{$e$}}}
\put(0,-31){\makebox(50,50){\footnotesize{$m$}}}
\put(-30,0){\makebox(50,50){\footnotesize{$r$}}}
\put(31,0){\makebox(50,50){\footnotesize{$s$}}}
\put(-5,5){\makebox(50,50){\footnotesize{$d$}}}
\thicklines
\put(12,0){\vector(1,0){26}}
\put(12,50){\vector(1,0){26}}
\put(0,38){\vector(0,-1){26}}
\put(50,38){\vector(0,-1){26}}
\put(38,38){\vector(-1,-1){26}}
\end{picture}
\end{center}
of adjunctions,
with pseudo-reflection $e$ and pseudo-coreflection $m$.
Then there is an adjunction $d : B \rightleftharpoons C$,
unique up to equivalence,
with $e \circ d \equiv r$ and $d \circ m \equiv s$.
\end{lemma}

\noindent
{\bfseries Proof:} 
The necessary conditions
$\check{d}
\equiv \check{s} \cdot \hat{m}
\equiv \hat{e} \cdot \check{r}$
and
$\hat{d}
\equiv \hat{r} \cdot \check{e} 
\equiv \check{m} \cdot \hat{s}$
give the definitions.
Choose either option for left and right $\mathsf{C}$-morphism.
Existence follows from these definitions.

In more detail,
the fundamental adjointness property,
the special conditions for (co) reflections
and the above commutative diagram,
resolve into the following identities and inequalities:
$\hat{e} \cdot \check{e} \equiv 1_{B}$,
$1_{A_0} \leq \check{e} \cdot \hat{e}$,
$\hat{s} \cdot \check{s} \leq 1_{A_1}$,
$1_{B} \leq \check{s} \cdot \hat{s}$,
$\hat{r} \cdot \check{r} \leq 1_{C}$,
$1_{A_0} \leq \check{r} \cdot \hat{r}$,
$\hat{m} \cdot \check{m} \leq 1_{A_1}$,
$1_{C} \equiv \check{m} \cdot \hat{m}$,
$\check{e} \cdot \check{s} \equiv
\check{r} \cdot \check{m}$,
and
$\hat{m} \cdot \hat{r} \equiv
\hat{s} \cdot \hat{e}$.
By suitable pre- and post-composition we can prove the identities:
$\check{e} \cdot \check{s} \cdot \hat{m} 
\equiv \check{r}$,
$\check{m} \cdot \hat{s} \cdot \hat{e}
\equiv \hat{r}$,
$\hat{m} \cdot \hat{r} \cdot \check{e} 
\equiv \hat{s}$
and
$\hat{e} \cdot \check{r} \cdot \check{m} 
\equiv \check{s}$,
(and then)
$\check{s} \cdot \hat{m} 
\equiv \hat{e} \cdot \check{r}$ and
$\hat{r} \cdot \check{e} 
\equiv \check{m} \cdot \hat{s}$.

[{\bfseries Existence}]
Define the $\mathsf{C}$-morphisms
\underline{$\check{d}
\equiv 
\check{s} \cdot \hat{m} 
\equiv \hat{e} \cdot \check{r}$}
and
\underline{$\hat{d}
\equiv 
\hat{r} \cdot \check{e} 
\equiv \check{m} \cdot \hat{s}$}.
The properties
$\hat{d} \cdot \check{d}
\equiv \check{m} \cdot \hat{s} \cdot \hat{e} \cdot \check{r}
\equiv \hat{r} \cdot \check{r}
\leq 1_{C}$ and
$\check{d} \cdot \hat{d}
\equiv \check{s} \cdot \hat{m} \cdot \hat{r} \cdot \check{e}
\equiv \check{s} \cdot \hat{s}
\geq 1_{B}$ 
show that
$d = \langle \check{d}, \hat{d} \rangle : B \rightleftharpoons C$
is a $\mathsf{C}$-adjunction.
The properties
$\check{d} \cdot \check{m}
\equiv \hat{e} \cdot \check{r} \cdot \check{m}
\equiv \check{s}$
and
$\hat{m} \cdot \hat{d}
\equiv \hat{m} \cdot \hat{r} \cdot \check{e}
\equiv \hat{s}$
show that
$d$
satisfies the required identity
$d \circ m \equiv s$.
The properties
$\check{e} \cdot \check{d}
\equiv \check{e} \cdot \check{s} \cdot \hat{m}
\equiv \check{r}$
and
$\hat{d} \cdot \hat{e}
\equiv \check{m} \cdot \hat{s} \cdot \hat{e}
\equiv \hat{r}$
show that
$d$
satisfies the required identity
$e \circ d \equiv r$.

[{\bfseries Uniqueness}]
Suppose $b = \langle \check{b}, \hat{b} \rangle
: B \rightleftharpoons C$
is another $\mathsf{C}$-adjunction satisfying the require identities
$e \circ b \equiv r$
and $b \circ m \equiv s$.
These identities resolve to the identities
$\check{e} \cdot \check{b} \equiv \check{r}$,
$\hat{b} \cdot \hat{e} \equiv \hat{r}$,
$\check{b} \cdot \check{m} \equiv \check{s}$,
and
$\hat{m} \cdot \hat{b} \equiv \hat{s}$.
Hence,
$\check{b}
\equiv \check{b} \cdot \check{m} \cdot \hat{m}
\equiv \check{s} \cdot \hat{m}
\equiv \check{d}$,
$\hat{b}
\equiv \hat{b} \cdot \hat{e} \cdot \check{e}
\equiv \hat{r} \cdot \check{e}
\equiv \hat{d}$
and thus
$b \equiv d$.
\rule{2pt}{6pt}

\subsection{Morphisms\label{subsec:morphisms:lot}}

\subsubsection{Kernel Factorization\label{subsubsec:kernel:factorization}}

The \emph{left adjoint kernel adjunction}
\[\mathsf{ker}_{\check{\mathbf{f}}}
= 
\langle \check{\mathbf{f}}_0, (\mbox{-})^{\bullet_{\mathbf{h}}} {\cdot} \hat{\mathbf{f}}_0 \rangle
: \mathsf{ker}(\check{\mathbf{g}}) \rightleftharpoons \mathsf{ker}(\check{\mathbf{h}}).\]
The inequality
$\left( (\mbox{-})^{\bullet_{\mathbf{h}}} {\cdot} \hat{\mathbf{f}}_0 \right) 
\cdot \check{\mathbf{f}}_0
\cdot \check{\mathbf{h}}
= \check{\mathbf{h}} 
\cdot \hat{\mathbf{h}} \cdot \hat{\mathbf{f}}_0 
\cdot \check{\mathbf{f}}_0 \cdot \check{\mathbf{h}}
= \check{\mathbf{h}} 
\cdot \hat{\mathbf{f}}_1 \cdot \hat{\mathbf{g}} 
\cdot \check{\mathbf{g}} \cdot \check{\mathbf{f}}_1
\leq \check{\mathbf{h}} \cdot \hat{\mathbf{f}}_1 \cdot \check{\mathbf{f}}_1
\leq \check{\mathbf{h}}$
and the fact that 
$\check{\mathbf{h}} 
: \mathsf{ker}(\check{\mathbf{h}}) \rightarrow \mathbf{B}_1$ 
is a cartesian $\mathsf{C}$-morphism,
imply the inequality
$\left( (\mbox{-})^{\bullet_{\mathbf{h}}} {\cdot} \hat{\mathbf{f}}_0 \right) 
\cdot \check{\mathbf{f}}_0
\leq \mathrm{id}_{\mathsf{ker}(\check{\mathbf{h}})}$.
This together with the inequality
$\check{\mathbf{f}}_0 \cdot \left( (\mbox{-})^{\bullet_{\mathbf{h}}} {\cdot} \hat{\mathbf{f}}_0 \right)
\geq \check{\mathbf{f}}_0 \cdot \left( \mathrm{id}_{\mathsf{ker}(\check{\mathbf{h}})} {\cdot} \hat{\mathbf{f}}_0 
\right)
= \check{\mathbf{f}}_0 \cdot \hat{\mathbf{f}}_0 
\geq \mathrm{id}_{\mathsf{ker}(\check{\mathbf{g}})}$
prove the fundamental condition for the left adjoint kernel adjunction.
The \emph{right adjoint kernel adjunction}
\[\mathsf{ker}_{\hat{\mathbf{f}}}
=
\langle (\mbox{-})^{\circ_{\mathbf{g}}} {\cdot} \check{\mathbf{f}}_1, \hat{\mathbf{f}}_1 \rangle 
: \mathsf{ker}(\hat{\mathbf{g}}) \rightleftharpoons \mathsf{ker}(\hat{\mathbf{h}})\]
The inequality
$\left( (\mbox{-})^{\circ_{\mathbf{g}}} {\cdot} \check{\mathbf{f}}_1 \right) 
\cdot \hat{\mathbf{f}}_1
\cdot \hat{\mathbf{g}}
= \hat{\mathbf{g}} 
\cdot \check{\mathbf{g}} \cdot \check{\mathbf{f}}_1 
\cdot \hat{\mathbf{f}}_1 \cdot \hat{\mathbf{g}}
= \hat{\mathbf{g}} 
\cdot \check{\mathbf{f}}_0 \cdot \check{\mathbf{h}} 
\cdot \hat{\mathbf{h}} \cdot \hat{\mathbf{f}}_0
\geq \hat{\mathbf{g}} \cdot \check{\mathbf{f}}_0 \cdot \hat{\mathbf{f}}_0
\geq \hat{\mathbf{g}}$
and the fact that 
$\hat{\mathbf{g}} 
: \mathsf{ker}(\hat{\mathbf{g}}) \rightarrow \mathbf{A}_0$ 
is a cartesian $\mathsf{C}$-morphism,
imply the inequality
$\left( (\mbox{-})^{\circ_{\mathbf{g}}} {\cdot} \check{\mathbf{f}}_1 \right) 
\cdot \hat{\mathbf{f}}_1
\geq \mathrm{id}_{\mathsf{ker}(\hat{\mathbf{g}})}$.
This together with the inequality
$\hat{\mathbf{f}}_1 \cdot \left( (\mbox{-})^{\circ_{\mathbf{g}}} {\cdot} \check{\mathbf{f}}_1 \right)
\leq \hat{\mathbf{f}}_1 \cdot \left( \mathrm{id}_{\mathsf{ker}(\hat{\mathbf{g}})} {\cdot} \check{\mathbf{f}}_1 
\right)
= \hat{\mathbf{f}}_1 \cdot \check{\mathbf{f}}_1 
\leq \mathrm{id}_{\mathsf{ker}(\hat{\mathbf{h}})}$
prove the fundamental condition for the right adjoint kernel adjunction.

We want to prove the eight adjunction identities in Table~\ref{adjunction-identities}.
\begin{itemize}
\item
To prove the identity
$\mathsf{clo}_{\mathbf{g}} \circ \mathsf{ker}_{\check{\mathbf{f}}}
= \mathbf{f}_0 \circ \mathsf{clo}_{\mathbf{h}}$
note that
$\mathrm{id}_{A_0} \cdot \check{\mathbf{f}}_0 = \check{\mathbf{f}}_0 \cdot \mathrm{id}_{B_0}$
and
$\left( (\mbox{-})^{\bullet_{\mathbf{h}}} {\cdot} \hat{\mathbf{f}}_0 \right) 
\cdot (\mbox{-})^{\bullet_{\mathbf{g}}}
= (\mbox{-})^{\bullet_{\mathbf{h}}} {\cdot} \hat{\mathbf{f}}_0$.
\item
To prove the identity
$\mathsf{ker}_{\check{\mathbf{f}}} \circ {\mathsf{lift}}_{\mathbf{h}}^{\propto}
= {\mathsf{lift}}_{\mathbf{g}}^{\propto} \circ \mathbf{f}_1$
note that
$\check{\mathbf{f}}_0 \cdot \check{\mathbf{h}} = \check{\mathbf{g}} \cdot \check{\mathbf{f}}_1$
and
$\hat{\mathbf{h}}
\cdot \left( (\mbox{-})^{\bullet_{\mathbf{h}}} {\cdot} \hat{\mathbf{f}}_0 \right) 
= \hat{\mathbf{f}}_1 \cdot \hat{\mathbf{g}}$.
\item
To prove the identity
$\mathsf{lift}_{\mathbf{g}} \circ \mathsf{ker}_{\hat{\mathbf{f}}}
= \mathbf{f}_0 \circ \mathsf{lift}_{\mathbf{h}}$
note that
$\check{\mathbf{f}}_0 \cdot \check{\mathbf{h}}
= \check{\mathbf{g}} \cdot \left( (\mbox{-})^{\circ_{\mathbf{g}}} {\cdot} \check{\mathbf{f}}_1 \right)$
and
$\hat{\mathbf{h}} \cdot \hat{\mathbf{f}}_0
= \hat{\mathbf{f}}_1 \cdot \hat{\mathbf{g}}$.
\item
To prove the identity
$\mathsf{int}_{\mathbf{g}} \circ \mathbf{f}_1
= \mathsf{ker}_{\hat{\mathbf{f}}} \circ \mathsf{int}_{\mathbf{h}}$
note that
$\hat{\mathbf{f}}_1 \cdot \mathrm{id}_{A_1}
= \mathrm{id}_{B_1} \cdot \hat{\mathbf{f}}_1$
and
$(\mbox{-})^{\circ_{\mathbf{g}}} \cdot \check{\mathbf{f}}_1
= \left( (\mbox{-})^{\circ_{\mathbf{g}}} {\cdot} \check{\mathbf{f}}_1 \right)
\cdot (\mbox{-})^{\circ_{\mathbf{h}}}$.
\end{itemize}

\begin{table}
\begin{center}
\begin{tabular}{|c|c|} \cline{1-1}
\multicolumn{1}{|l|}{Given:} 
& \multicolumn{1}{l}{}
\\ \cline{1-1}
$\begin{array}[t]{r@{\hspace{10pt}=\hspace{10pt}}l}
\mathbf{g} \circ \mathbf{f}_1 
& \mathbf{f}_0 \circ \mathbf{h}
\end{array}$
& \multicolumn{1}{l}{}
\\ \cline{1-1}
\multicolumn{1}{l}{} & \multicolumn{1}{l}{} \\ \hline
\multicolumn{1}{|l|}{Inner:} & \multicolumn{1}{|l|}{Outer:} \\ \hline
$\begin{array}[t]{r@{\hspace{10pt}=\hspace{10pt}}l}
\mathsf{ref}_{\mathbf{g}} \circ \diamondsuit_{\mathbf{f}}
& \mathbf{f}_0 \circ \mathsf{ref}_{\mathbf{h}}
\\
\diamondsuit_{\mathbf{f}} \circ \mathsf{ref}_{\mathbf{h}}^{\propto}
& \mathsf{ref}_{\mathbf{g}}^{\propto} \circ \mathbf{f}_1
\\
{\mathsf{equ}}_{\mathbf{g}} \circ \diamondsuit_{\mathbf{f}}
& \mathsf{ker}_{\check{\mathbf{f}}} \circ {\mathsf{equ}}_{\mathbf{h}}
\\
\diamondsuit_{\mathbf{f}} \circ {\mathsf{equ}}_{\mathbf{h}}^{\propto}
& {\mathsf{equ}}_{\mathbf{g}}^{\propto} \circ \mathsf{ker}_{\hat{\mathbf{f}}}
\end{array}$
&
$\begin{array}[t]{r@{\hspace{10pt}=\hspace{10pt}}l}
\mathsf{clo}_{\mathbf{g}} \circ \mathsf{ker}_{\check{\mathbf{f}}}
& \mathbf{f}_0 \circ \mathsf{clo}_{\mathbf{h}}
\\
\mathsf{ker}_{\hat{\mathbf{f}}} \circ \mathsf{int}_{\mathbf{h}}
& \mathsf{int}_{\mathbf{g}} \circ \mathbf{f}_1
\\
\mathsf{ker}_{\check{\mathbf{f}}} \circ {\mathsf{lift}}_{\mathbf{h}}^{\propto}
& {\mathsf{lift}}_{\mathbf{g}}^{\propto} \circ \mathbf{f}_1
\\
\mathsf{lift}_{\mathbf{g}} \circ \mathsf{ker}_{\hat{\mathbf{f}}}
& \mathbf{f}_0 \circ \mathsf{lift}_{\mathbf{h}}
\end{array}$
\\ \hline
\end{tabular}
\end{center}
\caption{Adjunction Identities}
\label{adjunction-identities}
\end{table}

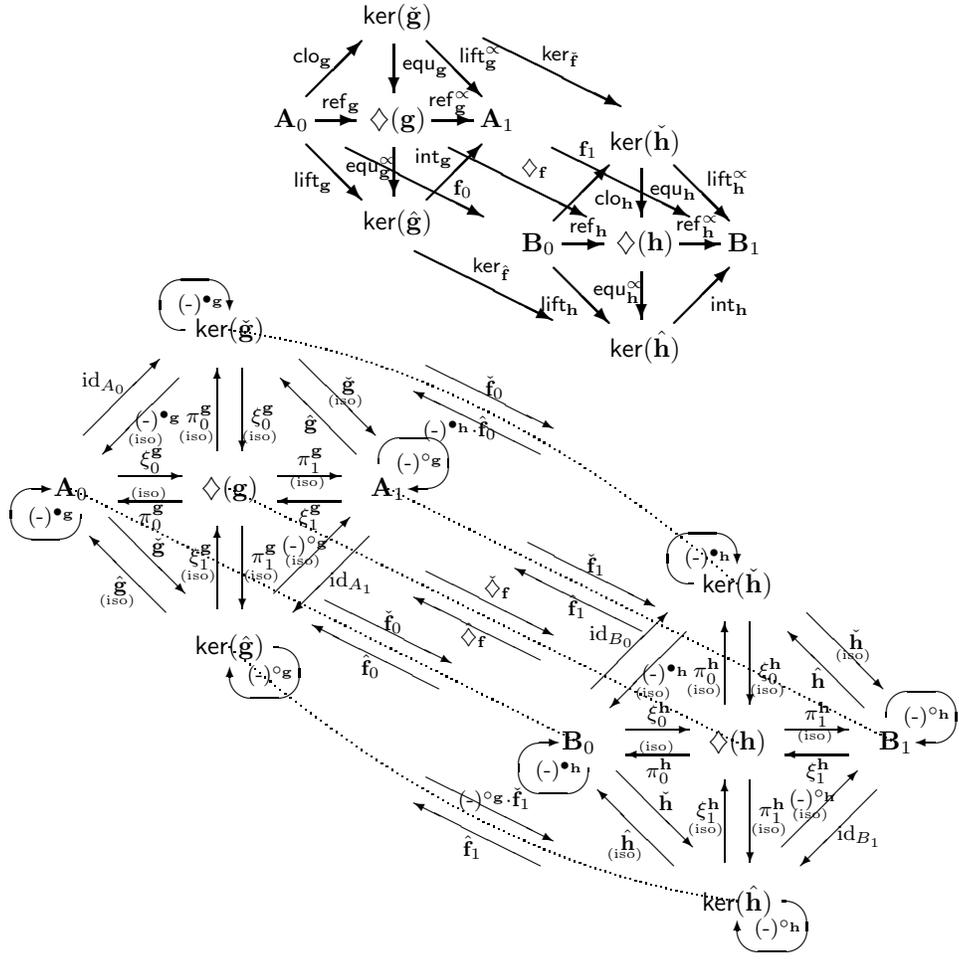
\begin{figure}
\begin{center}

\begin{tabular}{l}
\\ \\ \\ \\ \\
\setlength{\unitlength}{1.3pt}
\begin{picture}(160,130)(-50,-50)
\thicklines
\put(70,77){\begin{picture}(40,20)
\put(10,2.25){\makebox(25,25){\footnotesize{$\mathsf{ker}_{\check{\mathbf{f}}}$}}}
\put(0,20){\vector(2,-1){40}}
\end{picture}}
\put(90,45){\begin{picture}(40,20)
\put(-2,7){\makebox(25,25){\footnotesize{$\mathbf{f}_1$}}}
\put(0,20){\vector(2,-1){40}}
\end{picture}}
\put(60,45){\begin{picture}(40,20)
\put(12.5,1.5){\makebox(25,25){\footnotesize{$\diamondsuit_{\mathbf{f}}$}}}
\put(0,20){\vector(2,-1){40}}
\end{picture}}
\put(30,45){\begin{picture}(40,20)
\put(21.5,-4.5){\makebox(25,25){\footnotesize{$\mathbf{f}_0$}}}
\put(0,20){\vector(2,-1){40}}
\end{picture}}
\put(50,15){\begin{picture}(40,20)
\put(10,2.25){\makebox(25,25){\footnotesize{$\mathsf{ker}_{\hat{\mathbf{f}}}$}}}
\put(0,20){\vector(2,-1){40}}
\end{picture}}
\put(14,43){\begin{picture}(60,60)
\put(-30,15){\makebox(60,30){$\mathbf{A}_0$}}
\put(1,15){\makebox(60,30){$\diamondsuit(\mathbf{g})$}}
\put(30,15){\makebox(60,30){$\mathbf{A}_1$}}
\put(1,45){\makebox(60,30){$\mathsf{ker}(\check{\mathbf{g}})$}}
\put(1,-15){\makebox(60,30){$\mathsf{ker}(\hat{\mathbf{g}})$}}
\put(-18,33){\makebox(30,30)[r]{\footnotesize{$\mathsf{clo}_{\mathbf{g}}$}}}
\put(49,33){\makebox(30,30)[l]{\footnotesize{${\mathsf{lift}}_{\mathbf{g}}^{\propto}$}}}
\put(-18,-3){\makebox(30,30)[r]{\footnotesize{${\mathsf{lift}}_{\mathbf{g}}$}}}
\put(36,4){\makebox(30,30)[l]{\footnotesize{$\mathsf{int}_{\mathbf{g}}$}}}
\put(9,20){\makebox(30,30)[l]{\footnotesize{$\mathsf{ref}_{\mathbf{g}}$}}}
\put(40,20){\makebox(30,30)[l]{\footnotesize{$\mathsf{ref}_{\mathbf{g}}^{\propto}$}}}
\put(32.5,30){\makebox(30,30)[l]{\footnotesize{${\mathsf{equ}}_{\mathbf{g}}$}}}
\put(16,1){\makebox(30,30)[l]{\footnotesize{${\mathsf{equ}}_{\mathbf{g}}^{\propto}$}}}
\thicklines
\put(7,30){\vector(1,0){12}}
\put(41,30){\vector(1,0){12}}
\put(30,52){\vector(0,-1){14}}
\put(30,22){\vector(0,-1){14}}
\put(4.5,37){\vector(1,1){16}}
\put(39.5,53){\vector(1,-1){16}}
\put(4.5,23){\vector(1,-1){16}}
\put(39.5,7){\vector(1,1){16}}
\end{picture}}
\put(86,7){\begin{picture}(60,60)
\put(-30,15){\makebox(60,30){$\mathbf{B}_0$}}
\put(1,15){\makebox(60,30){$\diamondsuit(\mathbf{h})$}}
\put(30,15){\makebox(60,30){$\mathbf{B}_1$}}
\put(1,45){\makebox(60,30){$\mathsf{ker}(\check{\mathbf{h}})$}}
\put(1,-15){\makebox(60,30){$\mathsf{ker}(\hat{\mathbf{h}})$}}
\put(-2,28){\makebox(30,30)[r]{\footnotesize{$\mathsf{clo}_{\mathbf{h}}$}}}
\put(49,33){\makebox(30,30)[l]{\footnotesize{${\mathsf{lift}}_{\mathbf{h}}^{\propto}$}}}
\put(-18,-3){\makebox(30,30)[r]{\footnotesize{${\mathsf{lift}}_{\mathbf{h}}$}}}
\put(50,-3){\makebox(30,30)[l]{\footnotesize{$\mathsf{int}_{\mathbf{h}}$}}}
\put(9,20){\makebox(30,30)[l]{\footnotesize{$\mathsf{ref}_{\mathbf{h}}$}}}
\put(40,20){\makebox(30,30)[l]{\footnotesize{$\mathsf{ref}_{\mathbf{h}}^{\propto}$}}}
\put(32.5,30){\makebox(30,30)[l]{\footnotesize{${\mathsf{equ}}_{\mathbf{h}}$}}}
\put(16,1){\makebox(30,30)[l]{\footnotesize{${\mathsf{equ}}_{\mathbf{h}}^{\propto}$}}}
\thicklines
\put(7,30){\vector(1,0){12}}
\put(41,30){\vector(1,0){12}}
\put(30,52){\vector(0,-1){14}}
\put(30,22){\vector(0,-1){14}}
\put(4.5,37){\vector(1,1){16}}
\put(39.5,53){\vector(1,-1){16}}
\put(4.5,23){\vector(1,-1){16}}
\put(39.5,7){\vector(1,1){16}}
\end{picture}}
\end{picture}

\\

\setlength{\unitlength}{1.2pt}
\begin{picture}(260,130)(0,-20)

\put(0,80){\begin{picture}(100,100)
\put(-8,42){\begin{picture}(16,16)
\put(2,8){\vector(1,0){0}}
\put(0,0){\oval(22,16)[tl]}
\put(0,0){\oval(22,16)[b]}
\put(2,-1){\makebox(0,0){\footnotesize{$(\mbox{-})^{\bullet_{\mathbf{g}}}$}}}
\end{picture}}
\put(36,92){\begin{picture}(28,16)
\put(14,14){\vector(0,-1){0}}
\put(0,16){\oval(28,16)[tr]}
\put(0,16){\oval(16,16)[l]}
\put(5,16){\makebox(0,0){\footnotesize{$(\mbox{-})^{\bullet_{\mathbf{g}}}$}}}
\end{picture}}
\put(92,42){\begin{picture}(28,16)
\put(14,8){\vector(-1,0){0}}
\put(16,16){\oval(22,16)[br]}
\put(16,16){\oval(22,16)[t]}
\put(18,16){\makebox(0,0){\footnotesize{$(\mbox{-})^{\circ_{\mathbf{g}}}$}}}
\end{picture}}
\put(36,-8){\begin{picture}(28,16)
\put(14,2){\vector(0,1){0}}
\put(28,0){\oval(16,16)[tr]}
\put(25,0){\oval(22,16)[b]}
\put(27,-1){\makebox(0,0){\footnotesize{$(\mbox{-})^{\circ_{\mathbf{g}}}$}}}
\end{picture}}
\put(-50,25){\makebox(100,50){$\mathbf{A}_0$}}
\put(0,25){\makebox(100,50){$\diamondsuit(\mathbf{g})$}}
\put(50,25){\makebox(100,50){$\mathbf{A}_1$}}
\put(0,75){\makebox(100,50){$\mathsf{ker}(\check{\mathbf{g}})$}}
\put(0,-25){\makebox(100,50){$\mathsf{ker}(\hat{\mathbf{g}})$}}
\put(-5,50){\begin{picture}(50,50)
\put(3,21){\makebox(25,25){\footnotesize{$\mathrm{id}_{A_0}$}}}
\put(20,9){\makebox(25,25){\footnotesize{$(\mbox{-})^{\bullet_{\mathbf{g}}}$}}}
\put(16,3){\makebox(25,25){\tiny{$(\mathrm{iso})$}}}
\put(9,17){\vector(1,1){24}}
\put(39,35){\vector(-1,-1){24}}
\end{picture}}
\put(55,50){\begin{picture}(50,50)
\put(20,20){\makebox(25,25){\footnotesize{$\check{\mathbf{g}}$}}}
\put(19,15){\makebox(25,25){\tiny{$(\mathrm{iso})$}}}
\put(8,8){\makebox(25,25){\footnotesize{$\hat{\mathbf{g}}$}}}
\put(17,41){\vector(1,-1){24}}
\put(35,11){\vector(-1,1){24}}
\end{picture}}
\put(28,60){\makebox(25,25){\footnotesize{$\pi_0^{\mathbf{g}}$}}}
\put(27,54){\makebox(25,25){\tiny{$(\mathrm{iso})$}}}
\put(48,60){\makebox(25,25){\footnotesize{$\xi_0^{\mathbf{g}}$}}}
\put(47,54){\makebox(25,25){\tiny{$(\mathrm{iso})$}}}
\put(28.5,16){\makebox(25,25){\footnotesize{$\xi_1^{\mathbf{g}}$}}}
\put(27.5,10){\makebox(25,25){\tiny{$(\mathrm{iso})$}}}
\put(48.5,16){\makebox(25,25){\footnotesize{$\pi_1^{\mathbf{g}}$}}}
\put(47.5,10){\makebox(25,25){\tiny{$(\mathrm{iso})$}}}
\put(46,62){\vector(0,1){26}}
\put(54,88){\vector(0,-1){26}}
\put(46,12){\vector(0,1){26}}
\put(54,38){\vector(0,-1){26}}
\put(15,54){\vector(1,0){20}}
\put(35,46){\vector(-1,0){20}}
\put(65,54){\vector(1,0){20}}
\put(85,46){\vector(-1,0){20}}
\put(13,46.5){\makebox(25,25){\footnotesize{$\xi_0^{\mathbf{g}}$}}}
\put(12,36){\makebox(25,25){\tiny{$(\mathrm{iso})$}}}
\put(13,28.5){\makebox(25,25){\footnotesize{$\pi_0^{\mathbf{g}}$}}}
\put(63,46.5){\makebox(25,25){\footnotesize{$\pi_1^{\mathbf{g}}$}}}
\put(62,39.5){\makebox(25,25){\tiny{$(\mathrm{iso})$}}}
\put(63,28.5){\makebox(25,25){\footnotesize{$\xi_1^{\mathbf{g}}$}}}
\put(-5,0){\begin{picture}(50,50)
\put(20,19){\makebox(25,25){\footnotesize{$\check{\mathbf{g}}$}}}
\put(8,7){\makebox(25,25){\footnotesize{$\hat{\mathbf{g}}$}}}
\put(7,2){\makebox(25,25){\tiny{$(\mathrm{iso})$}}}
\put(17,41){\vector(1,-1){24}}
\put(35,11){\vector(-1,1){24}}
\end{picture}}
\put(55,0){\begin{picture}(50,50)
\put(6.5,19.5){\makebox(25,25){\footnotesize{$(\mbox{-})^{\circ_{\mathbf{g}}}$}}}
\put(5.5,14.5){\makebox(25,25){\tiny{$(\mathrm{iso})$}}}
\put(21,8){\makebox(25,25){\footnotesize{$\mathrm{id}_{A_1}$}}}
\put(9,17){\vector(1,1){24}}
\put(39,35){\vector(-1,-1){24}}
\end{picture}}
\end{picture}}

\put(160,0){\begin{picture}(100,100)
\put(-8,42){\begin{picture}(16,16)
\put(2,8){\vector(1,0){0}}
\put(0,0){\oval(22,16)[tl]}
\put(0,0){\oval(22,16)[b]}
\put(2,-1){\makebox(0,0){\footnotesize{$(\mbox{-})^{\bullet_{\mathbf{h}}}$}}}
\end{picture}}
\put(36,92){\begin{picture}(28,16)
\put(14,14){\vector(0,-1){0}}
\put(0,16){\oval(28,16)[tr]}
\put(0,16){\oval(16,16)[l]}
\put(5,16){\makebox(0,0){\footnotesize{$(\mbox{-})^{\bullet_{\mathbf{h}}}$}}}
\end{picture}}
\put(92,42){\begin{picture}(28,16)
\put(14,8){\vector(-1,0){0}}
\put(16,16){\oval(22,16)[br]}
\put(16,16){\oval(22,16)[t]}
\put(18,16){\makebox(0,0){\footnotesize{$(\mbox{-})^{\circ_{\mathbf{h}}}$}}}
\end{picture}}
\put(36,-8){\begin{picture}(28,16)
\put(14,2){\vector(0,1){0}}
\put(28,0){\oval(16,16)[tr]}
\put(25,0){\oval(22,16)[b]}
\put(27,-1){\makebox(0,0){\footnotesize{$(\mbox{-})^{\circ_{\mathbf{h}}}$}}}
\end{picture}}
\put(-50,25){\makebox(100,50){$\mathbf{B}_0$}}
\put(0,25){\makebox(100,50){$\diamondsuit(\mathbf{h})$}}
\put(50,25){\makebox(100,50){$\mathbf{B}_1$}}
\put(0,75){\makebox(100,50){$\mathsf{ker}(\check{\mathbf{h}})$}}
\put(0,-25){\makebox(100,50){$\mathsf{ker}(\hat{\mathbf{h}})$}}
\put(-5,50){\begin{picture}(50,50)
\put(3,21){\makebox(25,25){\footnotesize{$\mathrm{id}_{B_0}$}}}
\put(20,9){\makebox(25,25){\footnotesize{$(\mbox{-})^{\bullet_{\mathbf{h}}}$}}}
\put(16,3){\makebox(25,25){\tiny{$(\mathrm{iso})$}}}
\put(9,17){\vector(1,1){24}}
\put(39,35){\vector(-1,-1){24}}
\end{picture}}
\put(55,50){\begin{picture}(50,50)
\put(20,20){\makebox(25,25){\footnotesize{$\check{\mathbf{h}}$}}}
\put(19,15){\makebox(25,25){\tiny{$(\mathrm{iso})$}}}
\put(8,8){\makebox(25,25){\footnotesize{$\hat{\mathbf{h}}$}}}
\put(17,41){\vector(1,-1){24}}
\put(35,11){\vector(-1,1){24}}
\end{picture}}
\put(28,60){\makebox(25,25){\footnotesize{$\pi_0^{\mathbf{h}}$}}}
\put(27,54){\makebox(25,25){\tiny{$(\mathrm{iso})$}}}
\put(48,60){\makebox(25,25){\footnotesize{$\xi_0^{\mathbf{h}}$}}}
\put(47,54){\makebox(25,25){\tiny{$(\mathrm{iso})$}}}
\put(28.5,16){\makebox(25,25){\footnotesize{$\xi_1^{\mathbf{h}}$}}}
\put(27.5,10){\makebox(25,25){\tiny{$(\mathrm{iso})$}}}
\put(48.5,16){\makebox(25,25){\footnotesize{$\pi_1^{\mathbf{h}}$}}}
\put(47.5,10){\makebox(25,25){\tiny{$(\mathrm{iso})$}}}
\put(46,62){\vector(0,1){26}}
\put(54,88){\vector(0,-1){26}}
\put(46,12){\vector(0,1){26}}
\put(54,38){\vector(0,-1){26}}
\put(15,54){\vector(1,0){20}}
\put(35,46){\vector(-1,0){20}}
\put(65,54){\vector(1,0){20}}
\put(85,46){\vector(-1,0){20}}
\put(13,46.5){\makebox(25,25){\footnotesize{$\xi_0^{\mathbf{h}}$}}}
\put(12,36){\makebox(25,25){\tiny{$(\mathrm{iso})$}}}
\put(13,28.5){\makebox(25,25){\footnotesize{$\pi_0^{\mathbf{h}}$}}}
\put(63,46.5){\makebox(25,25){\footnotesize{$\pi_1^{\mathbf{h}}$}}}
\put(62,39.5){\makebox(25,25){\tiny{$(\mathrm{iso})$}}}
\put(63,28.5){\makebox(25,25){\footnotesize{$\xi_1^{\mathbf{h}}$}}}
\put(-5,0){\begin{picture}(50,50)
\put(20,19){\makebox(25,25){\footnotesize{$\check{\mathbf{h}}$}}}
\put(8,7){\makebox(25,25){\footnotesize{$\hat{\mathbf{h}}$}}}
\put(7,2){\makebox(25,25){\tiny{$(\mathrm{iso})$}}}
\put(17,41){\vector(1,-1){24}}
\put(35,11){\vector(-1,1){24}}
\end{picture}}
\put(55,0){\begin{picture}(50,50)
\put(6.5,19.5){\makebox(25,25){\footnotesize{$(\mbox{-})^{\circ_{\mathbf{h}}}$}}}
\put(5.5,14.5){\makebox(25,25){\tiny{$(\mathrm{iso})$}}}
\put(21,8){\makebox(25,25){\footnotesize{$\mathrm{id}_{B_1}$}}}
\put(9,17){\vector(1,1){24}}
\put(39,35){\vector(-1,-1){24}}
\end{picture}}
\end{picture}}

\qbezier[100](50,180)(130,170)(210,100)
\qbezier[100](50,130)(130,90)(210,50)
\qbezier[100](50,80)(130,10)(210,0)
\qbezier[100](0,130)(80,90)(160,50)
\qbezier[100](100,130)(180,90)(260,50)

\put(110,140){\begin{picture}(40,30)
\put(2,29){\vector(2,-1){40}}
\put(38,1){\vector(-2,1){40}}
\put(11,9.5){\makebox(25,25){\footnotesize{$\check{\mathbf{f}}_0$}}}
\put(0,-3.5){\makebox(25,25){\footnotesize{$(\mbox{-})^{\bullet_{\mathbf{h}}} 
{\cdot} \hat{\mathbf{f}}_0$}}}
\end{picture}}

\put(142,84){\begin{picture}(40,30)
\put(2,29){\vector(2,-1){40}}
\put(38,1){\vector(-2,1){40}}
\put(11,9.5){\makebox(25,25){\footnotesize{$\check{\mathbf{f}}_1$}}}
\put(5,-3.5){\makebox(25,25){\footnotesize{$\hat{\mathbf{f}}_1$}}}
\end{picture}}

\put(110,75){\begin{picture}(40,30)
\put(2,29){\vector(2,-1){40}}
\put(38,1){\vector(-2,1){40}}
\put(12,11.5){\makebox(25,25){\footnotesize{$\check{\diamondsuit}_{\mathbf{f}}$}}}
\put(4.5,-3.5){\makebox(25,25){\footnotesize{$\hat{\diamondsuit}_{\mathbf{f}}$}}}
\end{picture}}

\put(78,66){\begin{picture}(40,30)
\put(2,29){\vector(2,-1){40}}
\put(38,1){\vector(-2,1){40}}
\put(11,9.5){\makebox(25,25){\footnotesize{$\check{\mathbf{f}}_0$}}}
\put(4,-4.5){\makebox(25,25){\footnotesize{$\hat{\mathbf{f}}_0$}}}
\end{picture}}

\put(110,10){\begin{picture}(40,30)
\put(2,29){\vector(2,-1){40}}
\put(38,1){\vector(-2,1){40}}
\put(11,9.5){\makebox(25,25){\footnotesize{$(\mbox{-})^{\circ_{\mathbf{g}}} {\cdot} \check{\mathbf{f}}_1$}}}
\put(4,-5.5){\makebox(25,25){\footnotesize{$\hat{\mathbf{f}}_1$}}}
\end{picture}}

\end{picture}

\\

\end{tabular}
\end{center}
\caption{The Diamond Diagram}
\label{diamond-diagram-morphisms}
\end{figure}

\paragraph{Typing.}

\begin{center}
\begin{tabular}{c@{\hspace{25pt}}c}
\fbox{$\begin{array}[b]{r@{\hspace{5pt}}c@{\hspace{5pt}}l}
\mathbf{g} & \doteq & \langle \check{\mathbf{g}} \dashv \hat{\mathbf{g}} \rangle
: \mathbf{A}_0 \rightleftharpoons \mathbf{A}_1 \\
\\
\mathsf{ref}_{\mathbf{g}} & \doteq & \langle \xi_0^{\mathbf{g}} \dashv \pi_0^{\mathbf{g}} \rangle : \mathbf{A}_0 \rightleftharpoons \diamondsuit(\mathbf{g}) \\
\mathsf{ref}_{\mathbf{g}}^{\propto} & \doteq & \langle \pi_1^{\mathbf{g}} \dashv \xi_1^{\mathbf{g}} \rangle : \diamondsuit(\mathbf{g}) \rightleftharpoons \mathbf{A}_1 \\
\hline
\mathsf{equ}_{\mathbf{g}} & \doteq & \langle \xi_0^{\mathbf{g}} \dashv \pi_0^{\mathbf{g}} \rangle : \mathsf{ker}(\check{\mathbf{g}}) \rightleftharpoons \diamondsuit(\mathbf{g}) \\
\mathsf{equ}_{\mathbf{g}}^{\propto} & \doteq & \langle \pi_1^{\mathbf{g}} \dashv \xi_1^{\mathbf{g}} \rangle : \diamondsuit(\mathbf{g}) \rightleftharpoons \mathsf{ker}(\hat{\mathbf{g}}) \\
\hline
{\mathsf{lift}}^{\propto}_{\mathbf{g}} & \doteq & \langle \check{\mathbf{g}} \dashv \hat{\mathbf{g}} \rangle : \mathsf{ker}(\check{\mathbf{g}}) \rightleftharpoons \mathbf{A}_1 \\
{\mathsf{lift}}_{\mathbf{g}} & \doteq & \langle \check{\mathbf{g}} \dashv \hat{\mathbf{g}} \rangle : \mathbf{A}_0 \rightleftharpoons \mathsf{ker}(\hat{\mathbf{g}}) \\
\hline
\mathsf{clo}_{\mathbf{g}} & \doteq & \langle \mathrm{id}_{A} \dashv {\mathbf{g}}^{\bullet} \rangle : \mathbf{A}_0 \rightleftharpoons \mathsf{ker}(\check{\mathbf{g}}) \\
\mathsf{int}_{\mathbf{g}} & \doteq & \langle {\mathbf{g}}^{\circ} \dashv \mathrm{id}_{B} \rangle : \mathsf{ker}(\hat{\mathbf{g}}) \rightleftharpoons \mathbf{A}_1
\end{array}$}
&
\fbox{$\begin{array}[b]{r@{\hspace{5pt}}c@{\hspace{5pt}}l}
\mathsf{clo}_{\mathbf{g}} \circ \mathsf{equ}_{\mathbf{g}} & = & \mathsf{ref}_{\mathbf{g}} \\
\mathsf{equ}_{\mathbf{g}} \circ \mathsf{ref}_{\mathbf{g}}^{\propto} & = & {\mathsf{lift}}^{\propto}_{\mathbf{g}} \\
\\
\mathsf{ref}_{\mathbf{g}} \circ \mathsf{equ}_{\mathbf{g}}^{\propto} & = & {\mathsf{lift}}_{\mathbf{g}} \\
\mathsf{equ}_{\mathbf{g}}^{\propto} \circ \mathsf{int}_{\mathbf{g}} & = & \mathsf{ref}_{\mathbf{g}}^{\propto} \\
\\
\mathsf{ref}_{\mathbf{g}} \circ \mathsf{ref}_{\mathbf{g}}^{\propto} & = & \mathbf{g} \\
\mathsf{clo}_{\mathbf{g}} \circ \mathsf{lift}^{\propto}_{\mathbf{g}} & = & \mathbf{g} \\
\mathsf{lift}_{\mathbf{g}} \circ \mathsf{int}_{\mathbf{g}} & = & \mathbf{g}
\end{array}$}
\end{tabular}
\end{center}


\section{Outline}

A previous paper discussed the factorization of truth in the specific 
context of order adjunctions in a topos $\mathcal{B}$ and the underlying fibration
${|\mbox{-}|}_{\mathcal{B}} : \mathsf{Ord}(\mathcal{B}) \rightarrow \mathcal{B}$.
This paper extends that discussion to a abstract theory of truth and its factorization. 
More specifically, it lays down axioms necessary for the factorization of truth, by categorically characterizing $\mathsf{Ord}(\mathcal{B})$ and the underlying fibration.

\begin{enumerate}
\item
Define a pseudo-factoriation system in an order-enriched category.
\begin{itemize}
\item include strict factorization within pseudo-factorization
\item discuss factorization duality; that is, how involutions and factorization systems interact
\end{itemize}
\item Define an (abstract) conceptual structures (CS) category.
\begin{itemize}
\item define polar or CS factorization here
\end{itemize}
\item Define an (abstract) lattice of theories (LOT) category.
\begin{itemize}
\item define cartesian or LOT factorization here
\item LOT factorization centers on a ``diamond diagram''
\end{itemize}
\end{enumerate}

\cite{kent:dcs}

\end{document}